\documentclass[10pt]{article}

\usepackage{hyperref}
\usepackage{amsmath,amsthm}
\usepackage{fullpage}
\usepackage{amsfonts}
\usepackage{changepage}
\usepackage{bookmark}
\usepackage{graphicx}
\usepackage[absolute]{textpos}
\usepackage{paralist}
\usepackage{enumitem}
\usepackage{wrapfig}

\usepackage[ruled,vlined,linesnumbered]{algorithm2e}

\usepackage[capitalize]{cleveref}
\usepackage{tikz}
\usetikzlibrary{shapes,snakes}
\usetikzlibrary{calc}
\usepackage{ifthen}

\def\boxit#1{\vbox{\hrule\hbox{\vrule\kern4pt
  \vbox{\kern1pt#1\kern1pt}
\kern2pt\vrule}\hrule}}
\usetikzlibrary{arrows}
\usetikzlibrary{decorations.markings}

\pgfdeclarelayer{bg}    %
\pgfdeclarelayer{fg}    %
\pgfsetlayers{bg,main,fg}  %

\newtheorem{lemma}{Lemma}
\newtheorem{theorem}{Theorem}
\newtheorem{corollary}{Corollary}
\newtheorem{observation}{Observation}
\newtheorem{reduction}{Reduction Rule}

\newcommand{\proofparagraph}[1]{\smallskip\emph{#1}}
\newcommand{\proofparagraphf}[1]{\emph{#1}}

\newtheorem*{thmhalfintegralxp}{\cref{thm:half-integral-xp}}
\newtheorem*{thmhalfintegralpolynoexcess}{\cref{thm:half-integral-poly-no-excess}}

\newcommand{\defprobMD}[3]{\par
 \vspace{3mm}
\noindent\fbox{
 \begin{minipage}{0.96\textwidth}
 \begin{tabular*}{\textwidth}{@{\extracolsep{\fill}}lr} #1  \vspace{1mm} \\ \end{tabular*}
 {\textbf{Input:}} #2%
	\vspace{1mm}\\%
 {\textbf{Question:}} #3%
 \end{minipage}
 }
 \vspace{3mm}
\par
}

\newcommand{\pCElong}{\textsc{Cluster Editing}}
\newcommand{\pCE}{\pCElong}
\newcommand{\pCEAlong}{\textsc{Cluster Editing above Modification-Disjoint $P_{3}$ Packing}}
\newcommand{\pCEA}{\textsc{CEaMP}}
\newcommand{\pCEAT}{\textsc{CEaHMP}}
\newcommand{\pCEATlong}{\textsc{Cluster Editing above Half-Integral Modification-Disjoint $P_{3}$ Packing}}
\newcommand{\pCDA}{\textsc{CDaMP}}
\newcommand{\pCEMT}{\textsc{CEMHMP}}
\newcommand{\pCEMTlong}{\textsc{Cluster Editing Matching Half-Integral Modification-Disjoint $P_{3}$ Packing}}

\newcommand{\pkg}{\ensuremath{\mathcal{H}}}

\newcommand{\true}{\ensuremath{\textsf{true}}}
\newcommand{\false}{\ensuremath{\textsf{false}}}

\renewcommand{\P}{\ensuremath{\mathcal{P}}}

\newcommand{\cl}{\ensuremath{d}} %
\newcommand{\vp}{\ensuremath{a}} %
\newcommand{\vq}{\ensuremath{b}} %
\newcommand{\vr}{\ensuremath{c}} %

\begin{document}

\title{Cluster Editing parameterized above\\modification-disjoint \(P_3\)-packings\footnote{An extended abstract of this work appeared at \emph{STACS 2021}~\cite{li_cluster_2021} and a full version appeared in \emph{ACM Transactions on Algorithms}~\cite{li_cluster_2023}.
}}

\author{Shaohua Li\thanks{Institute of Informatics, University of Warsaw, Poland, \texttt{Shaohua.Li@mimuw.edu.pl}.}
\and Marcin Pilipczuk\thanks{Institute of Informatics, University of Warsaw, Poland, \texttt{malcin@mimuw.edu.pl}.}
\and Manuel Sorge\thanks{Institute of Informatics, University of Warsaw, Poland and TU Wien, Vienna, Austria, \texttt{manuel.sorge@ac.tuwien.ac.at}.}}

\maketitle

\begin{abstract}
Given a graph $G=(V,E)$ and an integer $k$, the \textsc{Cluster Editing} problem asks whether we can transform~$G$ into a union of vertex-disjoint cliques by at most $k$ modifications (edge deletions or insertions).
In this paper, we study the following variant of \textsc{Cluster Editing}.
We are given a graph $G=(V,E)$, a packing~$\mathcal{H}$ of modification-disjoint induced $P_3$s
(no pair of $P_3$s in $\mathcal{H}$ share an edge or non-edge) and an integer~$\ell$. The task is to decide whether $G$ can be transformed into a union of vertex-disjoint cliques by at most $\ell+|\mathcal{H}|$ modifications (edge deletions or insertions).
We show that this problem is NP-hard even when $\ell=0$
(in which case the problem asks to turn $G$ into a disjoint union of cliques by performing
 exactly one edge deletion or insertion per element of $\mathcal{H}$) and when each vertex is in at most 23 \(P_3\)s of the packing.
This answers negatively a question of van Bevern, Froese, and Komusiewicz (CSR 2016, ToCS 2018),
repeated by C.\ Komusiewicz at Shonan meeting no.\ 144 in March 2019.
We then initiate the study to find the largest integer \(c\) such that the problem remains tractable when restricting to packings such that each vertex is in at most \(c\) packed \(P_3\)s. Here packed $P_3$s are those belonging to the packing~$\mathcal{H}$.  
Van Bevern et al.\ showed that the case \(c = 1\) is fixed-parameter tractable with respect to \(\ell\) and we show that the case \(c = 2\) is solvable in \(|V|^{2\ell + O(1)}\)~time.

\end{abstract}

\tikzset{middlearrow/.style={
		decoration={markings,
			mark= at position 0.7 with {\arrow{#1}} ,
		},
		postaction={decorate}
	}
}
\tikzset{none/.style={
	}
}

\tikzset{nodestyle1/.style={
		fill,circle,scale=0.7
	}
}
\tikzset{nodestyle2/.style={
		fill,circle,scale=0.5
	}
}

\tikzset{b/.style={
		draw=black,fill=white,circle,scale=3.5
	}
}

\tikzset{bs/.style={
		draw=black,fill=white,circle,scale=3
	}
}
\tikzset{gn/.style={
	}
}

\tikzset{vertex/.style={
		fill,circle,scale=0.3
	}
}
\tikzset{largeVertex/.style={
		scale=0.5
	}
}

\tikzset{normal/.style={
		black,thick
	}
}

\tikzset{thickEdge/.style={
		gray,line width=3.5mm
	}
}
\tikzset{thickGreen/.style={
		brown,line width=2mm
	}
}

\tikzset{varGadget/.style={
		fill,circle,scale=0.5
	}
}
\tikzset{trasEdge1/.style={
		red, very thick
	}
}
\tikzset{trasEdge2/.style={
		black, very thick
	}
}
\tikzset{trasEdge3/.style={
		green, very thick
	}
}
\tikzset{trasEdge4/.style={
		blue, very thick
	}
}
\tikzset{varGadEdge/.style={
		gray, thick
	}
}

\tikzset{newP3/.style={
		brown, very thick
	}
}

\tikzset{newP3Dot/.style={
		brown, dotted, very thick
	}
}

\tikzset{C8Edge1/.style={
		black, thick
	}
}

\tikzset{type1/.style={
		red, thick
	}
}

\tikzset{type2/.style={
		red, thick, dotted
	}
}

\section{Introduction}\label{sec:intro}

\textsc{Correlation Clustering} is a well-known problem motivated by research in computational biology \cite{DBLP:journals/jcb/Ben-DorSY99} and machine learning \cite{DBLP:journals/ml/BansalBC04}. In this problem we aim to partition data points into groups or clusters according to their pairwise similarity and this has been intensively studied in the literature, see \cite{DBLP:journals/jacm/AilonCN08,DBLP:conf/stoc/AlonMMN05,DBLP:journals/eccc/ECCC-TR05-058,DBLP:journals/ml/BansalBC04,DBLP:journals/jcb/Ben-DorSY99,
DBLP:journals/jcss/CharikarGW05}, for example.

\looseness=-1
In this paper, we study \textsc{Correlation Clustering} from a graph-based point of view, resulting in the following problem formulation.
A graph $H$ is called a \emph{cluster graph} if $H$ is a union of vertex-disjoint cliques; we also call these cliques \emph{clusters}.
Given a graph $G=(V,E)$, in the optimization version of \pCElong{} we ask for a minimum-size \emph{cluster-editing set}~$S$, that is,
a set \(S \subseteq \binom{V}{2}\) of vertex pairs such that $G\triangle S := (V,E\triangle S)$ is a cluster graph.
Here $E\triangle S$ is the symmetric difference of $E$ and $S$, that is, $E\triangle S = (E \setminus S) \cup (S \setminus E)$.
We also sometimes refer to vertex pairs as \emph{edits}.
\pCElong{} is NP-hard~\cite{DBLP:journals/dam/ShamirST04}.
Constant-ratio approximation algorithms have been found for the optimization variant \cite{DBLP:journals/jacm/AilonCN08,DBLP:journals/ml/BansalBC04,DBLP:journals/jcss/CharikarGW05} but it is also APX-hard \cite{DBLP:journals/jcss/CharikarGW05}.
We focus here on exact algorithms and the decision version of \pCE.

Given a natural number $k$ and a graph $G=(V,E)$, the decision version of \pCE\ asks whether there exists a cluster-editing set $S$ such that $|S|\leq k$.
Exact parameterized algorithms for \pCE\ and some of its variants have been extensively studied~\cite{DBLP:journals/mst/GrammGHN05,bocker_going_2009,DBLP:journals/mst/ProttiSS09,DBLP:journals/mst/Damaschke10,DBLP:journals/tcs/BodlaenderFHMPR10,DBLP:journals/siamdm/GuoKNU10,DBLP:journals/algorithmica/BockerBK11,DBLP:journals/ipl/BockerD11, DBLP:journals/disopt/FellowsGKNU11, DBLP:journals/algorithmica/GuoKKU11,DBLP:journals/jda/Bocker12,komusiewicz_cluster_2012,fomin_tight_2014,marx_fixedparameter_2014,bousquet_multicut_2018,Abu-KhzamBFS21,askeland_overlapping_2022,ArrighiBDSW23,FDHSSVW23}.
\pCE\ is but one of a large group of edge modification problems that have been studied, see Crespelle et al.~\cite{crespelle_survey_2020} for a recent survey.
Perhaps it is one of the most important such problems because of the practical motivation.
Barring few exceptions~\cite{DBLP:journals/mst/Damaschke10,komusiewicz_cluster_2012,fomin_tight_2014,DBLP:journals/mst/BevernFK18}, \pCE\ has mainly been studied with respect to the solution-size parameter~\(k\).
It is not hard to observe that \pCE\ is fixed-parameter tractable with respect to \(k\) and a series of papers~\cite{DBLP:journals/mst/GrammGHN05,gramm_automated_2004, bocker_going_2009,DBLP:journals/ipl/BockerD11,DBLP:journals/jda/Bocker12} continually improved the base in the exponential part of the running time, culminating in the current fastest fixed-parameter algorithm with running time $O(1.62^{k}+n+m)$~\cite{DBLP:journals/jda/Bocker12}, where \(n\) is the number of vertices of the input graph and \(m\) its number of edges.
Similarly, a series of papers~\cite{DBLP:journals/mst/GrammGHN05,Fellows06,FellowsLRS07, guo_more_2009,DBLP:journals/mst/ProttiSS09,DBLP:journals/algorithmica/CaoC12,DBLP:journals/jcss/ChenM12} gave more and more effective problem kernels\footnote{A problem kernel is a formalization of provably effective and efficient data reduction. It is a polynomial-time self-reduction that produces instances of size bounded by some function of the parameter.} until a problem kernel with $2k$ vertices was achieved~\cite{DBLP:journals/algorithmica/CaoC12,DBLP:journals/jcss/ChenM12}.

As mentioned, the interest in \pCE\ is not merely theoretical.
Indeed, parameterized techniques are implemented in standard clustering tools~\cite{morris_clustermaker_2011,wittkop_partitioning_2010}.
Although practitioners report that in particular the parameterized data-reduction techniques are effective~\cite{DBLP:journals/algorithmica/BockerBK11,bocker_cluster_2013}, the parameter \(k\) is not very small in several real-world data sets~\cite{bocker_going_2009,DBLP:journals/algorithmica/BockerBK11,DBLP:journals/mst/BevernFK18}.
For instance, B\"{o}cker et al.~\cite[Table~2]{bocker_going_2009} considered 26 graphs derived from biological data with 91 to 100 vertices on which the average number of needed edits is 315, despite noting that the \pCE\ model outperformed other clustering models.

A technique to deal with such large parameters is \emph{parameterization above lower bounds}.
Herein, the parameter is of the form $\ell=k-h$ where $h$ is a lower bound on the solution size, usually computable in polynomial time, and $\ell$ is the \emph{excess} of the solution size above the lower bound.
Research into parameterizations above lower bounds has been active and fruitful for several parameterized problems, not only on the theory-side (see \cite{DBLP:journals/jal/MahajanR99,DBLP:journals/toct/CyganPPW13,DBLP:conf/soda/GargP16,DBLP:journals/talg/LokshtanovNRRS14,kratsch_randomized_2018}, for example)
but also in practice, as algorithms based on that approach yielded quite efficient implementations for \textsc{Vertex Cover}~\cite{akiba_branchandreduce_2016} and among the most efficient ones for \textsc{Feedback Vertex Set}~\cite{iwata_lineartime_2017,kiljan_experimental_2018}.
For \pCE\ we are aware of only one research work considering parameterizations above lower bounds:
Van Bevern, Froese, and Komusiewicz~\cite{DBLP:journals/mst/BevernFK18} studied edge-modification problems parameterized above the lower bound from packings of forbidden induced subgraphs and showed that \pCE\ parameterized by the excess above the size of a given packing of \emph{vertex-disjoint} $P_3$s is fixed-parameter tractable.
Observe that a graph is a cluster graph if and only if it does not contain any~$P_3$, a path on three vertices, as an induced subgraph.
Consequently, one needs to perform at least one edge deletion or insertion per element of the packing.

\looseness=-1
As the \(P_3\)s in the above packing are vertex-disjoint, the value by which the packing can decrease the parameter is limited.
In the previous example with 315 edits, subtracting the resulting lower bound would reduce the parameter by at most~33.
In their conclusion, van Bevern et al.~\cite{DBLP:journals/mst/BevernFK18} asked whether \pCE\ is fixed-parameter tractable when parameterized above a stronger lower bound, the size of a modification-disjoint packing of~$P_3$s.
Here, a packing $\mathcal{H}$ of induced $P_{3}$s in $G$ is \emph{modification-disjoint}
if every two $P_{3}$s in $\mathcal{H}$ do not share edges nor non-edges, that is, they share at most one vertex.
The formal problem definition is as follows.

\defprobMD{\pCEAlong~(\pCEA)}
{A graph $G=(V,E)$, a modification-disjoint packing \pkg\ of induced $P_{3}$s of $G$, and a non-negative integer~$\ell$.}
{Is there a cluster editing set, i.e., a set of vertex pairs
$S\subseteq \binom{V}{2}$ so that $G\triangle S$ is a union of disjoint cliques, with $|S|-|\mathcal{H}|\leq\ell$?}

We also say that a set $S$ as above is a \emph{solution}.

\paragraph{Our results.}
At Shonan Meeting no.\ 144~\cite{Shonan144} Christian Komusiewicz re-iterated the question of van Bevern et al.~\cite{DBLP:journals/mst/BevernFK18} and it was also asked in Vincent Froese's dissertation~\cite{froese_finegrained_2018}.
In this paper, we answer this question negatively by showing the following.
\begin{theorem} \label{main}
\pCEAlong{} is NP-hard even for $\ell=0$ and when each vertex in the input graph is incident with at most 23 $P_3$s of \pkg.
\end{theorem}
\noindent In other words, given a graph $G$ and a packing $\mathcal{H}$ of
modification-disjoint $P_3$s in $G$, it is NP-hard to decide if one can delete or insert exactly one edge per element of $\mathcal{H}$ to obtain a cluster graph.
Proving \cref{main} was surprisingly nontrivial.
A straightforward approach would be to amend the known reductions~\cite{komusiewicz_cluster_2012,fomin_subexponential_2013} that show NP-hardness for constant maximum vertex degree by specifying a suitable packing of \(P_3\)s.
However, an argument based on the linear-programming relaxation of packing modification-disjoint \(P_3\)s shows that the graphs produced by these reductions do not admit tight \(P_3\) packing bounds.
We did not find a way around this issue and thus developed a novel reduction based on new gadgets.

The verdict spelt by \cref{main} is unfortunately quite damning.
It indicates that even just reaching the lower bound given by a modification-disjoint \(P_3\) packing already captures the algorithmic hardness of the problem.
However, there may be a way out of this conundrum:
Call a modification-disjoint \(P_3\) packing \emph{\(1/c\)-integral} if each vertex is in at most~\(c\) packed~\(P_3\)s (and say \emph{integral} in place of \emph{1-integral} and \emph{half-integral} in place of \emph{$1/2$-integral}).
As the case \(c = 1\) is just the case of vertex-disjoint packings, van Bevern et al.~\cite{DBLP:journals/mst/BevernFK18} showed that \pCE\ parameterized by the excess over integral \(P_3\) packings is fixed-parameter tractable.
Thus it becomes an intriguing question to find the largest \(c < 23\) such that \pCEA\ remains tractable with respect to the excess over \(1/c\)-integral packings.
We provide progress towards answering this question here.
The problem \pCEATlong~(\pCEAT) is defined in the same way as \pCEA\ except that the input packing~\(\pkg\) is half-integral.
It turns out that the complexity of the problem indeed drops when making the packing half-integral:
\newcommand\thmhalfintegralxpstatement{%
  \pCEATlong\ parameterized by the number~\(\ell\) of excess edits is in XP.
  It can be solved in \(n^{2\ell + O(1)}\) time, where \(n\) is the number of vertices in the input graph.%
}
\begin{theorem}\label{thm:half-integral-xp}
  \thmhalfintegralxpstatement
\end{theorem}
A straightforward idea to prove \cref{thm:half-integral-xp} would be to adapt the fixed-parameter algorithm for vertex-disjoint packings given by van~Bevern~et al.~\cite{DBLP:journals/mst/BevernFK18}.
Their main idea is to show that if a packed \(P_3\)~\(P\) of the input graph~\(G\) admits a solution that is optimal for \(P\) and that respects certain conditions on the neighborhood of \(V(P)\) in \(G\) then this solution can be used in an optimal cluster-editing set for~\(G\).
Afterwards, each packed \(P_3\)~\(P\) either needs an excess edit in \(V(P)\) or an edit incident with \(V(P)\) in~\(G\).
Since the \(P_3\)s in the packing are vertex-disjoint, an edit incident with \(V(P)\) will be in excess over the packing lower bound as well.
It then follows that the overall number of edits is bounded by a function of the excess edits.

\looseness=-1
Unfortunately, the above idea fails for modification-disjoint packings for two reasons.
First, the property that packed \(P_3\)s have an edit incident with them is not helpful anymore, because these edits may be part of other packed \(P_3\)s and hence not be in excess.
Second, if we would like to preserve that these edits are excess, we need to check the special neighborhood properties of van Bevern et al.~\cite{DBLP:journals/mst/BevernFK18} for arbitrarily large connected components of packed \(P_3\)s efficiently.
We did not see a way around these issues and instead designed an algorithm from scratch:
A straightforward guessing of the excess edits reduces the problem to the case where we need to check for zero excess edits.
This case is then solved by an extensive set of reduction rules that exploit the structure given by the half-integral packing.
Essentially, we successively reduce the maximum size of clusters in the final cluster graph.
This then allows us to reduce the problem to \textsc{Cluster Deletion}.
Together with the properties of the packing, this problem allows a formulation as a \textsc{2-SAT} formula which we then solve in polynomial time.

\paragraph{Organization.}
After brief preliminaries in \cref{sec:prelims}, we give some intuition about \pCEA\ in \cref{sec:intuition}.
Then we proceed to the reduction used to show \cref{main} in \cref{sec:construction} (containing the construction) and \cref{sec:correctness} (containing the correctness proof).
\cref{sec:xp-algorithm} then contains the proof of \cref{thm:half-integral-xp}.

\section{Preliminaries}\label{sec:prelims}
In this paper, we denote an undirected graph by $G=(V,E)$, where $V = V(G)$ is the set of vertices, $E = E(G)$ is the set of edges, and $\binom{V}{2}\setminus E$ is the set of \emph{non-edges}.
An undirected edge between two vertices $u$ and $v$ will be denoted by $uv$ where we put $uv = vu$. An undirected non-edge between two vertices $x$ and $y$ will be denoted by $xy$, where we put $xy = yx$, and we will explicitly mention that $xy$ is a non-edge in case of confusion with the notation of an edge. If $uv$ is an edge in the graph, we say $u$ and $v$ are \emph{adjacent}.
We denote a bipartite graph by $B=(U,W,E)$, where $U,W$ are the two parts of the vertex set of $B$ and $E$ is the set of edges of~$B$. We say that a bipartite graph is \emph{complete} if for every pair of vertices $u\in U$ and $w\in W$, $uw \in E$. For a non-empty subset of vertices $X\subseteq V$, we denote the subgraph induced by $X$ by $G[X]$. A \emph{clique} $Q$ in a graph $G$ is a subgraph of $G$ in which any two distinct vertices are adjacent. A \emph{cluster graph} is a graph in which every connected component is a clique. A connected component in a cluster graph is called a \emph{cluster}.

\looseness=-1
Let $G'$ be a cluster graph and let $S$ be a cluster editing set $S$ such that $G\triangle S=G'$. We say that two cliques $Q_1$ and $Q_2$ of $G$ are \emph{merged} (in $G'$) if they belong to the same cluster in~$G'$. We say that $Q_1$ and $Q_2$ are \emph{separated} (in $G'$) if they belong to two different clusters in $G'$.
When mentioning the edges or non-edges between the vertices of the clique $Q_1$ and the vertices of the clique $Q_2$, we refer to the edges or non-edges between the clique $Q_1$ and the clique $Q_2$ for short. Let $\ell, r \in \mathbb{N}$. We denote a path with $\ell$ vertices by $P_{\ell}$ and a cycle with $r$ vertices by~$C_r$.

\looseness=-1
Let $x, y, z$ be vertices in a graph~$G$.
We say that $xyz$ is an \emph{induced} $P_3$ of $G$ if $xy, yz\in E(G)$ and $xz\notin E(G)$.
Vertex~$y$ is called the \emph{center} of~$xyz$.
We say that vertices $x,y,z$ \emph{belong to} $xyz$ or $x,y,z$ are \emph{incident with} $xyz$.
We also say that $xyz$ is \emph{incident with} the vertices $x,y$ and $z$.
In this paper, all $P_3$s we mention are induced $P_3$s; we sometimes skip the qualifier ``induced'' for convenience.

\looseness=-1
Given an instance $(G,{\mathcal{H}},\ell)$ of \pCEA, if $xyz$ is a $P_3$ in $G$ and $xyz\in \mathcal{H}$, we say that $xyz$ is \emph{packed}, and we say that the edges $xy,yz$ are \emph{covered} by $xyz$ and the non-edge $xz$ is \emph{covered} by $xyz$. If an edge $xy$ is covered by some $P_3$ of $\mathcal{H}$, we say that $xy$ is a \emph{packed} edge. Otherwise we say that $xy$ is a \emph{non-packed} edge. If a non-edge $uv$ is covered by some $P_3$ of $\mathcal{H}$, we say that $uv$ is a \emph{packed} non-edge. Otherwise we say that $uv$ is a \emph{non-packed} non-edge.
If none of the edges of a path $P$ is packed, we say that the path $P$ is \emph{non-packed}.
If $xyz$ is a $P_3$ in $G$ and $Q_1$, $Q_2$, and $Q_3$ are pair-wise non-intersecting vertex sets of $G$, we say that $xyz$ \emph{connects $Q_1$ and $Q_3$ via $Q_2$} if the center $y$ of $xyz$ belongs to $Q_2$ and $x,z$ belong to $Q_1$ and $Q_3$, respectively.

We sometimes need finite fields of prime order.
Let $p$ be some prime.
By $\mathbb{F}_p$ we denote the finite field with the $p$ elements $0, \ldots, p - 1$ with addition and multiplication modulo~$p$.
Let $x \in \mathbb{F}_p$.
Where it is not ambiguous, $-x$ and $x^{-1}$ will denote the additive and multiplicative inverse, respectively, of $x$ in $\mathbb{F}_p$.

When we say that we relabel the vertices of a graph, we use $v \leftarrow u$ to denote that we relabel the vertex $v$ by the new label $u$.

\section{Intuition}\label{sec:intuition}

Before giving the hardness proof, it is instructive to determine some easy and difficult cases when solving \pCEA\ with $\ell = 0$.
This will give us an intuition about the underlying combinatorial problem that we need to solve.

\newcommand{\gfix}{\ensuremath{G_\textsf{fix}}} \newcommand{\gsol}{\ensuremath{G_\textsf{sol}}} Let $(G, \pkg, 0)$ be an instance of \pCEA.
It is helpful to consider the subgraph \gfix\ of $G$ that contains only those edges of $G$ that are not contained in any $P_3$ in \pkg, that is, the non-packed edges.
Suppose that $(G, \pkg, 0)$ has a solution~$S$ and let \gsol\ be the associated cluster graph.
Observe that each connected component of \gfix\ is part of a single cluster in~\gsol.
Let us hence call the connected components of \gfix\ \emph{proto-clusters}.
Our task in finding \gsol\ is thus indeed to find a vertex partition $\mathcal{P}$ that is coarser than the vertex partition given by the proto-clusters and that satisfies certain further conditions.
The additional conditions herein are given by the $P_3$s in $G$ and also by the non-edges of $G$ which are not contained in any $P_3$ in \pkg, that is, by the non-packed non-edges.
A non-packed non-edge between two proto-clusters implies that these proto-clusters cannot be together in a cluster in \gsol.
Hence, we are searching for a vertex partition $\mathcal{P}$ as above subject to the constraints that certain proto-cluster pairs end up in different parts.

\begin{figure}[t]
  \centering
  \begin{tikzpicture}[
    every edge/.style = {draw, semithick}
    ]

    \tikzstyle{cluster} = [circle, draw, minimum size = 1cm]

    \node [cluster, label = left:A] (A) at (-1,-1) {};
    \node [cluster, label = above:B] (B) at (0,0) {};
    \node [cluster, label = below:C] (C) at (1,-1) {};
    \node [cluster, label = above:D] (D) at (2,0) {};
    \node [cluster, label = right:E] (E) at (3,-1) {};

    \draw (B) edge [dashed] (D);

    \node [vertex] (a) at (-1,-1) {};
    \node [vertex] (b) at (0,0) {};
    \node [vertex] (c) at (.9,-1) {};
    \node [vertex] (d) at (1.1,-1) {};
    \node [vertex] (e) at (2,0) {};
    \node [vertex] (f) at (3,-1) {};

    \draw (b) edge (a) edge (c);
    \draw (e) edge (d) edge (f);
  \end{tikzpicture}
  \caption{Five proto-clusters $A$ through $E$ and two $P_3$s in the underlying graph and in the $P_3$-packing that connects $A$ to $C$ via $B$ and $C$ to $E$ via $D$, respectively. The dashed edge between $B$ and $D$ means that there is a non-packed non-edge between $B$ and $D$.}
  \label{fig:dividing+p3}
\end{figure}

The constraints on $\mathcal{P}$ given by $P_3$s in $G$ can be distinguished based on the intersection of the $P_3$s with the proto-clusters.
We only want to highlight two situations that are most relevant for the hardness construction.
The first situation is when a $P_3$, name it $P$, intersects with three proto-clusters $D_1$, $D_2$, and $D_3$, each in exactly one vertex and with center vertex in $D_2$.
The corresponding constraint on $\mathcal{P}$ is that either $D_1$ and $D_2$ are merged or $D_2$ and $D_3$ are merged into one cluster.
We can satisfy such constraints easily, in the absence of further constraints, by merging all proto-clusters into one large cluster.
However, together with the constraints from non-packed non-edges a difficult picture emerges.
Consider \cref{fig:dividing+p3}.
Proto-clusters $B$ and~$D$ cannot be merged into one cluster because of a non-packed non-edge.
However, there is a path in $G$ from $B$ to $D$ via vertices of $C$.
Hence, either $B$ and $C$ are in different clusters in \gsol\ or $C$ and $D$ are.
If $B$ and $C$ are in different clusters, then since we have only budget one for the $P_3$ involving $A$, $B$, and $C$, it follows that $A$ and $B$ are merged into one cluster in \gsol.
It is not hard to imagine that such behavior can be very non-local and in fact two different generalizations of this behavior form the basis for the variable and clause gadget in our hardness reduction.

The second case is when there is a $P_3$ in $G$ and also in the packing \pkg\ that has an edge contained in one proto-cluster~$A$ and the remaining vertex in a different proto-cluster~$B$.
Call this $P_3$ $P$.
Intuitively, regardless of whether $A$ and $B$ are merged into one cluster in \gsol, $P$ can be edited without excess cost over \pkg\ to accommodate this choice.
In our hardness reduction, a main difficulty will be to pad subconstructions with $P_3$s in the packing \pkg, so that we are able to find a solution with zero excess edits.
For this we will heavily use $P_3$s of the form that we just described.

\section{NP-hardness for tight modification-disjoint packings} \label{sec:lower-bound}

In this section, we prove Theorem \ref{main} by showing a reduction from the NP-hard problem of deciding satisfiability of 3-CNF formulas.
Given a 3-CNF formula $\Phi$, we construct a graph $G=(V,E)$ with a modification-disjoint packing~$\mathcal{H}$ of induced $P_{3}s$ such that $\Phi$ has a satisfying assignment if and only if $G$ has a cluster editing set $S$ which consists of exactly one vertex pair of each $P_{3}$ in \pkg.
In other words, the \pCEA\ instance $(G, \pkg, 0)$ is a YES-instance.
We assume that every clause of $\Phi$ has exactly $3$ literals of pair-wise different variables as we can preprocess the formula to achieve this in polynomial time otherwise.
Similarly, we can assume that every variable of $\Phi$ appears at least twice.
In the following, we let $m$ denote the number of clauses in $\Phi$, denote the clauses of $\Phi$ by $\Gamma_0,\ldots,\Gamma_{m-1}$, let $n$ be the number of variables, and denote the variables of $\Phi$ by $x_0,\ldots,x_{n-1}$.
Furthermore, we let $m_i$ denote the number of clauses that contain the variable $x_i$, $i \ = 0, \ldots, n-1$.

\subsection{Construction}\label{sec:construction}

The outline of our construction is as follows.
In \cref{sec:var-gadget,sec:clause-gadget} we explain the basic construction of the variable and clause gadgets.
In these two sections we first show how to construct a subgraph of the final construction that enables us to show the soundness, that is, if the \pCEA\ instance is a yes-instance, then $\Phi$ is satisfiable.
The main difficulty is then to extend this construction so that the completeness also holds.
This we do in \cref{sec:merge-model,sec:implementation}.
\Cref{sec:completeness,sec:soundness} then contain the correctness proof.

Both the variable gadget and the clause gadget rely on some ideas outlined in \cref{sec:intuition}.
Our basic building blocks will be proto-clusters.
A proto-cluster is a subgraph that is connected through edges that are not contained in any $P_3$ in the constructed packing~\pkg.
The proto-clusters then have to be joined into larger clusters in a way that represents a satisfying assignment to $\Phi$.
The variable gadget basically consists of an even-length cycle of proto-clusters, connected by $P_3$s so that either odd or even pairs of proto-clusters on the cycle have to be merged.
These two options represent a truth assignment.
The construction of the variable gadget is more involved than a simple cycle of proto-clusters, however, because of the connection to the clause gadgets: We need to ensure that all vertex pairs between certain proto-clusters of a variable and clause gadget are covered by $P_3$s in \pkg, so to be able to merge these clusters in the completeness proof.
The way in which we cover these vertex pairs imposes some constraints on the construction of the variable gadgets, making the gadgets more complicated.

\subsubsection{Variable gadget}\label{sec:var-gadget}
As mentioned, a variable will be represented by a cycle of proto-clusters such that any solution needs to merge either each odd or each even pair of consecutive proto-clusters.
These two options represent the truth value assigned to the variable.
In order to enable both associated solutions with zero edits above the packing lower bound, we build an associated packing of $P_3$s such that all vertex pairs between consecutive proto-clusters are covered by a $P_3$ in the packing.
It would be tempting to make each proto-cluster a single vertex.
However, due to the connections to the clause gadget later on, we need proto-clusters containing five vertices each.

Recall that $m_i$ denotes the number of clauses that contain the variable $x_i$, $i = 0, 1, \ldots, n-1$.
For each variable~$x_i$, $i=0, 1,\ldots,n-1$, we create $4m_i$ vertex-disjoint cliques with~$5$ vertices each, namely $K_{0}^{i},\ldots,K_{4m_i-1}^{i}$.
In each $K_{j}^{i}$, $j = 0, 1, \ldots, 4m_i-1$, the vertices are $v_{j,0}^{i},\ldots,v_{j,4}^{i}$.
For each $j=0,2,\ldots,4m_{i}-2$, we create $P_{3}$s connecting $K_{j}^{i},K_{j+1}^{i}$ and $K_{j+2}^{i}$ (where we identify $K_{0}^{i}$ as $K_{4m_{i}}^{i}$) as we explain below, adding all edges between each two consecutive cliques.

Throughout the construction, the cliques we have just introduced will remain proto-clusters, that is, they contain a spanning tree of edges that are not covered by $P_3$s in the packing~\pkg.
We now add pairwise modification-disjoint $P_3$s so as to cover all edges between the cliques~$K_j^i$ we have just introduced.
Recall that $\mathbb{F}_5$ is the finite field of the integers modulo~5.
We take three consecutive cliques and add $P_3$s with one vertex in each of the three cliques.
To do this without overlapping two $P_3$s, we think about the cliques' vertices as elements of $\mathbb{F}_5$ and add a $P_3$ for each possible arithmetic progression.
That is, in each added $P_3$ the difference of the first two elements of the $P_3$ is equal to the difference of the second two elements.
In this way, each vertex pair is contained in a single $P_3$ since the third element is uniquely defined by the arithmetic progression.

Formally, for each $j=0,2,\ldots,4m_{i}-2$ and every triple of elements $p,q,r \in \mathbb{F}_5$ satisfying the equality $q - p = r - q$ over $\mathbb{F}_5$, we add to the graph the edges $v_{j, p}^{i}v_{j+1, q}^{i}$ and $v_{j+1, q}^{i} v_{j+2, r}^{i}$ and we add to the packing $\mathcal{H}$ the $P_3$ given by $v_{j, p}^i v_{j+1, q}^i v_{j+2, r}^i$.
Note that in this manner the clique $K_{j+1}^i$ becomes fully adjacent to $K_j^i$ and to $K_{j+2}^i$ while $K_{j+1}^i$ stays anti-adjacent to all other cliques $K_{j'}^i$.

Observe that the $P_3$s given by $v_{j, p}^{i}v_{j+1,q}^{i}v_{j+2, r}^{i}$ for $j=0,2,\ldots,4m_{i}-2$ such that $q - p = r - q$ are pairwise modification-disjoint: For each $j=0,2,\ldots,4m_{i}-2$, an arbitrary edge just introduced between $K_j^i$ and $K_{j + 1}^i$ has the form $\{v_{j, p}^{i}, v_{j+1, q}^{i}\}$ for some $p,q\in \mathbb{F}_5$.
It belongs to the unique $P_3$ given by $v_{j, p}^{i}v_{j+1, q}^{i}v_{j+2, r}^{i}$, where $r = 2q - p$.
Similarly, an arbitrary edge $\{v_{j+1, q}^{i}, v_{j+2, r}^{i}\}$ for $q,r\in \mathbb{F}_5$ belongs to the unique $P_3$ given by $v_{j, 2q-r}^{i}v_{j+1, q}^{i}v_{j+2, r}^{i}$ and an arbitrary non-edge $\{v_{j, p}^{i}, v_{j+2, r}^{i}\}$ for $p,r\in \mathbb{F}_5$ belongs to the unique $P_3$ given by $v_{j, p}^{i}v_{j+1, (p+r)\cdot 2^{-1}}^{i}v_{j+2, r}^{i}$, where $2^{-1}$ is the multiplicative inverse of $2$ over $\mathbb{F}_5$, that is, $2^{-1} = 3$.

After this construction, we set the modification-disjoint packing of the variable gadget to be
$${\mathcal{H}}_{\textsf{var}}=\{P_{3}\text{ given by }v_{j, p}^{i}v_{j+1, q}^{i}v_{j+2, r}^{i}~|~i=0,\ldots,n-1\text{; }j=0,2,\ldots,4m_{i}-2\text{; } p,q,r\in\mathbb{F}_5 \text{; and } q - p = r - q \}.$$
This finishes the first stage of the construction.
Notice that the cliques $K_j^i$ form a cyclic structure.
Intuitively, every second pair of cliques needs to be merged into one cluster by any solution due to the $P_3$s we have introduced, and we will see that the two resulting solutions are in fact the only ones.
The truth values of the variable are then represented as follows.
For every variable $x_i$, $i=0,\ldots,n-1$, if $K_{j}^{i}$ and $K_{j+1}^{i}$ are merged for $j=0, 2,\ldots,4m_{i}-2$, then this represents the situation that we assign false to the variable $x_i$.
If $K_{j+1}^{i}$ and $K_{j+2}^{i}$ are merged for $j=0, 2,\ldots,4m_{i}-2$, then this represents variable~$x_i$ being true.
We will make minor modifications to the variable gadgets and ${\mathcal{H}}_{\textsf{var}}$ in the following section, so as to transmit the choice of truth value to the clause gadgets.

\subsubsection{Skeleton of the clause gadget}\label{sec:clause-gadget}

In order to introduce the construction of the clause gadget, we first give a description of the skeleton of the clause gadget.
The skeleton is a subgraph of the final construction that allows us to prove the soundness.
The final construction is given in the succeeding sections.
We give a picture of the skeleton in Fig.~\ref{fig1}.
The basic idea is a generalization of the idea explained in \cref{sec:intuition}: A clause $\Gamma_d$ is represented by four proto-clusters (cliques), $Q^i_\cl{}$, $i = 1, \ldots, 4$, as in Fig.~\ref{fig1}.
The proto-clusters are connected by a path~$P$ of length $5$ containing vertices of $Q^1_\cl{}$, $Q^2_\cl{}$, $Q^3_\cl{}$, and $Q^4_\cl{}$ in that order.
However, between $Q^1_\cl{}$ and $Q^4_\cl{}$ there is a non-packed non-edge, meaning that every solution has to cut the path~$P$ by deleting all edges between $Q^1_\cl{}$ and $Q^2_\cl{}$, or between $Q^2_\cl{}$ and $Q^3_\cl{}$, or between $Q^3_\cl{}$ and $Q^4_\cl{}$.
We use this three-way choice to force the solution to select a variable that satisfies the clause~$\Gamma_\cl$.

\begin{figure}[tbp]
  \centering
  \scalebox{0.85}{\begin{tikzpicture}[scale=0.5]
		\node [style=b] (0) at (-3.25, 1.5) {};
		\node [style=b] (1) at (-3.25, 4.75) {};
		\node [style=bs] (2) at (-8.75, -1.25) {};
		\node [style=b] (3) at (-7.25, 3.25) {};
		\node [style=b] (4) at (-0.25, 4.75) {};
		\node [style=bs] (5) at (-12.25, 2) {};
		\node [style=bs] (6) at (-10, 1) {};
		\node [style=bs] (7) at (-17.25, 0.25) {};
		\node [style=b] (8) at (-0.25, 1.5) {};
		\node [style=bs] (9) at (-15, 1.75) {};
		\node [style=b] (11) at (-3.25, 4.75) {};
		\node [style=b] (12) at (4.75, 3.75) {};
		\node [style=b,label={[label distance=-1mm]90:$K_{4\pi(\vr{},\cl{})+2}^{\vr{}}$}] (13) at (10.25, 3.25) {};
		\node [style=b,label={[label distance=-2mm]190:$K_{4\pi(\vr{},\cl{})+1}^{\vr{}}$}] (14) at (7.75, 2) {};
		\node [style=b,label={[label distance=-1mm]180:$K_{4\pi(\vr{},\cl{})}^{\vr{}}$}] (15) at (7.75, -1) {};
		\node [style=b] (16) at (0, -2.5) {};
		\node [style=b] (17) at (-1, -7.5) {};
		\node [style=b] (18) at (1, -5.25) {};
		\node [style=b] (19) at (4.000001, -6) {};

		\node [style=vertex] (20) at (-3.25, 4.75) {};
		\node [style=vertex] (21) at (-7, 3.5) {};
		\node [style=vertex] (22) at (-7, 3) {};
		\node [style=vertex] (23) at (-3.5, 1.75) {};
		\node [style=vertex] (24) at (-3, 1.75) {};
		\node [style=vertex] (25) at (-0.25, 1.5) {};
		\node [style=vertex] (26) at (-3, 1.5) {};
		\node [style=vertex] (27) at (-3.25, 1.25) {};
		\node [style=vertex] (28) at (-0.2499996, -2.25) {};
		\node [style=vertex] (29) at (0.2499996, -2.25) {};
		\node [style=vertex] (30) at (-0.25, 4.75) {};
		\node [style=vertex] (31) at (-0.5, 1.75) {};
		\node [style=vertex] (32) at (0, 1.75) {};
		\node [style=vertex] (33) at (4.5, 4) {};
		\node [style=vertex] (34) at (4.5, 3.5) {};

		\node [style=largeVertex] (35) at (-7.5, 3.25) {};
		\node [style=largeVertex] (36) at (-9.75, 1) {};
		\node [style=largeVertex] (37) at (-8.6, -1.25) {};

		\node [style=largeVertex] (38) at (5, 3.75) {};
		\node [style=largeVertex] (39) at (7.75, 2.25) {};
		\node [style=largeVertex] (40) at (7.75, -1) {};
		\node [style=largeVertex] (41) at (0.2499996, -2.75) {};
		\node [style=largeVertex] (42) at (1, -5) {};
		\node [style=largeVertex] (43) at (-1, -7.4) {};
		
    \node [style=bs] (67) at (-17.75, -5.25) {};
		\node [style=bs] (68) at (-16, -7.25) {};
		\node [style=bs] (69) at (-13.5, -8) {};
		\node [style=bs] (70) at (-10.75, -7.5) {};
		\node [style=bs] (71) at (-9.000001, -6) {};
		\node [style=bs] (72) at (-8.5, -3.75) {};
		\node [style=bs] (73) at (-18.25, -2.5) {};

		\node  (44) at (-3.25, 6.5) {$Q_{\cl{}}^{1}$};
		\node  (45) at (-0.25, 6.5) {$Q_{\cl{}}^{4}$};
		\node (46) at (-4.7, 0.75) {$Q_{\cl{}}^{2}$};
		\node (47) at (1.2, 0.75) {$Q_{\cl{}}^{3}$};
		\node  (48) at (-7.25, 5) {$T_{\cl{}}^{\vp{}}$};
		\node  (49) at (4.75, 5.5) {$T_{\cl{}}^{\vr{}}$};
		\node (50) at (-1.6, -2.2) {$T_{\cl{}}^{\vq{}}$};

		\node  (51) at (-12.5, 3.5) {\small $K_{4\pi(\vp{},\cl{})}^{\vp{}}$};
		\node  (52) at (-7.5, 1) {\small $K_{4\pi(\vp{},\cl{})+1}^{\vp{}}$};
		\node  (53) at (-6.3, -1.5) {$K_{4\pi(\vp{},\cl{})+2}^{\vp{}}$};
		\node  (76) at (-6, -3.75) {$K_{4\pi(\vp{},\cl{})+3}^{\vp{}}$};

		\node [style=gn] (74) at (-11, 2.25) {\textbf{F}};
		\node [style=gn] (75) at (-8.75, 0.25) {\textbf{T}};
		\node [style=gn] (76) at (-8.000001, -2.5) {\textbf{F}};
		\node [style=gn] (77) at (-7.75, -5) {\textbf{T}};
		\node [style=gn] (78) at (-9.25, -7.5) {\textbf{F}};
		\node [style=gn] (79) at (-12, -8.75) {\textbf{T}};
		\node [style=gn] (80) at (-15.25, -8.5) {\textbf{F}};
		\node [style=gn] (81) at (-17.5, -7) {\textbf{T}};
		\node [style=gn] (82) at (-19, -4.25) {\textbf{F}};
		\node [style=gn] (83) at (-18.75, -0.9999999) {\textbf{T}};
		\node [style=gn] (84) at (-16.75, 1.75) {\textbf{F}};
		\node [style=gn] (85) at (-13.75, 2.75) {\textbf{T}};

	   \node [style=none] (60) at (-1.5, -4.5) {$K_{4\pi(\vq{},\cl{})+1}^{\vq{}}$};
		\node [style=none] (61) at (4.3, -4.4) {\small $K_{4\pi(\vq{},\cl{})+2}^{\vq{}}$};
		\node [style=none] (62) at (-3.3,-7.5) {$K_{4\pi(\vq{},\cl{})}^{\vq{}}$};
		\node [style=none] (63) at (2.5, -4.75) {\textbf{T}};
		\node [style=none] (64) at (-0.45, -6) {\textbf{F}};
        \node [style=none] (65) at (8.9,3.4) {\textbf{T}};
        \node [style=none] (65) at (7.0,0.5) {\textbf{F}};
		\draw [style=normal] (20) to (21);
		\draw [style=normal] (20) to (22);
		\draw [style=normal] (20) to (23);
		\draw [style=normal] (24) to (20);
		\draw [style=normal] (25) to (26);
		\draw [style=normal] (25) to (27);
		\draw [style=normal] (25) to (28);
		\draw [style=normal] (25) to (29);
		\draw [style=normal] (30) to (31);
		\draw [style=normal] (30) to (32);
		\draw [style=normal] (30) to (33);
		\draw [style=normal] (30) to (34);

		\draw [style=thickEdge] (5) to (6);
		\draw [style=thickEdge] (6) to (2);
		\draw [style=thickEdge] (9) to (7);
		\draw [style=thickEdge] (9) to (5);
		\draw [style=thickEdge] (18) to (17);
		\draw [style=thickEdge] (18) to (19);
		\draw [style=thickEdge] (14) to (15);
		\draw [style=thickEdge] (14) to (13);

		\draw [style=thickGreen] (35) to (36);
		\draw [style=thickGreen] (36) to (37);
		\draw [style=thickGreen] (38) to (39);
		\draw [style=thickGreen] (39) to (40);
		\draw [style=thickGreen] (42) to (41);
		\draw [style=thickGreen] (42) to (43);
		
    \draw [style=thickEdge](67) to (68);
		\draw [style=thickEdge](68) to (69);
		\draw [style=thickEdge](69) to (70);
		\draw [style=thickEdge](70) to (71);
		\draw [style=thickEdge](71) to (72);
		\draw [style=thickEdge](2) to (72);
		\draw [style=thickEdge] (73) to (7);
		\draw [style=thickEdge](73) to (67);

        \draw[fill=black,fill opacity=0.8] (35) ellipse (6pt and 13pt);
        \draw[fill=black,rotate=45,fill opacity=0.8] (36) ellipse (6pt and 13pt);
        \draw[fill=black,rotate=90,fill opacity=0.8] (37) ellipse (6pt and 13pt);

        \draw[fill=black,fill opacity=0.8] (38) ellipse (6pt and 13pt);
        \draw[fill=black,rotate=-45,fill opacity=0.8] (39) ellipse (6pt and 13pt);
        \draw[fill=black,rotate=-90,fill opacity=0.8] (40) ellipse (6pt and 13pt);

        \draw[fill=black,fill opacity=0.8,rotate=90] (41) ellipse (6pt and 13pt);
        \draw[fill=black,rotate=90,fill opacity=0.8] (42) ellipse (6pt and 13pt);
        \draw[fill=black,rotate=90,fill opacity=0.8] (43) ellipse (6pt and 13pt);

		\draw [gray, line width=3.5mm,dotted] (17) to [bend right=60](19);
		\draw [gray, line width=3.5mm,dotted] (15) to [bend right=60](13);

        \draw [blue, dotted, line width=0.3mm] (23) to [bend right=60](-3.75,1) to [bend right=60] (27);
        \draw [blue, dotted, line width=0.3mm] (25) to [bend right=60](0.4,1.25) to [bend right=60] (32);
\end{tikzpicture}}
        \caption{Skeleton of a clause gadget $\Gamma_{\cl{}}=(x_{\vp}\lor \neg x_{\vq{}}\lor \neg x_{\vr{}})$. The white circles represent cliques. The blue dotted lines inside $Q_\cl^2$ and $Q_\cl^3$ indicate that $Q_\cl^1$, $Q_\cl^2$, $Q_\cl^3$, and $Q_\cl^4$ are in one connected component. A pair of incident brown thick lines indicates a set of four transferring $P_3$s used to connect a clause gadget to a variable gadget. The cycles made from cliques and gray thick lines represent variable gadgets, where a dashed gray line indicates an omitted part of the cycle. The cycle for variable $x_\vp$ is shown completely, where we assume that $m_{\vp}=3$, that is, variable $x_\vp$ is in three clauses.
          Labels \textbf{T} and \textbf{F} on thick gray edges indicate the pairs of cliques that shall be merged into one cluster if the variable is to be set to \true\ or \false, respectively.} \label{fig1}\label{fig:skeleton}
\end{figure}

\paragraph{Main gadget.} Formally, for each variable $x_i$, $i=0,1,\ldots,n-1$, we fix an arbitrary ordering of the clauses that contain~$x_i$.
If a clause $\Gamma_j$ contains a variable $x_i$, let $\pi(i,j)\in\{0,\ldots,m_i-1\}$ denote the position of the clause $\Gamma_j$ in this ordering.
Let initially ${\mathcal{H}}_{\textsf{tra}}=\emptyset$.
For each clause $\Gamma_{\cl{}}$ ($\cl{}=0,\ldots,m-1$) proceed as follows.
We first introduce four cliques $Q_{\cl{}}^{1}, Q_{\cl{}}^{2}, Q_{\cl{}}^{3}$ and $Q_{\cl{}}^{4}$.
Let $\Gamma_\cl{}$ contain the variables $x_{\vp{}}, x_{\vq{}}$ and $x_{\vr{}}$.
We introduce the cliques $T_{\cl{}}^{\vp{}}, T_{\cl{}}^{\vq{}}$ and $T_{\cl{}}^{\vr{}}$, called \emph{transferring cliques}.
All of the cliques introduced are pairwise vertex disjoint and can be of different sizes.
We will give the exact sizes in \cref{sec:implementation}.

Next, we introduce the following $P_3$s:
\begin{itemize}
\item Introduce two $P_3$s, $P_{\cl{}}^{1}$ and $P_{\cl{}}^{2}$, that both connect $T_{\cl{}}^{\vp{}}$ and $Q_{\cl{}}^{2}$ via $Q_{\cl{}}^{1}$, such that $P_{\cl{}}^{1}$ and $P_{\cl{}}^{2}$ share the same vertex in~$Q_{\cl{}}^{1}$.
\item Introduce two $P_3$s, $P_{\cl{}}^{3}$ and $P_{\cl{}}^{4}$, that both connect $T_{\cl{}}^{\vq{}}$ and $Q_{\cl{}}^{2}$ via $Q_{\cl{}}^{3}$, such that $P_{\cl{}}^{3}$ and $P_{\cl{}}^{4}$ share the same vertex in $Q_{\cl{}}^{3}$.
\item Introduce two $P_3$s, $P_{\cl{}}^{5}$ and $P_{\cl{}}^{6}$, that both connect $T_{\cl{}}^{\vr{}}$ and $Q_{\cl{}}^{3}$ via $Q_{\cl{}}^{4}$, such that $P_{\cl{}}^{5}$ and $P_{\cl{}}^{6}$ share the same vertex in $Q_{\cl{}}^{4}$.
\end{itemize}
All the $P_3$s $P_\cl^i$ are pairwise vertex-disjoint except for the pairs sharing the center (as explicitly mentioned in the description).
We add each $P_{\cl{}}^{i}$ for $i=1,\ldots,6$ to ${\mathcal{H}}_{\textsf{tra}}$.
We call the $P_{3}$s of ${\mathcal{H}}_{\textsf{tra}}$ \emph{transferring $P_{3}$s}.

\begin{figure}
\centering
\begin{tikzpicture}
		\node [style=varGadget,label={[xshift=0.5mm, yshift=-0.5mm]$v_{5}$}] (0) at (-7, 1) {};
		\node [style=varGadget,label={[label distance=-2.5mm]95:$v_{1}$}] (1) at (-5, 1.5) {};
		\node [style=varGadget] (2) at (-3, 1) {};
		\node [style=varGadget,label={[label distance=-.5mm]270:$v_{6}$}] (3) at (-7, 0) {};
		\node [style=varGadget,label={[label distance=0mm]270:$v_{2}$}] (4) at (-5, 0.5) {};
		\node [style=varGadget,label={[label distance=0mm]0:$v_{3}$}] (5) at (-3, 0) {};
		\node [style=varGadget] (6) at (-7, -1) {};
		\node [style=varGadget] (7) at (-5, -0.5) {};
		\node [style=varGadget,label={[label distance=-1mm]0:$v_{4}$}] (8) at (-3, -1) {};
		\node [style=varGadget,label={[label distance=-1mm]135:$v_{7}$}] (9) at (-7, -2) {};
		\node [style=varGadget] (10) at (-5, -1.5) {};
		\node [style=varGadget] (11) at (-3, -2) {};
		\node [style=varGadget,label={[label distance=-1mm]270:$v_{8}$}] (12) at (-7, -3) {};
		\node [style=varGadget] (13) at (-5, -2.5) {};
		\node [style=varGadget] (14) at (-3, -3) {};
		\node [style=varGadget,label={[label distance=-1mm]90:$w_{1}$}] (15) at (-5, 3.25) {};
		\node [style=varGadget,label={[label distance=-1mm]90:$w_{2}$}] (16) at (-4.25, 3.25) {};
		\node [style=varGadget,label={[label distance=-1mm]90:$w_{3}$}] (17) at (-3.5, 3.25) {};
		\node [style=varGadget,label={[label distance=-1mm]90:$w_{4}$}] (18) at (-2.4, 3.25) {};

		\draw [style=trasEdge1] (15) to (1);
		\draw [style=trasEdge1] (1) to (5);
		\draw [style=trasEdge1,dotted] (15) to (5);
		\draw [style=trasEdge2] (16) to (4);
		\draw [style=trasEdge2] (4) to (8);
		\draw [style=trasEdge2,dotted] (16) to (8);
		\draw [style=trasEdge3] (17) to (4);
		\draw [style=trasEdge3] (4) to (5);
		\draw [style=trasEdge3,dotted] (17) to (5);
		\draw [style=trasEdge4] (18) to (1);
		\draw [style=trasEdge4] (1) to (8);
		\draw [style=trasEdge4,dotted] (18) to (8);
		\draw [style=varGadEdge] (0) to (1);
		\draw [style=varGadEdge] (1) to (2);
		\draw [style=varGadEdge] (6) to (7);
		\draw [style=varGadEdge] (7) to (8);
		\draw [style=varGadEdge] (9) to (10);
		\draw [style=varGadEdge] (10) to (11);
		\draw [style=varGadEdge] (12) to (13);
		\draw [style=varGadEdge] (13) to (14);
		\draw [style=newP3Dot] (1) to (12);
		\draw [style=newP3] (1) to (9);
		\draw [style=newP3] (4) to (3);
		\draw [style=newP3Dot] (4) to (0);
		\draw [style=newP3] (0) to (3);
		\draw [style=newP3] (9) to (12);

       \draw[draw=black] (6) ellipse (20pt and 80pt);
       \draw[draw=black] (7) ellipse (20pt and 80pt);
       \draw[draw=black] (8) ellipse (20pt and 80pt);
       \draw[draw=black] (-3.75, 3.6) ellipse (65pt and 25pt);

		\node [style=none] (19) at (-8.35, -1) {$K_{4\pi(\vp{},\cl{})}^{\vp{}}$};
		\node [style=none] (20) at (-6, 2.3) {$K_{4\pi(\vp{},\cl{})+1}^{\vp{}}$};
		\node [style=none] (21) at (-1.45, -1.25) {$K_{4\pi(\vp{},\cl{})+2}^{\vp{}}$};
		\node [style=none] (22) at (-6.35, 3.75) {$T_{\cl{}}^{\vp{}}$};
		\node [style=varGadget] (23) at (-4.25, 4) {};
		\node [style=varGadget] (24) at (-3.5, 4) {};

\end{tikzpicture}
\caption{Connection of a clause gadget with a variable gadget for
  a variable~$x_\vp{}$ which appears positively in the clause.
  White ellipses represent cliques.
  The vertices in the cliques in the variable gadget are ordered from top to bottom according to the elements of $\mathbb{F}_5$ which they represent.
  For example, the topmost vertex in $K_{4\pi(\vp{}, \cl{})}^\vp{}$ is $v_{4\pi(\vp{}, \cl{}), 0}^\vp{}$ (corresponding to $0 \in \mathbb{F}_5$) and the bottom-most is $v_{4\pi(\vp{}, \cl{}), 4}^\vp{}$ (corresponding to $4 \in \mathbb{F}_5$).
  The gray lines adjacent to cliques in the variable gadget represent some of the $P_3$s that were introduced into the variable gadgets in the beginning.
  (Some gray lines are super-seeded by edges of other colors.)
  The $P_3$s represented by the gray lines have the associated arithmetic progression ``$+ 0$'', that is, $q - p = r - q = 0$ in the definition of the $P_3$s.
  The $P_3$s for the remaining arithmetic progressions are omitted for clarity.
  In colors red, black, green, and blue we show the $P_3$s that connect the transferring clique $T_{\cl{}}^\vp{}$ with the variable gadget of variable~$x_\vp{}$.
  Herein, dotted lines are non-edges and solid lines are edges.
  Note that these connecting $P_3$s supplant some of the edges of previously present $P_3$s in the variable gadget---the previously present $P_3$s are then removed from both $G$ and \pkg.
  For example the green $P_3$ replaces the edge $v_2 v_3$ of the $P_3$ given by $v_6 v_2 v_3$ that was previously present.
  To maintain that each vertex pair between consecutive cliques in the variable gadget is covered by some $P_3$ in the packing, we add the two brown $P_3$s.}
\label{fig3}\label{fig:connection}
\end{figure}

\paragraph{Connection to the variable gadgets.}
Next we connect the transferring cliques $T_{\cl{}}^{\vp{}}$, $T_{\cl{}}^{\vq{}}$, and $T_{\cl{}}^{\vr{}}$ to the variable gadgets of $x_\vp{}$, $x_\vq{}$, and $x_\vr{}$, respectively.
To avoid additional notation, we only explain the procedure for $T_{\cl{}}^\vp{}$ and $x_\vp{}$, the other pairs are connected analogously.
We connect $T_{\cl{}}^\vp{}$ to the variable gadget of $x_\vp{}$ by a set of four modification-disjoint $P_{3}$s as shown in Fig.~\ref{fig3} and explained formally below.
The centers of these $P_{3}$s are in $K_{4\pi(\vp{},\cl{})+1}^{\vp{}}$.
For each of these four $P_{3}$s, exactly one endpoint is an arbitrary distinct vertex in $T_{\cl{}}^{\vp{}}$ which is different from the endpoints of the $P_3$s connecting $T_\cl^\vp$ to $Q_\cl^1$; we denote these endpoints as $w_1,w_2,w_3,w_4$.
The other endpoint is in $K_{4\pi(\vp{},\cl{})+2}^{\vp{}}$ if $x_{\vp{}}$ appears positively in $\Gamma_{\cl{}}$ and the other endpoint is in $K_{4\pi(\vp{},\cl{})}^{\vp{}}$ otherwise.
The precise centers and endpoints in the cliques $K_{4\pi(\vp{},\cl{})+2}^{\vp{}}$ or $K_{4\pi(\vp{},\cl{})}^{\vp{}}$ are specified below.
Since these newly introduced $P_3$s use edges that belong to some $P_3$s in ${\mathcal{H}}_{\textsf{var}}$ that were introduced while constructing the variable gadgets, we will remove such $P_3$s in the variable gadget from ${\mathcal{H}}_{\textsf{var}}$, remove their corresponding edges from the graph, and add some new $P_3$s to~${\mathcal{H}}_{\textsf{var}}$ as described below.
As a result, the clique $K_{4\pi(\vp{},\cl{})+1}^\vp{}$ may no longer be fully adjacent to $K_{4\pi(\vp{},\cl{})}^\vp{}$ or $K_{4\pi(\vp{},\cl{})+2}^\vp{}$.
We will however maintain the invariant that each vertex pair between $K_{4\pi(\vp{},\cl{})+1}^\vp{}$ and $K_{4\pi(\vp{},\cl{})}^\vp{}$ or $K_{4\pi(\vp{},\cl{})+2}^\vp{}$ is covered by a $P_3$ in the packing and that all the $P_3$s of ${\mathcal{H}}_{\textsf{var}}$ are pairwise modification-disjoint.

Formally, if $x_\vp{}$ appears positively in $\Gamma_\cl{}$, we denote:
\begin{align*}
v_1 &= v_{4\pi(\vp{},\cl{})+1,0}^\vp{} & v_2 &= v_{4\pi(\vp{},\cl{})+1,1}^\vp{} \\
v_3 &= v_{4\pi(\vp{},\cl{})+2,1}^\vp{} & v_4 &= v_{4\pi(\vp{},\cl{})+2,2}^\vp{} \\
v_5 &= v_{4\pi(\vp{},\cl{}),0}^\vp{} & v_6 &= v_{4\pi(\vp{},\cl{}),1}^\vp{} \\
v_7 &= v_{4\pi(\vp{},\cl{}),3}^\vp{} & v_8 &= v_{4\pi(\vp{},\cl{}),4}^\vp{}.
\end{align*}
If $x_\vp{}$ appears negatively in $\Gamma_\cl{}$, we swap the roles of $K_{4\pi(\vp{},\cl{})}^\vp{}$ and $K_{4\pi(\vp{},\cl{})+2}^\vp{}$, that is:
\begin{align*}
v_1 &= v_{4\pi(\vp{},\cl{})+1,0}^\vp{} & v_2 &= v_{4\pi(\vp{},\cl{})+1,1}^\vp{} \\
v_3 &= v_{4\pi(\vp{},\cl{}),1}^\vp{} & v_4 &= v_{4\pi(\vp{},\cl{}),2}^\vp{} \\
v_5 &= v_{4\pi(\vp{},\cl{})+2,0}^\vp{} & v_6 &= v_{4\pi(\vp{},\cl{})+2,1}^\vp{} \\
v_7 &= v_{4\pi(\vp{},\cl{})+2,3}^\vp{} & v_8 &= v_{4\pi(\vp{},\cl{})+2,4}^\vp{}.
\end{align*}
As shown in Fig.~\ref{fig3}, we remove $P_3$s given by $v_8v_1v_3, v_7v_1v_4, v_6v_2v_3, v_5v_2v_4$ from ${\mathcal{H}}_{\textsf{var}}$ and we remove their corresponding edges from the graph. Then we add the $P_3$s given by $v_5v_6v_2$ and $v_1v_7v_8$ to the graph and to~${\mathcal{H}}_{\textsf{var}}$.
Finally, we connect $T_\cl^\vp$ via $K_{4\pi(\vp{},\cl{})+1}^\vp{}$ by adding the $P_3$s given by $w_1v_1v_3$, $w_2v_2v_4$, $w_3v_2v_3$, and $w_4v_1v_4$ to the graph and to $\mathcal{H}_{\textsf{tra}}$.
Note that, indeed, each vertex pair between $K_{4\pi(\vp{},\cl{})+1}^\vp{}$ and $K_{4\pi(\vp{},\cl{})}^\vp{}$ and between $K_{4\pi(\vp{},\cl{})+1}^\vp{}$ and $K_{4\pi(\vp{},\cl{})+2}^\vp{}$ remains covered by a $P_3$ in the packing after replacing all~$P_3$s.
This finishes the construction of the skeleton of the clause gadgets.

\medskip

The intuitive idea behind the connection to the variable gadget and how it is used in the soundness proof is as follows.
Recall from above that we need to delete at least one of three sets of edges in the solution, namely the edges between $Q_{\cl{}}^{1}$ and $Q_{\cl{}}^{2}$, the edges between $Q_{\cl{}}^{2}$ and $Q_{\cl{}}^{3}$, or the edges between $Q_{\cl{}}^{3}$ and~$Q_{\cl{}}^{4}$.
Assume that the edges between $Q_{\cl{}}^{1}$ and $Q_{\cl{}}^{2}$ are deleted and the variable $x_{\vp{}}$ appears positively in the clause~$\Gamma_{\cl{}}$ as in Fig.~\ref{fig1}.
Since we can modify at most one vertex pair for each of the $P_3$s $P_{\cl{}}^{1}$ and $P_{\cl{}}^{2}$, cliques $T_{\cl{}}^{\vp{}}$ and $Q_{\cl{}}^{1}$ have to be merged in the final cluster graph.
Since $K_{4\pi(\vp{},\cl{})+1}^{\vp{}}$ cannot be merged with $Q_{\cl{}}^{1}$ (there are no edges between $Q_{\cl{}}^{1}$ and $K_{4\pi(\vp{},\cl{})+1}^{\vp{}}$, and no $P_3$s connecting $Q_{\cl{}}^{1}$ and $K_{4\pi(\vp{},\cl{})+1}^{\vp{}}$), we have to separate $T_{\cl{}}^{\vp{}}$ from $K_{4\pi(\vp{},\cl{})+1}^{\vp{}}$.
Then, the $P_{3}$s connecting $T_{\cl{}}^{\vp{}}$ with $K_{4\pi(\vp{},\cl{})+2}^{\vp{}}$ force $K_{4\pi(\vp{},\cl{})+1}^{\vp{}}$ and $K_{4\pi(\vp{},\cl{})+2}^{\vp{}}$ to merge.
This means $x_{\vp{}}$ is true and it satisfies the clause $\Gamma_{\cl{}}$.

The $P_3$s added so far are indeed sufficient to conduct a soundness proof of the above reduction: They ensure that there exists a satisfying assignment to the input formula provided that there exists an appropriate cluster editing set.
However, the completeness is much more difficult: We need to add some more ``padding'' $P_3$s to the packing (and edges to the graph between the cliques that can be potentially merged) to ensure that a satisfying assignment can always be translated into a cluster-editing set.
The goal of the next two sections is to develop a methodology of padding such cliques with $P_3$s in the packing.
The padding will rely on the special structure of $P_3$s that we have established above in the clause gadget and connection between clause and variable gadget.

\subsubsection{Merging model of the clause gadget}\label{sec:merge-model}

In the sections above, we have defined all proto-clusters of the final constructed graph: As we will see in the correctness proof, each clique will be a proto-cluster in the end.
Thus, all solutions will construct a cluster graph whose clusters represent a coarser partition than the partition given by the proto-clusters, or cliques.
What remains is to ensure that the proto-clusters indeed can be merged as required to construct a solution from a satisfying assignment to $\Phi$ in the completeness proof.
To do this, we pad the proto-clusters with $P_3$s (in the graph and packing $\pkg$).
To simplify this task we now divide the set of proto-clusters into five levels $L_0, \ldots, L_4$.
Then, we will go through the levels in increasing order and add padding $P_3$s from proto-clusters of the current level to proto-clusters of all lower~levels if necessary.

There are two issues that we need to deal with when introducing the padding $P_3$s.
For the padding, we will use a number-theoretic tool that we introduce in \cref{sec:implementation} which has the limitation that, when padding a proto-cluster~$D$ with $P_3$s to some sequence~$D_1, \ldots, D_s$ of proto-clusters of lower level, we need to increase the number of vertices in $D$ to be roughly $2 \cdot \sum_{i = 1}^{s} |D_i|$.
Hence, first, we need to make sure that the number of levels is constant since the number of size increases of proto-clusters compounds exponentially with the number of levels.
Second, we aim for the property that each vertex is only in a constant number of $P_3$s in \pkg\ and thus, we need to ensure that the number~$s$ of lower-level proto-clusters and their size is constant.

\begin{figure}[t]
\centering
\begin{tikzpicture}
		\node [nodestyle1,label=above:{\Large 1}] (0) at (-2, 2.5) {};
		\node [nodestyle1,label=below:{\Large 3}] (1) at (-2, 0.5) {};
		\node [nodestyle1,label=above:{\Large 1}](2) at (0.25, 2.5) {};
		\node [nodestyle1,label=below:{\Large 2}] (3) at (0.25, 0.5) {};
		\node [nodestyle1,label=above:{\Large 4}] (4) at (-3.75, 2.5) {};
		\node [nodestyle1,label=left:{\Large 0}] (5) at (-7.75, 1.5) {};
		\node [nodestyle1,label=right:{\Large 0}] (6) at (-5.75, 1.5) {};
		\node [nodestyle1,label=right:{\Large 0}] (7) at (-4.75, -0.25) {};
		\node [nodestyle1,label=right:{\Large 0}] (8) at (-5.75, -2) {};
		\node [nodestyle1,label=left:{\Large 0}] (9) at (-7.75, -2) {};
		\node [nodestyle1,label=left:{\Large 0}] (10) at (-8.75, -0.25) {};
		\node [nodestyle1,label=left:{\Large 0}] (12) at (4, 1.5) {};
		\node [nodestyle1,label=above left:{\Large 0}] (14) at (3.25, 2.75) {};
		\node [nodestyle1,label=left:{\Large 0}](15) at (4, 4) {};
		\node [nodestyle1,label=right:{\Large 0}] (17) at (3.75, -2) {};
		\node [nodestyle1,label=left:{\Large 0}] (20) at (1.5, -3.25) {};
		\node [nodestyle1,label=left:{\Large 0}] (21) at (2.25, -2) {};
		\node [nodestyle1,label=right:{\Large 4}] (23) at (1.5, -0.5) {};
		\node [nodestyle1,label=above:{\Large 4}] (24) at (1.75, 3) {};

		\node [nodestyle1,label={[label distance=-2mm]120:\Large 0}] (25) at (-6.75, 1.75) {};
		\node [nodestyle1,label=right:{\Large 0}] (26) at (-5, 0.75) {};
		\node [nodestyle1,label=right:{\Large 0}] (27) at (-5, -1.25) {};
		\node [nodestyle1,label={[label distance=-2mm]-45:\Large 0}] (28) at (-6.75, -2.25) {};
		\node [nodestyle1,label=left:{\Large 0}] (29) at (-8.5, -1.25) {};
		\node [nodestyle1,label=left:{\Large 0}] (30) at (-8.5, 0.75) {};

		\draw  [middlearrow={latex}] (3.center) to (2.center);
		\draw [middlearrow={latex}](1.center) to (3.center);
		\draw  [middlearrow={latex}] (1.center) to (2.center);
		\draw [middlearrow={latex}] (3.center) to (0.center);
		\draw [middlearrow={latex}] (1.center) to (0.center);
		\draw  (21.center) to (17.center);
		\draw  (21.center) to (20.center);
		\draw  (14.center) to (15.center);
		\draw  (14.center) to (12.center);
		\draw [middlearrow={latex}] (4.center) to (6.center);

		\draw  (6.center) to (25.center);
		\draw  (6.center) to (26.center);
		\draw [middlearrow={latex}] (4.center) to (25.center);
		\draw [middlearrow={latex}] (4.center) to (26.center);
		\draw  (7.center) to (26.center);
		\draw  (7.center) to (27.center);
		\draw  (8.center) to (27.center);
		\draw  (8.center) to (28.center);
		\draw  (9.center) to (28.center);
		\draw  (9.center) to (29.center);
		\draw  (10.center) to (29.center);
		\draw  (10.center) to (30.center);
		\draw  (5.center) to (30.center);
		\draw  (5.center) to (25.center);

		\draw [middlearrow={latex}] (4.center) to (0.center);
		\draw [middlearrow={latex}] (24.center) to (2.center);
		\draw [middlearrow={latex}] (24.center) to (14.center);
		\draw [middlearrow={latex}] (24.center) to (15.center);
		\draw [middlearrow={latex}] (24.center) to (12.center);
		\draw [middlearrow={latex}] (23.center) to (3.center);
		\draw [middlearrow={latex}] (23.center) to (2.center);
		\draw [middlearrow={latex}] (23.center) to (20.center);
		\draw [middlearrow={latex}] (23.center) to (21.center);
		\draw [middlearrow={latex}] (23.center) to (17.center);

		\draw [dotted,line width=0.4mm] (17) to [bend left=60](20);
        \draw [dotted,line width=0.4mm] (15) to [bend left=60](12);
\end{tikzpicture}
\caption{Merging model of a clause $\Gamma_{\cl{}}=(x_{\vp{}}\lor \neg x_{\vq{}}\lor \neg x_{\vr{}})$. The number $i \in \{0, 1, 2, 3, 4\}$ beside a vertex $v$ denotes that $v\in L_i$. The placement of vertices corresponds to the placement of the cliques in Fig.~\ref{fig:skeleton}. For example, the two vertices of level 1 on the top correspond to $Q_\cl^1$ and $Q_\cl^4$. We assume that $m_{\vp{}}=3$.} \label{fig2}\label{fig:merging-model}
\end{figure}

To achieve the above goals, we introduce an auxiliary graph $H$, the \emph{merging model}, which will further guide the padding process.
The merging model has as vertices the cliques that were introduced before and an edge between two cliques if we want it to be possible that they are merged by a solution.
Formally,

\begin{align*}
  V(H) & := \{ K_j^i \mid i = 0, 1, \ldots, n - 1 \text{ and } j = 0, 1, \ldots, 4m_i - 1\} \cup {} \\
       & \phantom{{}:={}} \{ Q^1_\cl, Q^2_\cl, Q^3_\cl, Q^4_\cl \mid \cl = 0, 1, \ldots, m - 1\} \cup {} \\
       & \phantom{{}:={}} \{ T^\vp_s \mid \text{variable $x_\vp{}$ occurs in clause } \Gamma_s \},
\end{align*}
and the edge set, $E(H)$, is defined as follows.
See also Fig.~\ref{fig2}.
First, it shall be possible to merge the cliques in the variable gadget in a cyclic fashion,\footnote{Indeed, we have already ensured that this is possible. The edges introduced in the first step purely serve to reinforce the intuition of the merging model.} that is, we add
\[\{ \{K_j^i, K_{j + 1}^i\} \mid i = 0, 1, \ldots, n - 1 \text{ and } j =
  0, 1, \ldots, 4m_i - 1\}\] to $E(H)$.
Second, it shall be possible to merge transferring cliques of a clause gadget to any of the relevant cliques of the associated variable gadget, that is, we add to $E(H)$ the set
\[\{ \{T^i_\cl, K_{4\pi(i,\cl{})}^{i}\}, \{T^i_\cl, K_{4\pi(i,\cl{})+1}^{i}\}, \{T^i_\cl, K_{4\pi(i,\cl{})+2}^{i}\} \mid \text{variable $x_i$ occurs in clause } \Gamma_\cl\}. \]
Third, it shall be possible to merge subsets of $\{Q^1_\cl, Q^2_\cl, Q^3_\cl, Q^4_\cl\}$, and hence we add to $E(H)$ the set
\[ \{ \{Q^1_\cl, Q^2_\cl\}, \{Q^1_\cl, Q^3_\cl\}, \{Q^2_\cl, Q^3_\cl\}, \{Q^2_\cl, Q^4_\cl\}, \{Q^3_\cl, Q^4_\cl\} \mid \cl = 0, 1, \ldots, m - 1\}. \]
Finally, it shall be possible to merge the transferring cliques to subsets of $\{Q^1_\cl, Q^2_\cl, Q^3_\cl, Q^4_\cl\}$. Hence, we add to $E(H)$ the set
\begin{align*}
  &\{ \{T^i_\cl, Q^k_\cl\} \mid \text{if variable $x_i$ occurs in $\Gamma_\cl$ and $T^i_\cl$ is adjacent in $G$ to $Q^k_\cl$ with } k \in \{1, 4\}\} \cup {} \\
  & \{ \{T^i_\cl, Q^3_\cl\}, \{T^i_\cl, Q^4_\cl\} \mid \text{if variable $x_i$ occurs in $\Gamma_\cl$ and $T^i_\cl$ is adjacent in $G$ to } Q^3_\cl\}.
\end{align*}
Note that this construction is slightly asymmetric (see Fig.~\ref{fig:merging-model}).

Now we define the levels~$L_0$ to $L_4$ such that orienting the edges in $H$ from higher to lower level gives an acyclic orientation when ignoring the edges in level~$L_0$.
\begin{itemize}
\item $L_0$ contains all cliques in variable gadgets.
\item $L_1$ contains $Q^1_\cl$ and $Q^4_\cl$ for each $\cl = 0, \ldots, m - 1$.
\item $L_2$ contains $Q^3_\cl$ for each $\cl = 0, \ldots, m - 1$.
\item $L_3$ contains $Q^2_\cl$ for each $\cl = 0, \ldots, m - 1$.
\item $L_4$ contains all transferring cliques.
\end{itemize}
We now orient all edges in $H$ from higher-level vertices to lower-level vertices.
Edges in level $L_{0}$ remain undirected.
Observe that, apart from edges in $L_0$, all edges in~$H$ are between vertices of different levels and, indeed, ignoring edges in $L_0$, there are no cycles in $G$ when orienting the edges from higher to lower level.
In the following section, we will look at each clique~$R$ in levels $L_1$ and higher, and add $P_3$s to the packing $\pkg$ so as to cover all vertex pairs containing a vertex of $R$ and an out-neighbor of $R$ in~$H$.

\subsubsection{Implementation of the clause gadget}\label{sec:implementation}
In this section, we first introduce a number-theoretical construction (\cref{lem:padding}) that serves as a basic building block for ``padding'' $P_3$s in the packing.
Then we use this construction to perform the actual padding of $P_3$s.

The abstract process of padding $P_3$s works as follows.
It takes as input a clique~$R$ in~$H$ (represented by $W$ in the below Lemma~\ref{lem:padding}), and a set of cliques that are out-neighbors of $R$ in~$H$ (represented by~$V$).
Furthermore, it receives a set of vertex pairs between $R$ and its out-neighbors that have previously been covered (represented by $F$).
The goal is then to find a packing of $P_3$s that cover all vertex pairs \emph{except} the previously covered pairs.
The previously covered vertex pairs have some special structure that we carefully selected so as to make covering of all remaining vertex pairs possible in a general way:
The construction so far was carried out in such a way that the connected components induced by previously covered vertex pairs are $P_3$s or $C_8$s.

In Lemma~\ref{lem:padding} we will indeed pack triangles instead of $P_3$s because this is more convenient in the proof.
We will replace the triangles by $P_3$s afterwards: Recall the intuition from \cref{sec:intuition} that $P_3$s in the packing~\pkg\ which have exactly one endpoint in one clique~$T$ and their remaining two vertices in another clique~$R$ can accommodate both merging $R$ and $T$ or separating~$R$ and~$T$ without excess edits.
Hence, we will replace the triangles by such $P_3$s.
Recall that we aim for each clique to be a proto-cluster in the final construction, that is, each clique contains a spanning tree of edges which are not contained in $P_3$s in \pkg.
Since putting the above kind of $P_3$s into the packing \pkg\ allows in principle to delete edges within~$R$, we need to ensure that $R$ remains a proto-cluster.
We achieve this via the connectedness property in \cref{lem:padding}.

\paragraph{Number-theoretic padding tool.}

\begin{figure}
  \mbox{}\hfill
\begin{tikzpicture}

  \node [style=varGadget,label=left:{$w_{i}$}] (0) at (-7, 2) {};
  \node [style=varGadget,label={[label distance=-0.5mm]90:$v_{i}$}] (1) at (-5, 2) {};
  \node [style=varGadget,label=right:{$w'_{i}$}] (2) at (-3, 2) {};
  \node [style=varGadget,label=left:{$w_{i+1}$}] (3) at (-7, 1) {};
  \node [style=varGadget,label={[label distance=-0.5mm]90:$v_{i+1}$}] (4) at (-5, 1) {};
  \node [style=varGadget,label=right:{$w'_{i+1}$}] (5) at (-3, 1) {};
  \node [style=varGadget,label=left:{$w_{i+2}$}] (6) at (-7, -0) {};
  \node [style=varGadget,label={[label distance=-0.5mm]90:$v_{i+2}$}] (7) at (-5, -0) {};
  \node [style=varGadget,label=right:{$w'_{i+2}$}] (8) at (-3, -0) {};
  \node [style=varGadget,label=left:{$w_{i+3}$}] (9) at (-7, -1) {};
  \node [style=varGadget,label={[label distance=-0.5mm]below:$v_{i+3}$}] (10) at (-5, -1) {};
  \node [style=varGadget,label=right:{$w'_{i+3}$}] (11) at (-3, -1) {};

  \draw [style=C8Edge1] (1) to (3);
  \draw [style=C8Edge1] (1) to (8);
  \draw [style=C8Edge1] (3) to (4);
  \draw [style=C8Edge1] (4) to (5);
  \draw [style=C8Edge1] (5) to (10);
  \draw [style=C8Edge1] (10) to (6);
  \draw [style=C8Edge1] (6) to (7);
  \draw [style=C8Edge1] (7) to (8);
\end{tikzpicture}
\hfill
\begin{tikzpicture}
  \input{tikz-hypergraph}

  \foreach \vert / \label in
  {0/$w_{i + 1}$, 1/$v_i$, 2/$w'_{i + 2}$, 3/$v_{i + 2}$, 4/$w_{i + 2}$, 5/$v_{i + 3}$, 6/$w'_{i + 1}$, 7/$v_{i + 1}$}{
    \pgfmathparse{360 - 360/8 * \vert}
    \node [style=varGadget,label={[label distance=2mm]\pgfmathresult:{\label}}] (\vert) at (\pgfmathresult:1.5cm) {};
  }

  \foreach \vert [remember=\vert as \lastvert (initially 7)] in {0, 1, ..., 7}{
    \draw [style=C8Edge1] (\vert) to (\lastvert);
  }

  \newcommand\radius{3mm}

  \begin{pgfonlayer}{bg}
  \draw [fill = red, fill opacity = .25] \hedgeiii{0}{1}{2}{\radius};
  \draw [fill = blue, fill opacity = .25] \hedgeiii{2}{3}{4}{\radius};
  \draw [fill = green, fill opacity = .25] \hedgeiii{4}{5}{6}{\radius};
  \draw [fill = brown, fill opacity = .25] \hedgeiii{6}{7}{0}{\radius};
\end{pgfonlayer}
\end{tikzpicture}
\hfill
\mbox{}

\caption{Left: The labels of a $C_{8}$ in $(V\cup W,F)$. Right: The triangles in $\tau^2_F$ covering a $C_8$.}
\label{fig4}
\end{figure}

\begin{lemma} \label{lemma1}\label{lem:padding}
  Let $p$ be a prime number with $p \geq 2$. Let $B=(V,W,E)$ be a complete bipartite graph such that $|V|=p$ and $|W|=2p$. Let $F\subseteq E$ be a set of edges such that each connected component of $(V\cup W, F)$ is a either a singleton, a $P_3$ with a center in $V$, or a~$C_8$. Then there exists an edge-disjoint triangle packing $\tau$ in $(V\cup W, (E\setminus F) \cup \binom{W}{2})$ which covers $E\setminus F$ such that the graph $(W, \binom{W}{2}\setminus E(\tau))$ is connected.
  Moreover, each vertex~\(v \in V \cup W\) is in at most \(p\) triangles of \(\tau\), it is in at most \(p - 1\) triangles if \(v\) is in a connected component of $(V\cup W, F)$ that is a \(P_3\), and in at most \(p - 2\) triangles if \(v\) is in connected component of $(V\cup W, F)$ that is a~\(C_8\).
\end{lemma}
\begin{proof}
First, we divide $W$ into two parts $W_{1}$ and $W_{2}$ of equal sizes such that if two vertices $w,w'\in W$ are connected to the same vertex $v\in V$ by edges in $F$, then $w$ and $w'$ are in different parts. Note that this is easy for a connected component of $(V\cup W, F)$ if it is a $P_3$. For a connected component of $(V\cup W, F)$ which is a $C_8$, this is also doable as shown in Fig.~\ref{fig4}, where $w_{i},w_{i+1},w_{i+2},w_{i+3}$ belong to $W_{1}$, $w'_{i},w'_{i+1},w'_{i+2},w'_{i+3}$ belong to $W_{2}$, and $v_{i},v_{i+1},v_{i+2},v_{i+3}$ belong to $V$.

We now label the vertices by elements from the finite field $\mathbb{F}_p$ of size $p$ (recall that $\mathbb{F}_p$ consists of the elements $\{0, 1, \ldots,p-1\}$ with addition and multiplication modulo $p$).
To each vertex $v \in V$, each vertex $w \in W_1$, and each vertex $w' \in W_2$, we will assign a unique label $v_i$, $w_j$, and $w'_k$, respectively, with $i, j, k \in \mathbb{F}_p$.
In other words, we construct three bijections that map $\mathbb{F}_p$ to $V$, $W_1$, and $W_2$, respectively.

First, we label the vertices from the connected components of $(V\cup W, F)$ (and some singleton vertices) by going through the connected components one-by-one.
For each yet-unlabeled connected component of $(V\cup W, F)$ that is a $P_3$ given by $wvw'$ such that $v\in V, w\in W_1, w'\in W_2$, we label vertex $w$ as $w_j$, vertex $v$ as $v_j$ and vertex $w'$ as $w'_j$ for the smallest $j$ from $\mathbb{F}_p$ which is not yet used in the labeling of vertices of~$V$.
For each yet-unlabeled connected component~$C$ in $(V\cup W, F)$ that is a $C_8$ we proceed as follows.
By the way we have divided vertices from $W$ into $W_1$ and $W_2$, we can assign, to each such connected component~$C$, four vertices which have degree~zero in~$(V\cup W, F)$: two in $W_1$ and two in~$W_2$; see also Fig.~\ref{fig4}.
We thus label the vertices in $C$ and the four degree-zero vertices assigned to~$C$ as in Fig.~\ref{fig4}, for the smallest integer $i$ from $\mathbb{F}_p$ such that $i,i+1,i+2$ and $i+3$ are not used in the labeling of vertices of $V$.

Second, we label the remaining unlabeled vertices that are not in the connected components of $(V\cup W, F)$.
For an unlabeled vertex $w\in W_{1}$, label it as $w_{k}$ for an arbitrary integer $k$ from $\mathbb{F}_p$ which is not used in the labeling of vertices in $W_{1}$.
Similarly, for an unlabeled vertex $v\in V$, we label it as $v_{h}$ for an arbitrary integer~$h$ from $\mathbb{F}_p$ which is not used in the labeling of vertices in $V$ and for an unlabeled vertex $w'\in W_{2}$, we label it as~$w'_{s}$ for an arbitrary integer $s$ from $\mathbb{F}_p$ which is not used in the labeling of vertices in $W_{2}$.
After the labeling, the vertices in $V,W_1$ and $W_2$ are $v_1,\ldots,v_{p-1}$, $w_1,\ldots,w_{p-1}$ and $w'_1,\ldots,w'_{p-1}$, respectively.

We now proceed to constructing the packing~$\tau$.
First, let
\begin{align*}
\tau_{\textsf{all}} & =\left\{uvw \;\middle|\; uvw\text{ is a triangle in } \left(V\cup W, E\cup \binom{W}{2}\right)\text{ such that }u\in V, v\in W_1, w\in W_2\right\} \text{, and}\\
\tau_{\textsf{cover}}& =\{v_i w_j w'_k \in \tau_{\textsf{all}} \mid i, j, k \in \mathbb{F}_p \text{ and } j-i=k-j \text{ over }\mathbb{F}_p \}.
\end{align*}
In the following, for any triangle packing $\tau$, by $E(\tau)$ we will denote the union of the edge sets of the triangles in~$\tau$.

We claim that the triangles in $\tau_{\textsf{cover}}$ are edge-disjoint and cover all edges of $E$.
Consider an arbitrary edge $v_i w_j\in E$ between $V$ and $W_1$ for $i,j\in \mathbb{F}_p$.
According to the definition of $\tau_{\textsf{cover}}$, each triangle $v_iw_jw'_x \in \tau_{\textsf{cover}}$ that covers edge $v_i w_j$
satisfies $x = 2j-i$ (over $\mathbb{F}_p$).
Since $\mathbb{F}_p$ is a field, there is thus exactly one such triangle.
Similarly, each edge $v_hw'_k\in E$ between $V$ and $W_1$ for some $h,k\in \mathbb{F}_p$ is covered by the unique triangle $v_hw_{(h+k)\cdot 2^{-1}}w'_k \in \tau_{\textsf{cover}}$.
Finally, each edge $w_sw'_t$ between $W_1$ and $W_2$ is covered by the unique triangle $v_{2s-t}w_sw'_t \in \tau_{\textsf{cover}}$.
Thus the claim holds.

Let
\begin{equation*}
  \tau_{F}^{1} =\{v_{h}w_{h}w'_{h}\in \tau_{\textsf{all}} \mid \text{vertices }w_{h}, v_{h}, w'_{h} \text{ induce a $P_3$ in } (V\cup W, F) \}\text{, and}
\end{equation*}
\begin{multline*}
  \tau_{F}^{2} =\{v_{h}w_{h+1}w'_{h+2},
  \ v_{h+1}w_{h+1}w'_{h+1},
  \ v_{h+2}w_{h+2}w'_{h+2},
  \ v_{h+3}w_{h+2}w'_{h+1} \in \tau_{\textsf{all}} \mid\\
  \text{vertices }v_{h},w'_{h+2},v_{h+2},w_{h+2},v_{h+3},w'_{h+1},v_{h+1},w_{h+1} \text{ induce a $C_{8}$ in $(V\cup W, F)$} \}.
\end{multline*}
Observe that $\tau_{F}^{1},\tau_{F}^{2}\subseteq \tau_{\textsf{cover}}$.
For example, if we put $v_{h + 3}w_{h + 2}w'_{h + 1} = v_{i}w_{j}w'_{k}$, then it follows that $j - i = p - 1 = k - j$ over $\mathbb{F}_p$, that is, $v_{h + 3}w_{h + 2}w'_{h + 1}$ satisfies the conditions in the definition of $\tau_{\textsf{cover}}$.
Moreover, $\tau_{F}^{1}\cup \tau_{F}^{2}$ covers all edges of~$F$.
Furthermore, each edge in the edge set $E(\tau_{F}^{1}\cup \tau_{F}^{2})$ of $\tau_{F}^{1}\cup \tau_{F}^{2}$ is either in $F$ or between $W_1$ and~$W_2$.
(See also Fig.~\ref{fig4}.)
Thus, $E\setminus F$ has an empty intersection with $E(\tau_{F}^{1}\cup \tau_{F}^{2})$.
Let $\tau=\tau_{\textsf{cover}}\setminus (\tau_{F}^{1}\cup \tau_{F}^{2})$.
It follows that $\tau$ covers all edges of $E\setminus F$.
It remains only to show that $\tau$ satisfies the connectedness condition.
Since $\tau_{\textsf{cover}}$ does not cover any edge of $\binom{W_1}{2}$ or $\binom{W_2}{2}$, it follows that $(W_1, \binom{W_1}{2}\setminus E(\tau))$ and $(W_2, \binom{W_2}{2}\setminus E(\tau))$ are cliques.
Now observe that $\tau_{F}^1 \cup \tau_{F}^2$ contains at most $|V| = p$ edges of $\binom{W}{2}$, while $W_1 \times W_2$ is of size $p^2 > p$.
Thus in the graph $(W, \binom{W}{2}\setminus E(\tau))$ there is at least one edge $\{w_1,w_2\}$ such that $w_1\in W_1$ and $w_2\in W_2$.
As a result, $(W, \binom{W}{2}\setminus E(\tau))$ is connected.
Finally, observe that each vertex \(v \in V \cup W\) is in at most \(p\) triangles in \(\tau_{\textsf{cover}}\).
If \(v\) is in a \(P_3\) of $(V\cup W, F)$, then at least one of these triangles is removed from \(\tau_{\textsf{cover}}\) to obtain \(\tau\).
If \(v\) is in a \(C_8\) of $(V\cup W, F)$, then at least two of the triangles in \(\tau_{\textsf{cover}}\) that contain \(v\) are removed to obtain \(\tau\).
This concludes the proof.
\end{proof}

The following corollary is slightly easier to apply than \cref{lem:padding}.

\begin{corollary} \label{cor1}\label{cor:padding}
  Let $p$ be a prime and let $B=(V,W,E)$ be a complete bipartite graph with $|V| \leq p,|W|=2p$. Let $F\subseteq E$ be a nonempty set of edges such that every connected component of $(V\cup W, F)$ is a either a $P_3$ with a center in $V$ or a $C_8$. Then there exists an edge-disjoint triangle packing $\tau$ in $(V\cup W, (E\setminus F) \cup \binom{W}{2})$ which covers $E\setminus F$ such that $(W, \binom{W}{2}\setminus E(\tau))$ is connected.
  Each vertex~\(v \in V \cup W\) is in at most \(p\) triangles of \(\tau\), at most \(p - 1\) if \(v\) is in a connected component of $(V\cup W, F)$ that is a \(P_3\), and at most \(p - 2\) if \(v\) is in connected component of $(V\cup W, F)$ that is a~\(C_8\).
\end{corollary}
\begin{proof}
Add extra $p - |V|$ dummy vertices to $V$, obtaining a complete bipartite graph $B' = (V', W, E)$, apply Lemma~\ref{lemma1} to $B'$, $p$, and $F$, obtaining a packing $\tau'$, and return a sub-packing $\tau \subseteq \tau'$ containing only triangles
with vertices in $B$. Since every triangle in $\tau'$ contains exactly one vertex of $V'$, $\tau$ satisfies all the required properties.
\end{proof}

\paragraph{Concluding the construction.}
Equipped with Lemma~\ref{lemma1} and \cref{cor:padding}, we can finish the construction of the clause gadgets and indeed the whole instance $(G, \pkg, 0)$ of \pCEA.
We now specify the exact size of each clique introduced above and add padding $P_3$s to $G$ and \pkg\ so as to cover all vertex pairs between cliques that are adjacent in the merging model~$H$.
Put initially the set ${\mathcal{H}}_{\textsf{pad}}$ of \emph{padding $P_3$s} to be ${\mathcal{H}}_{\textsf{pad}} =\emptyset$.
We start with levels~$0$ and~$1$.
We do not change the sizes of any clique on level~$0$.
That is, as shown in the variable gadget, there are five vertices in every clique of level~$0$.
Besides, we set the size of every clique of level~$1$ to be one.
Note that no cliques of levels $0$ and $1$ are adjacent in the merging model~$H$, that is, no two of them need to be merged in the solution.
Hence, it is not necessary to add padding $P_3$s within these levels.

Now we turn each level~$i$, $i \geq 2$, in order of increasing~$i$.
For each clique $Q$ of level~$i$, we apply Corollary~\ref{cor1} in the following scenario.
Let $V$ be the union of all cliques of levels $j < i$ that are out-neighbors of $Q$ in the merging model~$H$.
Let $p$ be the smallest prime with $p \geq |V|$ and $2p \geq |Q|$.
Introduce $2p - |Q|$ new vertices, put them into~$Q$, and make~$Q$ a clique.
Put $W = Q$ and let $E=\{\{u,v\} \mid u\in V, v\in W\}$.

We claim that Corollary~\ref{cor1} is applicable to $p$, graph $B = (V, W, E)$, and $F$.
To see this, we need to show that each connected component in $(V \cup W, F)$ is either a $P_3$ with center in~$V$ or a $C_8$.
Indeed, if $Q$ is not a transferring clique, that is, $Q = Q^j_\cl$ for some $\cl \in \{0, 1, \ldots, m - 1\}$ and $j \in \{1, 2, 3, 4\}$, then each connected component in $(V \cup W, F)$ consists of two edges of two different transferring $P_3$s with the same center in~$V$, as claimed (see also Fig.~\ref{fig1}).
If $Q$ is a transferring clique, then each connected component of $(V \cup W, F)$ consists either of two edges of two different transferring $P_3$s with the same center in some $Q^j_\cl \subseteq V$ for some $j \in \{1, 3, 4\}$, or of some vertex pairs of transferring $P_3$s between $Q$ and the cliques of a variable gadget.
In the first case, the claim clearly holds.
In the second case, observe that the edges and non-edges between $V$ and $W$ in the transferring $P_3$s are each incident with one of $w_1, w_2, w_3, w_4$ and one of $v_1, v_2, v_3, v_4$ as defined when connecting variable and clause gadgets.
These edges and non-edges indeed induce a $C_8$ given by $v_1 w_1 v_3 w_3 v_2 w_2 v_4 w_4 v_1$ (see also Fig.~\ref{fig:connection}).
Thus, Corollary~\ref{cor1} is applicable.

Corollary~\ref{cor1} gives us an edge-disjoint triangle packing $\tau$ in $(V\cup W, (E\setminus F) \cup \binom{W}{2})$ which covers all edges of $E\setminus F$ such that $(W, \binom{W}{2}\setminus E(\tau))$ is connected.
Note that every triangle $vw_{1}w_{2}\in \tau$ has one vertex $v\in V$ and two vertices $w_1,w_2\in W$.
For every triangle $vw_{1}w_{2}\in \tau$, we add a $P_3$ to $G$ by using exactly two edges of the triangle in $G$; more precisely, we put $\{v,w_1\},\{w_1,w_2\}\in E(G),vw_2\notin E(G)$, and then add the $P_3$ of $G$ given by $v w_1 w_2$ into ${\mathcal{H}}_{\textsf{pad}}$.
Finally, let $\pkg= \pkg_{\textsf{var}} \cup \pkg_{\textsf{tra}} \cup \pkg_{\textsf{pad}}$.
Note that $\pkg$ is a modification-disjoint packing of $P_3$s: This is by construction for $\pkg_{\textsf{var}} \cup \pkg_{\textsf{tra}}$ and, by Corollary~\ref{cor1}, no $P_3$ in $\pkg_{\textsf{pad}}$ shares a vertex pair with any $P_3$ in $\pkg_{\textsf{var}} \cup \pkg_{\textsf{tra}}$.
This concludes the construction of the \pCEA\ instance $(G, \pkg, 0)$.

To see that the construction takes polynomial time and to see that indeed each vertex is in some constant number of $P_3$s in \pkg, let us now derive the precise sizes of each clique in the construction.
Recall that the cliques on level $0$ are exactly those in the variable gadgets, and these have exactly five vertices each.
The cliques on level~$1$ are $Q_{\cl}^{1}$ and $Q_\cl^4$ for $\cl \in \{0, 1, \ldots, m - 1\}$, and they have $1$ vertex each.
On level~$2$ we have the cliques~$Q_{\cl}^{3}$, $\cl \in \{0, 1, \ldots, m - 1\}$, and since the only out-neighbor in $H$ of $Q_{\cl}^{3}$ is $Q_{\cl}^{4}$, our procedure sets $p=2$ and thus $Q_\cl^3$ has $4$ vertices.
On level $3$ there are the cliques $Q_{\cl}^{2}$, $\cl \in \{0, 1, \ldots, m - 1\}$, and we set $p = 7$ as $|Q_{\cl}^{1}\cup Q_{\cl}^{3}\cup Q_{\cl}^{4}|=6$.
Thus clique $Q_{\cl}^{2}$ has $14$ vertices.
For the clique $T_{\cl}^{\vp{}}$, we set $p = 17$ as $|Q_{\cl}^{1}\cup K_{4\pi(\vp{},\cl)}^{\vp{}}\cup K_{4\pi(\vp{},\cl)+1}^{\vp{}}\cup K_{4\pi(\vp{},\cl)+2}^{\vp{}}| = 16$.
So the clique $T_{\cl}^{\vp{}}$ has $2 \cdot 17=34$ vertices.
Similarly, $T_\cl^\vr{}$ has $34$ vertices as well.
For the clique $T_{\cl}^{\vq{}}$, we set $p = 23$, as $|Q_{\cl}^{3}\cup Q_{\cl}^{4}\cup K_{4\pi(\vq{},\cl)}^{\vq{}}\cup K_{4\pi(\vq{},\cl)+1}^{\vq{}}\cup K_{4\pi(\vq{},\cl)+2}^{\vq{}}| = 20$.
Thus $T_{\cl}^{\vq{}}$ is a clique of size $2 \cdot 23=46$.
By the bounds on the number of triangles in the packing, each vertex is in at most 23 $P_3$s of \pkg.
It also follows that the construction takes overall polynomial time.

\subsection{Correctness}\label{sec:correctness}

We now prove the correctness of the reduction given in \cref{sec:construction}

\subsubsection{Completeness}\label{sec:completeness}
Now we show how to translate a satisfying assignment of $\Phi$ into a cluster editing set of size $|\mathcal{H}|$ for the constructed instance.
\begin{lemma}
  If the input formula $\Phi$ is satisfiable, then the constructed instance $(G,{\mathcal{H}}, \ell = 0)$ is a YES-instance.
\end{lemma}
\begin{proof}
  Assume that there is a satisfying assignment $\alpha$ for the formula $\Phi$.
  Recall that $n$ is the number of variables of~$\Phi$ and $m$ is the number of clauses of $\Phi$.
  Instead of building the solution directly, we build a partition~\P\ of $V(G)$ into clusters.
  Then, we argue that the number of edges between clusters and the number of non-edges inside clusters is at most $|\pkg|$.
  Thus, the partition~\P\ will induce a solution with the required number of edge edits.

  Recall that $H$ denotes the merging model of our hardness construction.
  The basic building blocks of our vertex partition~\P\ are the cliques in $G$ that correspond to the vertices of $V(H)$.
  We will never separate such a clique during building \P, that is, \P\ corresponds to a partition of~$V(H)$.
  For simplicity, we will slightly abuse notation and indeed also treat $\P$ as a partition of $V(H)$.
  We build \P\ by taking initially $\P = V(H)$ and then successively \emph{merging} parts of \P, which means to take the parts out of \P\ and replace them by their union.
  Each vertex of $H$ is a clique of $G$, so has no non-edges in $G$.
  Thus, below it suffices to consider edges and non-edges between pairs of cliques corresponding to vertices in $V(H)$ to determine the number of edits in the solution corresponding to \P.

  We start with the variable gadgets.
  Consider each variable $x_i$, $i=0,1,\ldots,n-1$.
  Call a pair of cliques $K_{j}^{i}$, $K_{j + 1}^{i}$ in $x_i$'s variable gadget \emph{even} if $j$ is even and \emph{odd} otherwise (indices are taken modulo $4m_i$).
  If $\alpha(x_i) = \true$, then merge each odd pair.
  If $\alpha(x_i) = \false$, then merge each even pair.
  We will not merge any further pair of cliques contained in variable gadgets.

  Now consider each clause $\Gamma_\cl$, $\cl=0,\ldots,m-1$, in some arbitrary order.
  Let $x_{\vp{}}$, $x_{\vq{}}$, and $x_{\vr{}}$ be the variables in $\Gamma_\cl$.
  We use the same notation as when defining the clause gadgets.
  See Fig.~\ref{fig1} for the skeleton of the clause gadget of $\Gamma_\cl$, up to variables appearing positively instead of negatively or vice versa.
  We choose an arbitrary variable that satisfies $\Gamma_\cl$.
  The basic idea is to separate (that is, to not merge) the transferring clique from the the cliques in the satisfying variable's gadget by deleting some edges of the transferring $P_3$s.
  This will induce at most one edit for each transferring $P_3$ since the remaining edge in a transferring $P_3$ will be part of a cluster in \P.
  Then we cut from the clause gadget all transferring cliques belonging to variables that have not been chosen.
  Since we do not spend edits inside of transferring $P_3$s in this way, this allows us to merge the transferring cliques to the variable gadgets regardless of whether the variable was set to \true\ or \false.

  Formally, we perform the following merges in \P.
  \begin{description}
  \item[\rm If we have chosen $x_{\vp{}}$ from the variables satisfying the clause $\Gamma_\cl$:]\mbox\\
    \begin{itemize}
    \item Merge $T_{\cl}^{\vp{}}$ with $Q_\cl^1$.
    \item Merge the cliques $Q_{\cl}^{2}, Q_{\cl}^{3}$ and $Q_{\cl}^{4}$.
    \end{itemize}
  \item[\rm If we have chosen $x_\vq$:]\mbox\\
    \begin{itemize}
    \item Merge the cliques $Q_{\cl}^{1}, Q_{\cl}^{2}$.
    \item Merge the cliques $T_\cl^\vq$, $Q_\cl^3$, and $Q_{\cl}^{4}$.
    \end{itemize}
  \item[\rm If we have chosen $x_\vr$:]\mbox\\
    \begin{itemize}
    \item Merge $T_{\cl}^{\vr{}}$ with $Q_\cl^4$.
    \item Merge the cliques $Q_{\cl}^{1}$, $Q_{\cl}^{2}$ and $Q_{\cl}^{3}$.
    \end{itemize}
  \end{description}
  Finally, let $\beta \in \{\vp, \vq, \vr\}$ be the index of the chosen variable that satisfies $\Gamma_\cl$.
  For both $\gamma \in \{\vp, \vq, \vr\} \setminus \{\beta\}$ do the following.
  If $\alpha(x_{\gamma})=\true$, then merge $T_{\cl}^{\gamma}$ with the part of \P\ consisting of $K_{4\pi(\gamma,\cl)+1}^{\gamma}$ and $K_{4\pi(\gamma,\cl)+2}^{\gamma}$.
  If $\alpha(x_{\gamma})=\false$, then merge $T_\cl^\gamma$ with the part of \P\ consisting of $K_{4\pi(\gamma,\cl)+1}^{\gamma}$ and $K_{4\pi(\gamma,\cl)}^{\gamma}$.
  This concludes the definition of the vertex partition~\P.
  Let us denote the corresponding cluster editing set by~$S$.
  That is, $S$ contains all edges in~$G$ between parts of \P\ and all non-edges within parts of \P.

  We claim that (c1) each edit in $S$ is contained in a $P_3$ of \pkg\ and (c2) every $P_3$ of \pkg\ is edited at most once by~$S$.
  Note that the claim implies that $S$ is a solution to $(G, \pkg, 0)$.
  We first prove part (c1) of the claim.
  Note that each edit in $S$ is between two cliques in $V(H)$.
  There are three types of edits in \pkg: within a variable gadget, between a clause and a variable gadget, and within a clause gadget.

  Consider first the edits contained in the variable gadget of an arbitrary variable~$x_i$.
  Observe that each such edit is contained in an odd or an even pair of $x$'s gadget.
  Such an edit is contained in a $P_3$ in \pkg, because, by construction of the variable gadgets, all edges and non-edges between the cliques of an odd or an even pair are covered by $P_3$s in \pkg.

  For the edits in $S$ which are not contained in variable gadgets, observe that between each pair of cliques in a single level $L_s$, $s > 0$, there are no edges in~$G$.
  Whenever we merge two or more parts during the construction of \P, we either merge a clique on level $L_4$ to two cliques on level $L_0$ or we merge cliques on pairwise different positive levels.
  Hence, each edit $e \in S$ which is not in a variable gadget is between two cliques on different levels.
  Moreover, observe that the cliques containing the endpoints of $e$ are adjacent in $V(H)$.
  Thus, by the way we have defined $\pkg_{\textsf{pad}}$ via \cref{cor:padding}, there is a $P_3$ in $\pkg_{\textsf{pad}}$ containing $e$.
  We have thus shown that claim (c1) holds.

  For part (c2) of the claim, we first observe the following.
  Each $P_3$ in \pkg\ that intersects only two cliques in $V(H)$ contains at most one edit of~$S$.
  Let $P$ be such a $P_3$ and let $D_1$, $D_2$ be the two cliques in $V(H)$ that intersect~$P$.
  Note that $\pkg_{\textsf{tra}}$ does not contain $P_3$s that intersect only two cliques in $V(H)$ and thus either $P \in \pkg_{\textsf{var}}$ or $P \in \pkg_{\textsf{pad}}$.
  In both cases, there is exactly one edge and one non-edge of $P$ between $D_1$ and $D_2$:
  This is clear if $P \in \pkg_{\textsf{pad}}$.
  If $P \in \pkg_{\textsf{var}}$ then $P$ was introduced when connecting a clause gadget to a variable gadget.
  In the notation used there, either $P = v_5 v_6 v_2$ or $P = v_1 v_7 v_8$, both of which have the required form.
  Thus, as $D_1$ and $D_2$ are either merged or not in $\P$, there is at most one edit in~$P$.

  To prove (c2) it remains to consider $P_3$s in \pkg\ that intersect three cliques in $V(H)$.
  Let $P$ be such a $P_3$.
  Note that $P \notin \pkg_{\textsf{pad}}$.
  If $P \in \pkg_{\textsf{var}}$, then it connects $K_j^i$ to $K_{j + 2}^i$ via $K_{j + 1}^i$ for some even $j$ and some variable index $i \in \{0, 1, \ldots, n - 1\}$.
  Since we merge either all odd or all even pairs in $x_i$'s variable gadget to obtain~\P, indeed exactly one edge of $P$ is edited, as claimed.
  If $P \in \pkg_{\textsf{tra}}$, then we distinguish two cases.

  First, $P$ does not contain a vertex of some variable-gadget clique.
  Then, $P$ connects some clique $Q_\cl^s$ to some transferring clique~$T^\delta_\cl$ via $Q_\cl^{s'}$.
  According to the construction of \P, either $T^\delta_\cl$ and $Q_\cl^{s'}$ are in different parts of \P\ and $Q_\cl^{s'}$ and $Q_\cl^s$ are merged, or $T^\delta_\cl$ and $Q_\cl^{s'}$ are merged and $Q_\cl^s$ and $Q_\cl^{s'}$ are in different parts of~\P.
  In both cases, there is at most one edit of $S$ in $P$.

  Second, $P$ contains a vertex of some variable-gadget clique.
  Then, by construction of $G$ and \pkg, path $P$ indeed contains two vertices of two variable-gadget cliques, say $K_j^i$ and $K_{j + 1}^i$ and one vertex of a transferring clique, say $T_\cl^i$.
  Assume that variable $x_i$ appears positively in clause $\Gamma_\cl$, the other case is analogous.
  Then the center of $P$ is $K_j^i$ and moreover $j$ is odd.
  If $x_i$ was not chosen among the variables satisfying clause~$\Gamma_\cl$ when constructing \P, then $T_\cl^i$ and $K_j^i$ is in the same part~$Q$ of~\P.
  Furthermore $K_{j + 1}^i$ is either in a part different from $Q$ or also in $Q$.
  In both cases, there is at most one edit from $S$ in~$P$.
  If $x_i$ was chosen among the the variables satisfying clause~$\Gamma_\cl$ when constructing \P, then $T_\cl^i$ is in a part in \P\ which is different from the one(s) containing $K_j^i$ and $K_{j + 1}^i$.
  However, since $x_i$ satisfies $\Gamma_\cl$, we have $\alpha(x_i) = \true$ and thus $K_j^i$ and $K_{j + 1}^i$ are merged (recall that $j$ is odd).

  Thus, indeed, the claim holds, that is, each edit in $S$ is contained in a $P_3$ in \pkg\ and every $P_3$ of \pkg\ is edited at most once by~$S$.
\end{proof}

\subsubsection{Soundness}\label{sec:soundness}
Before we show how to translate a cluster editing set of size $|\mathcal{H}|$ for the constructed instance into a satisfying assignment of $\Phi$, we make some structural observations.

Recall the definition of a proto-cluster, a connected component of the subgraph of~$G$ whose edge set contains precisely those edges of $G$ which are not contained in any $P_3$ in $\pkg$.

\begin{lemma}\label{lem:proto-clusters}
  $V(H)$ is precisely the set of proto-clusters of $G$ with respect to \pkg.
\end{lemma}
\begin{proof}
  By construction, all edges in $G$ between two cliques in $V(H)$ are in a $P_3$ in \pkg.
  Thus each proto-cluster is contained in some clique in $V(H)$.
  We claim that each clique $C \in V(H)$ contains a spanning tree of edges which are not contained in a $P_3$ in \pkg.
  If $C \in L_1$, then this is clear; such a $C$ contains only a single vertex and a trivial spanning tree.
  If $C \in L_0$, then there are only two $P_3$s in \pkg\ that contain edges of $C$: The one given by $v_5 v_6 v_2$ and the one given by $v_1 v_7 v_8$ as defined in \cref{sec:clause-gadget} when connecting variable and clause gadgets.
  Since $|C| = 5$, indeed $C$ contains the required spanning tree.
  If $C \in L_i$ for $i \geq 2$, then by the connectedness property of \cref{cor:padding}, $C$ has the required spanning tree.
\end{proof}

Recall that each solution~$S$ to $(G, \pkg, 0)$ cannot remove any edge from $G$ which is not contained in a $P_3$ in \pkg.
Thus, since $V(H)$ is a vertex partition of $G$, each solution~$S$ generates a cluster graph $G \triangle S$ whose clusters induce a coarser vertex partition than $V(H)$.
This leads to the following.
\begin{observation}\label{obs:cluster-struc}
  For each solution $S$ to $(G, \pkg, 0)$, each cluster in $G \triangle S$ is a disjoint union of cliques in $V(H)$.
\end{observation}

Using the above structural observations, we are now ready to prove the soundness of the construction.

\begin{lemma}
  If the constructed instance $(G,{\mathcal{H}}, \ell = 0)$ is a YES-instance, then the formula $\Phi$ is satisfiable.
\end{lemma}
\begin{proof}
Suppose that there exists a set of vertex pairs $S\subseteq \binom{V}{2}$ so that $G\Delta S$ is a union of vertex-disjoint cliques and $|S|-|{\mathcal{H}}|=0$. In other words, there exists a solution that transforms $G$ into a cluster graph~$G'$ by editing exactly one edge or non-edge of every $P_3$ of $\mathcal{H}$. We will construct a satisfying assignment $\alpha \colon \{x_0, x_1, \ldots, x_{n - 1}\} \to \{\true, \false\}$ for the formula $\Phi$.

By \cref{obs:cluster-struc}, the set of clusters in $G'$ induces a partition of the cliques in $V(H)$.
Recall that we say that two cliques in $V(H)$ are \emph{merged} if they are in the same cluster in $G'$ and \emph{separated} otherwise.

To define $\alpha$, we need the following observation on the solution.
Consider variable $x_i$ and the cliques $K^i_j$, $j = 0, 1, \ldots, 4m_i - 1$, in $x_i$'s variable gadget.
Call a pair $K^i_j$, $K^i_{j + 1}$ \emph{even} if $j$ is even (where $j + 1$ is taken modulo $4m_i$) and call this pair \emph{odd} otherwise.
We claim that either (i) each even pair is merged and each odd pair is separated, or (ii) each odd pair is merged and each even pair is separated (and not both).
Note that, for each even $j$, pair $K^i_j$, $K^i_{j + 1}$ is merged or pair $K^i_{j + 1}$, $K^i_{j + 2}$ is merged, because there is a $P_3$ in~$G$ containing vertices in these cliques with center in $K^i_{j + 1}$.
To show the claim, it is thus enough to show that not both an odd pair and an even pair is merged.

For the sake of contradiction, suppose that an odd pair is merged and an even pair is merged.
Then, there exists an index $j \in \{0, 1, \ldots, 4m_i - 1\}$ and a cluster~$C$ in $G'$ such that $K^i_j, K^i_{j + 1}, K^i_{j + 2} \subseteq C$, where here and below the indices are taken modulo~$4m_i$.
Observe that there are no edges between $K^i_j$ and $K^i_{j + 2}$ in~$G$.
If $j$ is odd, then all of these non-edges are non-packed.
All of these non-edges are thus in~$S$.
This is a contradiction to the fact that $S$ contains at most $|\pkg|$ vertex pairs.
Thus, $j$ is even.

We now show that for each $k \in \mathbb{N} \cup \{0\}$, pair $K^i_{j + 1 + 2k}$, $K^i_{j + 2 + 2k}$ is merged by induction on~$k$.
Clearly, for $k = 0$, this holds by supposition.
If $k > 0$ then, by the construction of $\pkg_{\textsf{var}}$, there are non-packed non-edges between $K^i_{j + 2k - 1}$ and $K^i_{j + 2k + 1}$.
Combining this with the fact that $K^i_{j + 1 + 2(k - 1)} = K^i_{j + 2k - 1}$ and $K^i_{j + 2 + 2(k - 1)} = K^i_{j + 2k}$ are merged by inductive assumption, it follows that $K^i_{j + 2k}$ and $K^i_{j + 2k + 1}$ are separated.
Since there is a $P_3$ in~$G$ connecting $K^i_{j + 2k}$, $K^i_{j + 2k + 1}$, and $K^i_{j + 2k + 2}$ with center in $K^i_{j + 2k + 1}$ and $S$ contains at most
one edit in this $P_3$, it follows that $K^i_{j + 2k + 1}$, $K^i_{j + 2k + 2}$ are merged, as required.

It now follows in particular that $K^i_{j - 1}$ and $K^i_j$ are merged (recall that indices are taken modulo $4m_i$). Since by assumption also $K^i_j$ and $K^i_{j + 1}$ are merged, we have that $K^i_{j'}$, $K^i_{j' + 1}$, and $K^i_{j' + 2}$ are contained in the same cluster in~$G'$ for some odd $j'$.
As already argued, this leads to a contradiction.
Thus the claim holds.

We define the assignment $\alpha$ as follows.
For each variable $x_i$, $i=0,1,\ldots,n-1$, if in $G'$ all even pairs $K_{2j}^{i}$, $K_{2j+1}^{i}$, $j=0,1\ldots,m_{i}-1$, are merged, then $\alpha(x_i)= \false$.
Otherwise $\alpha(x_i)= \true$.

We now show that $\alpha$ satisfies $\Phi$.
Consider an arbitrary clause $\Gamma_\cl$ of $\Phi$ containing the three variables $x_\vp$, $x_\vq$, and~$x_\vr$.
We use the same notation as when defining the clause gadget and its connection to the variable gadget.
Since there are non-packed non-edges between cliques $Q_\cl^1$ and $Q_\cl^4$, cliques $Q_\cl^1$ and $Q_\cl^4$ must end up in different clusters in $G'$. In other words, $Q_\cl^1$ and $Q_\cl^4$ are separated.
Observe that there is a path in $G$ consisting of vertices in $Q_\cl^1$, $Q_\cl^2$, $Q_\cl^3$, and $Q_\cl^4$ in this sequence.
Since each of these four cliques is a proto-cluster (\cref{lem:proto-clusters}), in order to separate $Q_\cl^1$ and $Q_\cl^4$, one of the following three cases must happen in the solution~$S$: (i) The edges between $Q_\cl^1$ and $Q_\cl^2$ are deleted. In other words, $Q_\cl^1$ and $Q_\cl^2$ are separated.
(ii)~$Q_\cl^2$ and $Q_\cl^3$ are separated.
(iii)~$Q_\cl^3$ and $Q_\cl^4$ are separated.
We now show that case (i), (ii), and (iii) imply that variable $x_\vp$, $x_\vq$, and $x_\vr$, respectively, is set by $\alpha$ so as to satisfy $\Gamma_\cl$.
We only give the proof showing that case~(i) implies that $x_\vp$ is set accordingly.
The other cases are analogous.

Assume that case~(i) holds.
Then, by the constraints imposed by the two transferring $P_3$s $P_\cl^1$ and $P_\cl^2$, cliques $T_\cl^\vp$ and $Q_\cl^1$ are merged.
Since there are non-packed non-edges between $K_{4\pi(\vp,\cl)+1}^\vp$ and $Q_\cl^1$, it follows that $K_{4\pi(\vp,\cl)+1}^{\vp}$ and $Q_\cl^1$ are separated.
Consider the case that $x_\vp$ appears positively in $\Gamma_\cl$.
Then, when connecting the variable gadget of $x_\vp$ to the clause gadget of $\Gamma_\cl$ we have introduced into~$G$ a $P_3$ connecting $T^\vp_\cl$, $K_{4\pi(\vp,\cl)+1}^{\vp}$, and $K_{4\pi(\vp,\cl)+2}^{\vp}$ with center in $K_{4\pi(\vp,\cl)+1}^{\vp}$ (for example, the $P_3$ given by $w_1 v_1 v_3$).
Since $T^\vp_\cl$ and $K_{4\pi(\vp,\cl)+1}^{\vp}$ are separated, thus $K_{4\pi(\vp,\cl)+1}^{\vp}$ and $K_{4\pi(\vp,\cl)+2}^{\vp}$ are merged.
There is thus at least one odd pair in $x_\vp$'s variable gadget that is merged and thus $\alpha(x_\vp) = \true$.
The case where $x_\vp$ appears negatively in $\Gamma_\cl$ is similar: We have introduced into~$G$ a $P_3$ connecting $T^\vp_\cl$, $K_{4\pi(\vp,\cl)+1}^{\vp}$, and $K_{4\pi(\vp,\cl)}^{\vp}$ with center in $K_{4\pi(\vp,\cl)+1}^{\vp}$ (for example, the $P_3$ given by $w_1 v_1 v_3$).
It follows that $K_{4\pi(\vp,\cl)+1}^{\vp}$, and $K_{4\pi(\vp,\cl)}^{\vp}$ are merged, showing that at least one even pair is merged in $x_\vp$'s variable gadget.
Thus, $\alpha(x_\vp) = \false$.

Thus each clause $\Gamma_\cl$ is satisfied, finishing the proof.
\end{proof}

\section{XP-algorithm for half-integral packings}\label{sec:xp-algorithm}

In this section, we study \pCEA\ in the special setting where every vertex is incident with at most two $P_3$s of the packing \pkg. More precisely, we consider the following variant of \pCEA.

\defprobMD{\pCEATlong\ (\pCEAT)}
{A graph $G=(V,E)$, a modification-disjoint packing \pkg\ of induced $P_{3}$s of $G$ such that every vertex $v\in V(G)$ is incident with at most two $P_3$s of \pkg, and a non-negative integer~$\ell$.}
{Is there a cluster editing set, i.e., a set of vertex pairs
$S\subseteq \binom{V}{2}$ so that $G\triangle S$ is a union of disjoint cliques, with $|S|-|\mathcal{H}|\leq\ell$?}

We give a polynomial-time algorithm to solve \pCEAT\ when $\ell$ is a fixed constant, in contrast with the NP-hardness of the general version of \pCEA\ when $\ell=0$.
\begin{thmhalfintegralxp}[Restated]
  \thmhalfintegralxpstatement
\end{thmhalfintegralxp}
The main tool in proving \cref{thm:half-integral-xp} is a polynomial-time algorithm for the case where \(\ell = 0\):
\newcommand\thmhalfintegralpolynoexcessstatement{%
  \pCEATlong\ can be solved in polynomial time when $\ell=0$, that is, when no excess edits are allowed.%
}
\begin{theorem} \label{thm:half-integral-poly-no-excess}
  \thmhalfintegralpolynoexcessstatement
\end{theorem}
The proof of \cref{thm:half-integral-poly-no-excess} will be given in~\cref{sec:poly-algo}. With this tool in hand, we can show \cref{thm:half-integral-xp}.

\begin{algorithm}[t]
  \DontPrintSemicolon

  \KwIn{An instance \((G, \pkg, \ell)\) of \pCEAT.}

  \KwOut{Whether \((G, \pkg, \ell)\) is a YES-instance.}

  \ForEach{\(\ell_a = 0, 1, \ldots, \ell\)}{
    \ForEach{\(\ell_b = 0, 1, \ldots, \ell - \ell_a\)}{
      \ForEach{set \(S_a\) of \(\ell_a\) vertex pairs \(\{u, v\} \in \binom{V(G)}{2}\) such that \(\forall P \in \pkg \colon |\{u, v\} \cap V(P)| \leq 1\)}{
        \(G_a \leftarrow G \triangle S_A\)\;
        \ForEach{set \(\pkg_b\) of \(\ell_b\) distinct \(P_3\)s in \(\pkg\)}{
          \ForEach{set \(S_b\) containing for each \(P \in \pkg_b\) at least two vertex pairs in \(V(P)\)}{
            \If{\(|S_a| + |S_b| \leq |\pkg_b| + \ell\)}{
              \(G_b \leftarrow G_a \triangle S_B\)\;
              \(\pkg' \leftarrow \pkg \setminus \pkg_b\)\;
              \If(\tcc*[f]{Using \cref{thm:half-integral-poly-no-excess}}){\(G_b\) has a cluster-editing set with \(|\pkg'|\) edits}{accept and halt}
            }
          }
        }
      }
    }
  }
  reject\;
  \caption{Solve \pCEAT.}
  \label{alg:half-integral-xp}
\end{algorithm}

\begin{proof}[Proof of \cref{thm:half-integral-xp}]
  Let \((G, \pkg, \ell)\) be an instance of \pCEAT.
  The algorithm is given in \cref{alg:half-integral-xp}.
  Essentially, it guesses (by trying all possibilities) the number, \(\ell_a\), of excess edits that are not contained in any \(P_3\) in \pkg\ and guesses the concrete edits to be made (Lines~1-4).
  Then it guesses the \(P_3\)s in \pkg\ that harbor the remaining excess edits and it guesses how these \(P_3\)s are resolved (Lines~5-9).
  Then it checks whether the remaining instance has a cluster-editing set without excess edits over the remaining \(P_3\) packing \(\pkg'\) using the algorithm from \cref{thm:half-integral-poly-no-excess}.

  For the running time, observe that there are at most \(n^{2\ell_a}\) choices for \(S_a\).
  Since each vertex is in at most two \(P_3\)s in \pkg\ and each \(P_3\) covers exactly three vertices, we have \(3|\pkg| \leq 2n\) and thus there are in total at most \(n\) \(P_3\)s in \(\pkg\).
  Thus, there are \(O(n^{\ell_b})\) choices for \(\pkg_b\).
  Since there are four possibilities to select a set of at least two vertex-pairs in the vertex set of a \(P_3\), there are \(O(4^{\ell_b})\) possibilities for \(S_b\) in Line~6.
  Hence, overall the running time is \(O(4^{\ell_b}n^{2\ell_a + \ell_b + O(1)}) \leq n^{2\ell + O(1)}\).

  It remains to prove the correctness.
  If the algorithm accepts, then there is a cluster-editing set~\(S_0\) for \(G_b\) with \(|\pkg'|\) edits.
  Since \(S_0\) is contained in the vertex sets of the \(P_3\)s in \(\pkg'\), set \(S_0\) is disjoint from \(S_a\) and \(S_b\).
  Thus, \(G \triangle S^\star\) is a cluster graph where \(S^\star = S_a \cup S_b \cup S_0\).
  Moreover, \(|S^\star| \leq |\pkg'| + |\pkg_b| + \ell = |\pkg| + \ell\), and thus, \((G, \pkg, \ell)\) is a YES-instance.

  Conversely, if \((G, \pkg, \ell)\) is a YES-instance, then there is a cluster-editing set~\(S^\star\) of \(G\) with \(|S^\star| \leq |\pkg| + \ell\).
  Let \(S^\star_a\) be the subset of \(S^\star\) that contains precisely those edits in \(S^\star\) that are not contained in \(P_3\)s of \(\pkg\).
  In one of the iterations of \cref{alg:half-integral-xp}, \(\ell_a = |S^\star_a|\) and \(S_a = S^\star_a\).
  Now let \(\pkg_b^\star\) be the subset of \(\pkg\) that contains precisely those \(P_3\)s \(P\) such that \(S^\star\) contains at least two edits in \(V(P)\).
  Observe that \(|\pkg_b^\star| \leq \ell - \ell_a\).
  Thus, in one of the iterations of \cref{alg:half-integral-xp}, we have \(\ell_b = |\pkg_b^\star|\) and \(\pkg_b = \pkg_b^\star\).
  Moreover, in one of the iterations \(S_b = S^\star_b\), where \(S^\star_b\) is the subset of \(S^\star\) that contains precisely those edits that are contained in the \(P_3\)s in \(\pkg_b\).
  Let \(S^\star_0 = S^\star \setminus (S^\star_a \cup S^\star_b)\).
  Since each edit in \(S^\star_0\) is contained in a unique \(P_3\) in \(\pkg \setminus \pkg_b^\star\), we have \(|S_a| + |S_b| = |S^\star_a| + |S^\star_b| \leq |\pkg_b^\star| + \ell = |\pkg_b| + \ell\).
  Thus, in that iteration the algorithm proceeds to the if-condition in Line~10.
  Again since each edit in \(S^\star_0\) is contained in a unique \(P_3\) in \(\pkg \setminus \pkg_b^\star\), this set witnesses that \((G_b, \pkg', 0)\) is a YES-instance and thus the algorithm accepts.
  Hence, the algorithm is correct.
\end{proof}

\subsection{Polynomial-time algorithm for zero excess edits}\label{sec:poly-algo}
Let \pCEMTlong\ (\pCEMT) be the special case of \pCEAT\ where $\ell = 0$.
That is, an instance of \pCEMT\ is given by a tuple $(G, \mathcal{H})$ of a graph $G$ and a half-integral $P_3$ packing $\mathcal{H}$ in $G$.
In this section we give a polynomial-time algorithm for \pCEMT.
Again, we use the term \emph{proto-clusters} to denote the connected components of the graph obtained by removing the edges of all packed $P_3$s.

The intuition behind the polynomial-time result is that, with the constraint that every vertex $v\in V(G)$ is incident with at most two packed $P_3$s, we cannot freely merge or separate two large proto-clusters without excess edits as in the NP-hardness proof of \cref{sec:lower-bound}.
This is because the triangles formed by the packed $P_3$s cannot cover every vertex pair between two large proto-clusters. Thus we can separate the large proto-clusters and deal with them separately.

The polynomial-time algorithm mainly proceeds by applying reduction rules that simplify the instance step by step.
Herein, our first goal is to eliminate proto-clusters of size at least four, which can be done by a series of straightforward reduction rules (\cref{sec:simple-rules}).
We then look at proto-clusters of size three and observe that their connections to the rest of the graph have quite a limited structure.
This observation can be used to eliminate proto-clusters of size three as well (\cref{sec:clust-size-struc}).
The reduction rules we have developed at this point give more structural observations on smaller proto-clusters which can be used to show that the size of solution clusters is at most four (\cref{sec:sol-clust-size}).
Afterwards, we show that the only situation in which solution clusters of size four can occur is when there is a certain path-like structure in the instance.
A final, quite involved reduction rule takes care of such path-like structures (\cref{sec:path-like}).
This then results in an instance with a solution whose clusters have size at most three.
Using this cluster-size bound we can finally show that, if there is a solution, then there is also one that only deletes edges.
This then leads to a formulation as an instance of \textsc{2-SAT} (\cref{sec:2sat}), which is well-known to be polynomial-time solvable.

We use the following notation.
We say a proto-cluster $C$ is \emph{isolated} from a proto-cluster $D$ if there are no edges of $G$ between $C$ and $D$.
We classify the $P_3$s of $\mathcal{H}$ into four types.
For an induced $P_3$ $xyz\in \mathcal{H}$:
\begin{itemize}
\item if $x,y$ belong to one proto-cluster and $z$ belongs to another proto-cluster, or symmetrically $y,z$ belong to one proto-cluster and $x$ belongs to another proto-cluster, then $xyz$ is a \emph{type-$\alpha$} $P_3$;
\item if $x,z$ belong to one proto-cluster and $y$ belongs to another proto-cluster, then $xyz$ is a \emph{type-$\beta$} $P_3$;
\item if $x,y,z$ belong to three distinct proto-clusters respectively, then $xyz$ is a \emph{type-$\gamma$} $P_3$; and
\item if $x,y,z$ belong to one proto-cluster then $xyz$ is a \emph{type-$\delta$} $P_3$.
\end{itemize}
As mentioned, in the following, we present a series of reduction rules, which are algorithms that take an instance of \pCEMT\ and produce a new instance of \pCEMT.
By saying that a reduction rule is \emph{safe}, we mean that the instance before applying this reduction rule is a YES-instance if and only if the instance after applying this reduction rule is a YES-instance.
Since the $P_3$s of $\mathcal{H}$ are modification-disjoint, we have the following handy observation.
\begin{observation} \label{obs2}
  A solution $S$ to an instance of \pCEMT\ must edit exactly one edge or non-edge of every $P_3$ of $\mathcal{H}$, and neither non-packed edges nor non-packed non-edges can be edited by $S$.
\end{observation}

\subsubsection{Simple reduction rules}\label{sec:simple-rules}
We start by getting rid of several simple situations.
\begin{reduction} \label{RR1}
For any proto-cluster $C$, if there are two vertices $u,v\in V(C)$ such that $uv$ is a non-packed non-edge, i.e., $uv$ is not covered by any $P_3$ of $\mathcal{H}$, then return NO.
\end{reduction}

\begin{lemma}
Reduction Rule~\ref{RR1} is safe.
\end{lemma}
\begin{proof}
Given an instance $(G,{\mathcal{H}})$ of \pCEMT\ satisfying the condition of Reduction Rule~\ref{RR1}, suppose for contradiction that there is a solution $S$ to this instance. Since $u,v$ belong to the same proto-cluster, there is a non-packed path $P$ from $u$ to $v$. By Observation~\ref{obs2}, $uv\notin S$ and none of the edges of $P$ is edited by $S$. Thus $G\triangle S$ is not a cluster graph, contradicting that the instance has a solution. This completes the proof for the lemma.
\end{proof}

The second reduction rule handles type-$\beta$ and type-$\delta$ $P_3$s (see \cref{fig:ReductionRule2}).

\begin{reduction} \label{RR2}
If there is a type-$\beta$ or type-$\delta$ $P_3$ $xyz\in\mathcal{H}$, insert the edge $xz$ and remove $xyz$ from~$\mathcal{H}$.
\end{reduction}

\begin{figure}[t]
\begin{center}
\begin{tikzpicture}[scale=0.7]

		\node (4) at (-4, -0) {};
		\node [style=nodestyle2] (5) at (-4, -0.75) {};
		\node [style=nodestyle2] (6) at (-4, 0.75) {};
		\node [style=nodestyle2] (7) at (-2, -0) {};

        \draw[draw=black] (4) ellipse (15pt and 45pt);
        \draw[draw=black] (7) ellipse (15pt and 45pt);
		\node  (16) at (-3, -2.25) {type-$\beta$};

		\draw [style=type1] (5) to (7);
		\draw [style=type1] (6) to (7);
		\draw [style=type2] (6) to (5);

\begin{scope}[shift={(-3cm,0cm)}]

        \node (11) at (5.5, -0) {};
		\node [style=nodestyle2] (12) at (5.5, 1) {};
		\node [style=nodestyle2] (13) at (6, -0) {};
		\node [style=nodestyle2] (14) at (5.5, -1) {};
        \draw[draw=black] (5.6,0) ellipse (20pt and 45pt);

		\node  (18) at (5.75, -2.25) {type-$\delta$};

		\draw [style=type1] (12) to (13);
		\draw [style=type1] (13) to (14);
		\draw [style=type2] (12) to (14);

\end{scope}
\end{tikzpicture}
\end{center}
\vspace*{-3mm}
\caption{Examples for Reduction Rule \ref{RR2}.} \label{fig:ReductionRule2}
\end{figure}

\begin{lemma}
Reduction Rule~\ref{RR2} is safe.
\end{lemma}
\begin{proof}
  Suppose that the given instance of \pCEMT\ is $(G,{\mathcal{H}})$ such that there exists a type-$\beta$ $P_3$ $xyz$ in $G$. After inserting the edge $xz$ and removing $xyz$ from $\mathcal{H}$, we get an instance $(G',{\mathcal{H}}')$. We claim that $(G,{\mathcal{H}})$ is a YES-instance if and only if $(G',{\mathcal{H}}')$ is a YES-instance. On one hand, suppose that $(G',{\mathcal{H}}')$ is a YES-instance and $S'$ is a cluster editing set of $G'$ such that $|S'|=|{\mathcal{H}}'|$. Obviously, $S'\cup\{xz\}$ is a cluster editing set for $G$ and $|S'\cup\{xz\}|=|{\mathcal{H}}|$. On the other hand, suppose that $(G,{\mathcal{H}})$ is a YES-instance and $S$ is a cluster editing set of $G$ such that $|S|=|\mathcal{H}|$. We show that $xz\in S$ and $S\setminus \{xz\}$ is the solution for $(G',{\mathcal{H}}')$. For contradiction, suppose this is not true. Then either $xy\in S$ or $yz\in S$ holds. Without loss of generality we assume that $xy\in S$. Suppose that after deleting $xy$ from $G$ and removing $xyz$ from $\mathcal{H}$, we get an instance $(G'',{\mathcal{H}}'')$. Since $x,z$ belong to one proto-cluster of $G$, there is a non-packed path $P$ from $x$ to $z$ in $G$. Thus $x,z$ belong to one proto-cluster of $G''$. Since $xyz$ is removed from $\mathcal{H}$, $xz$ becomes a non-packed non-edge. By Reduction Rule~\ref{RR1}, $(G'',{\mathcal{H}}'')$ is a NO-instance, contradicting that $S$ is a solution to $(G,{\mathcal{H}})$.

A similar analysis applies to the case that $xyz\in \mathcal{H}$ is a type-$\delta$ $P_3$. This completes the proof for the lemma.
\end{proof}

After applying Reduction Rules~\ref{RR1} and~\ref{RR2} exhaustively, if the algorithm did not return NO, then there is no type-$\beta$ or type-$\delta$ $P_3$s in the instance.
The next reduction rule applies to the case in which there is both a non-packed non-edge and a packed edge between two proto-clusters, see \cref{fig:ReductionRule3} for an illustration.

\begin{reduction} \label{RR3}
For any two proto-clusters $A$ and $B$, if there is a non-packed non-edge $uv$ such that $u\in V(A)$ and $v\in V(B)$, and there is a packed edge $xy$ such that $x\in V(A)$ and $y\in V(B)$ (not necessarily distinct from $u$ or $v$), then delete $xy$ and remove the corresponding packed $P_3$ from $\mathcal{H}$.
\end{reduction}

\begin{figure}[t]
\begin{center}
\begin{tikzpicture}[scale=0.7]

		\node (0) at (-5, -0) {};
		\node (1) at (-3, -0) {};
		\node [style=nodestyle2,label=above:$u$] (2) at (-5, 0.75) {};
		\node [style=nodestyle2,label=above:$v$] (3) at (-3, 0.75) {};
		\node [style=nodestyle2,label=below:$x$] (4) at (-5, -0.5) {};
		\node [style=nodestyle2,label=below:$y$] (5) at (-3, -0.5) {};
		\node [style=nodestyle2] (6) at (-1, 0.35) {};
        \node (9) at (-1, -0) {};

		\node[font=\large] (7) at (-5, -2.3) {$A$};
		\node[font=\large] (8) at (-3, -2.3) {$B$};
		\node  (16) at (-3, -3) {type-$\gamma$};

        \draw[draw=black] (0) ellipse (15pt and 45pt);
        \draw[draw=black] (1) ellipse (15pt and 45pt);
        \draw[draw=black] (9) ellipse (15pt and 45pt);

		\draw [draw=black,thick,dotted] (2) to (3);

		\draw [style=type1] (5) to (6);

		\draw [style=type1] (4) to (5);
        \draw [style=type2] (4) to (6);

\begin{scope}[shift={(9cm,0cm)}]

		\node (0) at (-5, -0) {};
		\node (1) at (-3, -0) {};
		\node [style=nodestyle2,label=above:$u$] (2) at (-5, 0.75) {};
		\node [style=nodestyle2,label=above:$v$] (3) at (-3, 0.75) {};
		\node [style=nodestyle2,label=below:$x$] (4) at (-5, -0.7) {};
		\node [style=nodestyle2,label=below:$y$] (5) at (-3, -0.7) {};
		\node [style=nodestyle2] (6) at (-3, 0) {};

		\node  (18) at (-4, -3) {type-$\alpha$};

		\node[font=\large] (7) at (-5, -2.3) {$A$};
		\node[font=\large] (8) at (-3, -2.3) {$B$};

        \draw[draw=black] (0) ellipse (15pt and 45pt);
        \draw[draw=black] (1) ellipse (15pt and 45pt);

		\draw [draw=black,thick,dotted] (2) to (3);

		\draw [style=type1] (5) to (6);

		\draw [style=type1] (4) to (5);
        \draw [style=type2] (4) to (6);
\end{scope}

\end{tikzpicture}
\end{center}
\caption{Examples for Reduction Rule \ref{RR3}.} \label{fig:ReductionRule3}
\end{figure}

\begin{lemma}
Reduction Rule~\ref{RR3} is safe.
\end{lemma}
\begin{proof}
  Given an instance $(G,{\mathcal{H}})$ of \pCEMT\ satisfying the condition of Reduction Rule~\ref{RR3} with $xy$ covered by a type-$\gamma$ $P_3$ $xyz$.
  Without loss of generality, we do not analyze the symmetrical case where $x$ is the center vertex of the $P_3$ instead of $y$.
  We get an instance $(G',{\mathcal{H}}')$ of \pCEMT\ after deleting $xy$ and removing $xyz$ from $\mathcal{H}$.
  We claim that $(G,{\mathcal{H}})$ is a YES-instance if and only if $(G',{\mathcal{H}}')$ is a YES-instance.
For the soundness, assume that $(G',{\mathcal{H}}')$ is a YES-instance and $S'$ is a cluster editing set of size $|{\mathcal{H}}'|$ for $G'$. Then obviously $S'\cup \{xy\}$ is a solution to $(G,{\mathcal{H}})$.
For the completeness, assume that $(G,{\mathcal{H}})$ is a YES-instance and $S$ is a cluster editing set of size $|\mathcal{H}|$ for $G$. We claim that $xy\in S$. Suppose for contradiction that $xy\notin S$. Then $xy$ becomes a non-packed edge in $G\triangle S$. Since $u,x\in V(A)$ and $v,y\in V(B)$, there is a non-packed path $P_A$ from $u$ to $x$ and a non-packed path $P_B$ from $v$ to $y$ in $G$. By Observation~\ref{obs2}, the edges of $P_A$ and $P_B$ are not edited by $S$ and $uv\notin S$. Thus there is a non-packed path from $u$ to $v$. Since $uv$ is a non-packed non-edge in $G\triangle S$, $G\triangle S$ is not a cluster graph, contradicting the assumption that $S$ is a solution to $(G,{\mathcal{H}})$.

A similar analysis applies to the case in which $xy$ is covered by a type-$\alpha$ $P_3$ $xyz$ (and its symmetrical case where $x$ is the center vertex instead of $y$). This concludes the proof for the lemma.
\end{proof}

The next reduction rule deals with isolated cliques in graph $G$.

\begin{reduction} \label{RRIsClique}
If there is a proto-cluster $C$ which is an isolated clique of $G$, then remove $C$ from the graph.
\end{reduction}

\begin{lemma} \label{safeRRIsClique}
Reduction Rule~\ref{RRIsClique} is safe.
\end{lemma}

\begin{proof}
Given an instance $(G,{\mathcal{H}})$ of \pCEMT\ such that there is a proto-cluster $C$ which is an isolated clique, we remove $C$ from $G$ and get an instance $(G',{\mathcal{H}})$. We claim that $(G,{\mathcal{H}})$ is a YES-instance if and only if $(G',{\mathcal{H}})$ is a YES-instance. On one hand, assume that $(G',{\mathcal{H}})$ is a YES-instance. Then obviously $(G,{\mathcal{H}})$ is a YES-instance. On the other hand, assume that $(G,{\mathcal{H}})$ is a YES-instance and $S$ is a solution. Since $C$ is an isolated clique, by Observation~\ref{obs2}, neither edges of $C$ nor non-edges between $V(C)$ and $V(G)\setminus V(C)$ are edited by $S$. Thus $S$ is also a solution to $(G',{\mathcal{H}})$. This completes the proof for the lemma.
\end{proof}

In later analysis, we will see that some constant-size configurations cannot be connected to the rest of the graph.
To remove such configurations, we introduce the following reduction rule.

\begin{reduction} \label{RR4}
  If there is a connected component $C$ in $G$ of size at most $6$, then do brute force on $C$ to check if there is a cluster editing set $F$ for $C$ such that $|F|$ is equal to the number of packed $P_3$s incident with a vertex of $C$.
  If there is such a cluster editing set $F$, then perform the operations of $F$ to $C$ and remove the corresponding packed $P_3$s from $\mathcal{H}$.
  Otherwise, if there is no such cluster editing set $F$, return NO.
\end{reduction}

\begin{lemma}
Reduction Rule~\ref{RR4} is safe.
\end{lemma}
\begin{proof}
  Given an instance $(G,{\mathcal{H}})$ of \pCEMT\ such that there is a connected component $ C$ in the graph of size at most $6$, suppose that there is a cluster editing set $F$ for $ C$ satisfying the condition of Reduction Rule~\ref{RR4}. After performing the operations of $F$, we get an instance $(G',{\mathcal{H}}')$ of \pCEMT.
  We claim that $(G,{\mathcal{H}})$ is a YES-instance if and only if $(G',{\mathcal{H}}')$ is a YES-instance. On one hand, assume that $(G',{\mathcal{H}}')$ has a solution $S'$. Obviously, $S'\cup F$ is a cluster editing set for $G$ and $|S'\cup F|=|{\mathcal{H}}|$. On the other hand, assume that $(G,{\mathcal{H}})$ has a solution $S$. By Observation~\ref{obs2}, no vertex pair between $V( C)$ and $V(G)\setminus V( C)$ is edited by $S$. Let $S_1\subseteq S$ be the set of vertex pairs which are edges or non-edges of $ C$. Then $S\setminus S_1$ is a solution to~$(G',{\mathcal{H}}')$.

Suppose that there is no such cluster editing set $F$ for $ C$. We claim that $(G,{\mathcal{H}})$ is a NO-instance. For contradiction, assume that $(G,{\mathcal{H}})$ has a solution $S$. Let $S_1\subseteq S$ be the set of vertex pairs which are edges or non-edges of $ C$. Then $S_1$ is a cluster editing set for $ C$ and $|S_1|$ is equal to the number of packed $P_3$s incident with a vertex of $ C$ by Observation~\ref{obs2}, a contradiction. Thus $(G,{\mathcal{H}})$ is a NO-instance.

The component $ C$ is of size at most $6$ so we can do brute force in constant time.
This completes the proof for the lemma.
\end{proof}

We now move to analyzing the size of the remaining proto-clusters.

\begin{lemma} \label{noC5}
After applying Reduction Rules \ref{RR1} to~\ref{RRIsClique} exhaustively, if the algorithm did not return NO, then there is no proto-cluster of size at least $5$.
\end{lemma}
\begin{proof}
  Suppose for contradiction that there is a proto-cluster $C$ of size at least $5$. If $C$ is a proto-cluster which is isolated from other proto-clusters, then $C$ must be a clique since otherwise Reduction Rule~\ref{RR1} or Reduction Rule~\ref{RR2} can be applied, a contradiction. Then Reduction Rule~\ref{RRIsClique} can be applied and $C$ will be removed from the graph. Thus $C$ is not an isolated proto-cluster.

  Let $D$ be a proto-cluster such that there is an edge $uv$ between $C$ and $D$, say $u\in V(C)$ and $v\in V(D)$. If $uv$ is covered by a type-$\beta$ $P_3$, then Reduction Rule~\ref{RR2} can be applied, a contradiction. Thus we assume that $uv$ is covered by a type-$\alpha$ or a type-$\gamma$ $P_3$.
Since $v$ is incident with at most two packed $P_3$s, there must be one vertex $w\in V(C)$ such that $wv$ is a non-packed non-edge. Then Reduction Rule~\ref{RR3} can be applied, a contradiction. As a result, there is no proto-cluster of size at least $5$. This completes the proof for the lemma.
\end{proof}

Next we focus on proto-clusters of size $4$.

\begin{figure}[t]
\begin{center}
\begin{tikzpicture}

		\node [style=nodestyle2,label=right:{$v_1$}] (0) at (-3, 2) {};
		\node [style=nodestyle2,label={[label distance=-0.5mm]30:$v_{2}$}] (1) at (-3, 1) {};
		\node [style=nodestyle2,label={[label distance=-0.5mm]-30:$v_{3}$}] (2) at (-3, -0) {};
		\node [style=nodestyle2,label=right:{$v_4$}] (3) at (-3, -1) {};
		\node [style=nodestyle2,label={[label distance=-1mm]0:$w$}] (4) at (0, 0.5) {};
		\node [style=none] (5) at (-3, 0.5) {};
        \draw [draw=black] (4) ellipse (10pt and 12pt);
		\draw [draw=black] (5) ellipse (25pt and 54pt);

		\node [style=none] (6) at (-4.5, 0.5) {$C$};
		\node [style=none] (7) at (0.65, 0.5) {$D$};

		\draw [red,thick] (0) to (1);
		\draw [red,thick] (1) to (4);
		\draw [green,thick] (2) to (3);
		\draw [green,thick] (2) to (4);
		\draw [red,thick,dashed] (0) to (4);
		\draw [green,thick,dashed] (3) to (4);
		\draw [black,thick] (1) to  (2);
		\draw [black,thick] (0)[bend right=30] to (2);
		\draw [black,thick] (3)[bend left=30] to (1);
        \draw [black,thick] (0)[bend right=45] to (3);

\end{tikzpicture}
\end{center}
\caption{An example for Lemma \ref{noC4}.} \label{fig:Lemma-noC4}
\end{figure}

\begin{lemma} \label{noC4}
After applying Reduction Rules \ref{RR1} to \ref{RR3} exhaustively, if there is a proto-cluster $C$ of size $4$ which is not an isolated clique of $G$, then there is a proto-cluster $D$ of size $1$ such that the vertex pairs between $C$ and~$D$ are covered by two type-$\alpha$ $P_3$s. In addition, $V(C)\cup V(D)$ forms a connected component in the graph.
\end{lemma}
\begin{proof}
  After applying Reduction Rules \ref{RR1} to \ref{RR3} exhaustively, let $C$ be a proto-cluster of size $4$ and $V(C)=\{v_1,v_2,v_3,v_4\}$.
  See \cref{fig:Lemma-noC4} for an illustration.
  Let $w$ be a vertex such that there is an edge between $w$ and $V(C)$.
  If the vertex pairs between $V(C)$ and $w$ are not covered by two type-$\alpha$ $P_3$s,
  then either there is a non-packed non-edge between $C$ and $D$ or there is a type-$\beta$ $P_3$ between $C$ and $D$.
  Thus Reduction Rule~\ref{RR2} or~\ref{RR3} can be applied, a contradiction.
  Without loss of generality, suppose that $v_1v_2$ and $v_3v_4$ are covered by these two type-$\alpha$ $P_3$s. Assume for contradiction that there is another vertex $u$ such that $u$ and (without loss of generality) $v_1$ are adjacent, and $uv_1$ is a packed edge. Since we have applied Reduction Rule~\ref{RR2} exhaustively, there are neither type-$\beta$ nor type-$\delta$ $P_3$s in the graph. Thus $uv_1$ must be covered by a type-$\alpha$ or a type-$\gamma$ $P_3$.

  We claim that there must be a non-packed non-edge from $u$ to a vertex of $C$. For contradiction, suppose this is not true. Then either $v_1v_4,v_2v_3$ are covered by two type-$\alpha$ $P_3$s respectively, or $v_1v_3,v_2v_4$ are covered by two type-$\alpha$ $P_3$s respectively. In both cases, $v_1,v_2,v_3$ and $v_4$ are not in one proto-cluster anymore since after removing the packed edges, $v_1,v_2,v_3$ and $v_4$ are not in one connected component, a contradiction. Thus there must be a non-packed non-edge between $V(C)$ and $u$. Since $uv_1$ is a packed edge, Reduction Rule~\ref{RR3} can be applied to $C$ and the proto-cluster containing $u$, a contradiction. Thus there are no edges between $V(C)$ and any other vertices except $w$.

  Suppose that $w$ belongs to a clique of size at least two. Then there must be a non-packed non-edge and a packed edge between $C$ and $D$ (there cannot be more than two packed $P_3$s between a proto-cluster of size $4$ and another proto-cluster). Thus Reduction Rule~\ref{RR3} can be applied, a contradiction. Thus $w$ belongs to a proto-cluster of size one and let this proto-cluster be $D$. Since $w$ is already incident with two packed $P_3$s, $w$ is isolated from any other proto-clusters except $C$. Obviously, $V(C)\cup V(D)$ forms a connected component in the graph.
  This completes the proof for the lemma.
\end{proof}

\begin{lemma} \label{noC4-B}
After applying Reduction Rules \ref{RR1} to~\ref{RR4} exhaustively, there is no proto-cluster of size $4$.
\end{lemma}
\begin{proof}
Suppose for contradiction that there is a proto-cluster $C$ of size at least $4$. If $C$ is an isolated proto-cluster, $C$ must be a clique since otherwise Reduction Rule~\ref{RR1} or~\ref{RR2} can be applied, a contradiction. Then Reduction Rule~\ref{RRIsClique} can be applied and $C$ will be removed from the graph. Thus $C$ is not an isolated proto-cluster. By Lemma~\ref{noC4}, there is a proto-cluster $D$ of size $1$ such that $V(C)\cup V(D)$ forms a connected component of size $5$ in the graph. Then Reduction Rule~\ref{RR4} can be applied, a contradiction. As a result, there is no proto-cluster of size at least $4$. This completes the proof for the lemma.
\end{proof}

Summarizing, using the simple Reduction Rules \ref{RR1} to \ref{RRIsClique} we have successfully removed all proto-clusters of size at least four.

\subsubsection{Decreasing the proto-cluster size and structural observations}\label{sec:clust-size-struc}

Next, we focus on the structure of proto-clusters of size three and how to remove them as well.
First, we observe how connections around proto-clusters of size three look like.
See \cref{fig:ReductionRule6} for an illustration of these connections.

\begin{lemma} \label{noC3}
After applying Reduction Rules \ref{RR1} to \ref{RRIsClique} exhaustively, if there is a proto-cluster $C$ of size $3$, then there must be a proto-cluster $B$ of size $1$ and a proto-cluster $A$ of size $1$, such that the vertex pairs between $C$ and $B$ are covered by a type-$\alpha$ $P_3$ and a type-$\gamma$ $P_3$, and the type-$\gamma$ $P_3$ connects $C$ and $A$ via $B$. In addition, $C$ is isolated from any other proto-clusters except~$B$, and $B$ is isolated from any other proto-clusters except $A$ and $C$.
\end{lemma}

\begin{proof}
  After applying Reduction Rules \ref{RR1} to \ref{RRIsClique} exhaustively, let $C$ be a proto-cluster of size $3$.
  If $C$ is isolated from other proto-clusters, then $C$ must be a clique since otherwise Reduction Rule~\ref{RR1} can be applied.
  However, then Reduction Rule~\ref{RRIsClique} can be applied, a contradiction. Thus we assume that $C$ is not an isolated proto-cluster.

  Let the three vertices of $C$ be $u_1$, $u_2$, and $u_3$.
  Let $v$ be a vertex such that there is an edge between $v$ and $V(C)$. If the vertex pairs between $V(C)$ and $v$ are not covered by a type-$\alpha$ $P_3$ and a type-$\gamma$ $P_3$, then Reduction Rule~\ref{RR2} or~\ref{RR3} can be applied as $v$ can be incident with at most two packed $P_3$s, a contradiction. Without loss of generality, suppose that $u_1$, $u_3$, and $v$ belong to a type-$\alpha$ $P_3$. Assume for contradiction that there is another vertex $w$ such that $w$ is adjacent to some vertex of $V(C)$ ($w$ can either belong to the same proto-cluster as $v$ or belong to a different proto-cluster from $v$). If the vertex pairs between $V(C)$ and $w$ are not covered by a type-$\alpha$ $P_3$ and a type-$\gamma$ $P_3$, then Reduction Rule~\ref{RR2} or~\ref{RR3} can be applied to the corresponding $P_3$ or proto-clusters, a contradiction. If the vertex pairs between $V(C)$ and $w$ are covered by a type-$\alpha$ $P_3$ and a type-$\gamma$ $P_3$, say $u_1,u_2$ and $w$ belong to the type-$\alpha$ $P_3$, then $u_1$, $u_2$, and $u_3$ are not in one proto-cluster, a contradiction. It follows that there is no vertex adjacent to one of the vertices of $V(C)$ except $v$.

  Let $B$ be the proto-cluster to which $v$ belongs. Assume for contradiction that $|B|>1$ and there is another vertex $y$ belonging to $B$. As argued above, $y$ is not adjacent to any vertex of $V(C)$ and there is a non-packed non-edge between $V(B)$ and $V(C)$. Thus Reduction Rule~\ref{RR3} can be applied, a contradiction. It follows that $|B|=1$ and $C$ is isolated from any other proto-clusters except~$B$. We have assumed that $u_1,u_3$ and $v$ belong to a type-$\alpha$ $P_3$. As argued above, $u_2v$ is covered by a type-$\gamma$ $P_3$. Let $u_2vx$ be that type-$\gamma$ $P_3$ where $x$ belongs to a proto-cluster~$A$. We claim that $|A|=1$. Suppose for contradiction that $|A|>1$ and there is another vertex $z\in V(A)$. Then $vz$ must be a non-packed non-edge since $v$ is already incident with two packed $P_3$s. Thus Reduction Rule~\ref{RR3} can be applied, a contradiction. It follows that $|A|=1$. This concludes the proof for the lemma.
\end{proof}

\cref{noC3} now suffices to determine a solution around proto-clusters of size three.
See \cref{fig:ReductionRule6} for an illustration of the following \cref{RR6}.

\begin{reduction} \label{RR6}
  After applying Reduction Rules \ref{RR1} to \ref{RRIsClique} exhaustively, if there is a proto-cluster $C$ of size~$3$, a proto-cluster $B$ of size $1$ and a proto-cluster $A$ of size $1$ such that $C$ is not isolated from $B$, and a type-$\gamma$ $P_3$ connects $C$ and $A$ via~$B$, then delete the packed edge between $A$ and $B$, insert an edge to the packed non-edge between $C$ and $B$, and remove the corresponding $P_3$s from $\mathcal{H}$.
\end{reduction}

\begin{figure}[t]
\begin{center}
\begin{tikzpicture}[scale=0.8]

		\node [style=nodestyle2,label={[label distance=-0.5mm]90:$u_{2}$}] (0) at (-7, 1) {};
		\node [style=nodestyle2,label={[label distance=-0.5mm]-90:$u_{3}$}] (1) at (-6, -0) {};
		\node [style=nodestyle2,label={[label distance=-0.5mm]-90:$u_{1}$}] (2) at (-7, -1) {};
		\node [style=nodestyle2,label={[label distance=-0.5mm]-90:$v$}] (3) at (-4.25, -0.25) {};
		\node [style=nodestyle2,label={[label distance=-0.5mm]-90:$w$}] (4) at (-2.5, 0.75) {};
		\node (5) at (-6.7, -0) {};
        \draw[draw=black] (5) ellipse (28pt and 50pt);
        \draw[draw=black] (4) ellipse (10pt and 15pt);
        \draw[draw=black] (3) ellipse (10pt and 15pt);
		\node  (6) at (-8, -0) {$C$};
		\node  (7) at (-4.25, -1.25) {$B$};
		\node  (8) at (-2.5, -0.25) {$A$};

		\draw [draw=black,thick](2) to (0);
		\draw [draw=black,thick](0) to (1);

		\draw [draw=green,thick] (2) to (1);
		\draw [draw=green,thick] (1) to (3);
		\draw [draw=green,thick,dashed] (2) to (3);

		\draw [draw=red,thick] (3) to (4);
		\draw [draw=red,thick,dashed] (0) to (4);

		\draw [draw=red,thick] (0) to (3);

\end{tikzpicture}
\end{center}
\caption{An example for \cref{noC3} and Reduction Rule \ref{RR6}.} \label{fig:ReductionRule6}
\end{figure}

\begin{lemma} \label{safeRR6}
Reduction Rule~\ref{RR6} is safe.
\end{lemma}

\begin{proof}
Given an instance $(G,{\mathcal{H}})$ of \pCEMT\ satisfying the condition of Reduction Rule~\ref{RR6}, let $u_1$, $u_2$, and $u_3$ be the three vertices of $C$, let $v$ be the vertex of $B$ and $w$ be the vertex of $A$. Without loss of generality, let $u_1u_3v$ and $u_2vw$ be two packed $P_3$s. After applying Reduction Rule~\ref{RR6}, we get an instance $(G',{\mathcal{H}}')$ of \pCEMT. We claim that $(G,{\mathcal{H}})$ is a YES-instance if and only if $(G',{\mathcal{H}}')$ is a YES-instance.

For the soundness, suppose that $(G',{\mathcal{H}}')$ is a YES-instance and $S'$ is a cluster editing set of $G'$ such that $|S'|=|{\mathcal{H}}'|$. Obviously $S=S'\cup \{u_1v,vw\}$ is a solution to $(G,{\mathcal{H}})$.

For the completeness, suppose that $(G,{\mathcal{H}})$ is a YES-instance and $S$ is a cluster editing set of $G$ such that $|S|=|{\mathcal{H}}|$.
If $vw \in S$, then $u_2v$ becomes a non-packed edge between $C$ and $B$ after removing the $P_3$ $u_2vw$ from $\mathcal{H}$.
Thus, in this case we have $u_1v \in S$ as well by Reduction Rule~\ref{RR2}, that is, $\{u_1v,vw\}\subseteq S$.
Then $S'=S\setminus \{u_1v,vw\}$ is a solution to $(G',{\mathcal{H}}')$ because by \cref{noC3}, $C$ and $B$ are isolated from the rest of the graph.

Thus, assume $vw\notin S$ from now on.
Then, either $u_2w\in S$ or $u_2v\in S$.
First, we assume that $u_2w\in S$, and after inserting $u_2w$ and removing $u_2vw$ from $\mathcal{H}$ we get an instance $(G'',{\mathcal{H}}'')$ of \pCEMT.
Observe that since $C$ is a proto-cluster and $u_1u_3$ is packed, $u_2u_3$ is not packed.
Thus, $u_3u_2w$ is a non-packed path in $G''$ and $u_3w$ is a non-packed non-edge.
Thus Reduction Rule~\ref{RR1} can be applied to $(G'', \mathcal{H}'')$ and $(G'', \mathcal{H}'')$ is a NO-instance.
This contradicts the fact that $S$ is a solution to $(G,{\mathcal{H}})$.
Thus, we have $u_2v\in S$.  After deleting $u_2v$ and removing $u_2vw$ from $\mathcal{H}$, $u_2v$ becomes a non-packed non-edge. Thus Reduction Rule~$\ref{RR3}$ can be applied, showing $u_3v\in S$.
By Lemma~\ref{noC3}, $C$ is isolated from any other proto-clusters except $B$, and $B$ is isolated from any other proto-clusters except $A$ and $C$. It follows that in $G\triangle S$, $u_1,u_2$ and $u_3$ form a clique of size $3$ while $v$ and $w$ form a clique of size $2$.
Furthermore, $V(G)\setminus \{u_1,u_2,u_3,v,w\}$ forms a cluster graph in $G\triangle S$. Let $\widehat{S}=(S\setminus \{u_2v,u_3v\})\cup \{vw,u_1v\}$. Obviously $G\triangle \widehat{S}$ is also a cluster graph and $|\widehat{S}|=|{\mathcal{H}}|$. Thus $\widehat{S}$ is also a solution to $(G,{\mathcal{H}})$.
It follows that $\widehat{S}\setminus \{vw,u_1v\}$ is a solution for $(G'', \mathcal{H}'')$.
This completes the proof for the lemma.
\end{proof}

\begin{corollary} \label{cor:noBigClique}
After applying Reduction Rules \ref{RR1} to~\ref{RR6} exhaustively, there are no isolated cliques in the instance and every proto-cluster of the instance is of size at most $2$. Moreover, since the edge in a proto-cluster of size $2$ cannot be a packed edge, every packed $P_3$ in the remaining graph is a type-$\gamma$ $P_3$.
\end{corollary}

\subsubsection{Reducing the size of solution clusters}\label{sec:sol-clust-size}

In the previous section we have successfully removed all proto-clusters of size at least~3.
Suppose that after applying Reduction Rules~\ref{RR1} to~\ref{RR6} exhaustively, we have an instance $(G,{\mathcal{H}})$ of \pCEMT.
Suppose that $S$ is a solution to $(G,{\mathcal{H}})$.
Now we consider the size of the clusters in the cluster graph $G\triangle S$.
We first show that the largest clique in this graph has size at most~6.

\begin{lemma} \label{noClique7}
After applying Reduction Rules \ref{RR1} to \ref{RR6} exhaustively, we have an instance $(G,{\mathcal{H}})$ of \pCEMT. Suppose that $S$ is a solution to $(G,{\mathcal{H}})$. Then there is no clique of size larger than $6$ in $G\triangle S$.
\end{lemma}
\begin{proof}
  Suppose for contradiction that $A$ is a clique of size at least $7$ in $G\triangle S$ and let $u$ be a vertex in $A$. Then there are at least six vertex pairs between $\{u\}$ and $V(A)\setminus \{u\}$, which are either non-packed edges or covered by packed $P_3$s.
  Since $u$ is incident with at most two packed $P_3$s, at most four vertex pairs between $\{u\}$ and $V(A)\setminus \{u\}$ are covered by a packed $P_3$.
  Thus at least two vertex pairs between $\{u\}$ and $V(A)\setminus \{u\}$ are non-packed edges.
  By Corollary~\ref{cor:noBigClique}, every proto-cluster in $G$ is of size at most $2$, a contradiction. This completes the proof for the lemma.
\end{proof}

\begin{figure}[t]
\centering
\begin{tikzpicture}[scale=0.8]
		\node [style=nodestyle1] (0) at (-6, 2) {};
		\node [style=nodestyle1] (1) at (-4, 2) {};
		\node [style=nodestyle1] (2) at (-6.75, 0.25) {};
		\node [style=nodestyle1] (3) at (-6, -1.25) {};
		\node [style=nodestyle1] (4) at (-3.25, 0.25) {};
		\node [style=nodestyle1] (5) at (-4, -1.25) {};

		\node [none,label=above:{\Large $C_1$}] (6) at (-5, 2) {};
		\node [none,label=right:{\Large $C_3$}] (7) at (-3.5, -0.5) {};
		\node [none,label=left:{\Large $C_2$}] (8) at (-6.5, -0.5) {};

		\draw [style=black,very thick](0) to (1);
		\draw [style=black,very thick](2) to (3);
		\draw [style=black,very thick](5) to (4);
		\draw [style=green,thick] (2) to (0);
		\draw [style=green,thick,dashed] (0) to (4);
		\draw [style=green,thick] (2) to (4);
		\draw [style=red,thick] (0) to (3);
		\draw [style=red,thick,dashed] (3) to (5);
		\draw [style=red,thick] (0) to (5);
		\draw [style=blue,thick] (1) to (2);
		\draw [style=blue,thick,dashed] (2) to (5);
		\draw [style=blue,thick] (1) to (5);
		\draw [style=olive,thick,dashed] (1) to (3);
		\draw [style=olive,thick] (3) to (4);
		\draw [style=olive,thick] (1) to (4);
\end{tikzpicture}
\caption{An example of forming a clique of size $6$ in $G\triangle S$. The black edges are non-packed edges. The vertex pairs of the same color which is not black belong to the same packed $P_3$ and the dashed edges represent non-edges. The same rule of notation applies to the following pictures.} \label{fig:formClique6}
\end{figure}

We can now determine more precisely the structure of potential cliques of size~6 in $G \triangle S$. See Fig.~\ref{fig:formClique6} as an example.
\begin{lemma} \label{Clique6}
  Let $(G,{\mathcal{H}})$ be an instance of \pCEMT\ such that the size of every proto-cluster in $G$ is at most $2$.
  Let $S$ be a solution to $(G,{\mathcal{H}})$ and suppose that $A$ is a clique of size exactly $6$ in $G\triangle S$.
  Then the following statements hold:
  \begin{itemize}
  \item The vertices of $V(A)$ belong to three proto-clusters $C_1$, $C_2$, and $C_3$ of size 2 in $G$.
  \item Every vertex pair between $C_1$ and $C_2$, between $C_1$ and $C_3$, and between $C_2$ and $C_3$ is covered by some $P_3$ of $\mathcal{H}$.
  \item Furthermore, $V(C_1)\cup V(C_2)\cup V(C_3)$ forms a connected component in $G$.
  \end{itemize}
\end{lemma}
\begin{proof}
  Suppose for contradiction that $u\in V(A)$ belongs to a proto-cluster of size 1 in $G$.
  Then there are five vertex pairs between $\{u\}$ and $V(A)\setminus \{u\}$, which are covered by packed $P_3$s.
  Since $u$ belongs to at most two packed $P_3$s, at most four vertex pairs between $\{u\}$ and $V(A)\setminus \{u\}$ are covered by a packed $P_3$, a contradiction.

  Next we show that the vertices of $V(A)$ belong to three proto-clusters $C_1$, $C_2$, and $C_3$ of size 2 in $G$; see also Fig~\ref{fig:formClique6}.
  We see that for every vertex $v\in V(A)$, four of the vertex pairs between $\{v\}$ and $V(A)\setminus \{v\}$ are covered by packed $P_3$s and the other one is a non-packed edge.
  Thus every vertex $v\in V(A)$ belongs to two packed $P_3$s.
  It follows that for each $i \in [3]$ the proto-cluster $C_i$ is isolated from any other proto-cluster in $G\setminus (V(C_1) \cup V(C_2) \cup V(C_3))$.
  Note that there are no type-$\alpha$, type-$\beta$, or type-$\delta$ $P_3$s in $\mathcal{H}$ anymore.
  Thus the edges between the proto-clusters in $A$ are covered by type-$\gamma$ $P_3$s.
  Thus, without loss of generality, let $xyz$ be a $P_3$ such that $x\in V(C_1), y\in V(C_2)$ and $z\in V(C_3)$.
  Thus, $V(C_1)\cup V(C_2)\cup V(C_3)$ forms a connected component.
  This completes the proof for the lemma.
\end{proof}

By the reduction rule that solved small connected components it follows that cliques of size~6 cannot exist in $G \triangle S$.

\begin{lemma} \label{noClique6}
  After applying Reduction Rules \ref{RR1} to \ref{RR6} exhaustively, we have an instance $(G,{\mathcal{H}})$ of \pCEMT. Suppose that $S$ is a solution to $(G,{\mathcal{H}})$. Then there is no clique of size exactly $6$ in $G\triangle S$.
\end{lemma}
\begin{proof}
  Suppose for contradiction that $A$ is a clique of size exactly $6$ in $G\triangle S$.
  According to Lemma~\ref{Clique6}, $V(A)$ induces a connected component of size exactly $6$ in the input graph.
  Then Reduction Rule~\ref{RR4} or Reduction Rule~\ref{RRIsClique} can be applied, a contradiction.
  This completes the proof for the lemma.
\end{proof}

Now we consider the structure of potential cliques of size~5 in $G \triangle S$. See Fig.~\ref{formClique5} for examples.

\begin{lemma} \label{Clique5}
  After applying Reduction Rules \ref{RR1} to \ref{RR3} exhaustively, let $(G,{\mathcal{H}})$ be an instance of \pCEMT\ such that the size of every proto-cluster in $G$ is at most $2$ and $S$ is a solution to $(G,{\mathcal{H}})$.
  Suppose that $A$ is a clique of size exactly $5$ in $G\triangle S$.
  Then there are four proto-clusters $C_i$ for $i \in [4]$ such that the following statements hold:
  \begin{itemize}
  \item $|C_1|=|C_4|=1$ and $|C_2|=|C_3|=2$.
  \item The vertices of $A$ belong to the three proto-clusters $C_1$, $C_2$, and $C_3$ or to the three proto-clusters $C_2$, $C_3$, and $C_4$.
  \item Every vertex pair between $C_i$ and $C_j$ ($i,j\in \{1,2,3,4\}, i\neq j$) is covered by a packed $P_3$ except that the vertex pair between $C_1$ and $C_4$ is a non-packed non-edge.
  \item Furthermore, $V(C_1)\cup V(C_2)\cup V(C_3)\cup V(C_4)$ forms a connected component in $G$.
  \end{itemize}
\end{lemma}
\begin{proof}
  Suppose for a contradiction that at least three vertices of $V(A)$ belong to proto-clusters of size~1 in $G$; say $u,v,w\in V(A)$ belong to three distinct proto-clusters of size one, respectively, and two vertices of $V(A)$, say $x,y\in V(A) \setminus \{u, v, w\}$, belong to a proto-cluster of size two or belong to two distinct proto-clusters of size one, respectively.
  It follows that every vertex pair of $\binom{V(A)}{2}$ is either a non-packed edge or covered by some $P_3$ of $\mathcal{H}$.
  Then $uv,wv,xv,yv$ are four vertex pairs that are covered by packed $P_3$s.
  Since $v$ is incident with at most two packed $P_3$s, there are the two following cases:
  (a) We assume that $u,v,x$ belong to a packed $P_3$ and $w,v,y$ belong to another packed $P_3$.
  We omit the symmetric case that $u,v,y$ belong to a packed $P_3$ and $w,v,x$ belong to another packed $P_3$ since the analysis is analogous.
  (b) We assume that $u,v,w$ belong to a packed $P_3$ and $x,v,y$ belong to another packed $P_3$.

  For case (a), $uw,uy$ are also covered by one packed $P_3$ or two distinct packed $P_3$s.
  If $uw$ and $uy$ are covered by one packed $P_3$, then this $P_3$ is not modification disjoint with the packed $P_3$ covering $w,v,y$, a contradiction.
  If $uw$ and $uy$ are covered by two distinct packed $P_3$s, then $u$ is incident with three packed $P_3$s, a contradiction.

  For case (b), $ux,uy$ are also be covered by one packed $P_3$ or two distinct packed $P_3$s.
  If $ux$ and $uy$ are covered by one packed $P_3$, then it is not modification disjoint with the packed $P_3$ covering $x,v,y$, a contradiction.
  If $ux$ and $uy$ are covered by two distinct packed $P_3$s, then $u$ is incident with three packed $P_3$s, a contradiction.
  As all cases lead to a contradiction, it follows that the vertices of $V(A)$ belong to one proto-cluster of size~1 and two proto-clusters of size~2.

  Next we show that the vertices in $V(A)$ belong to three proto-clusters $C_1$, $C_2$, and $C_3$ (or $C_2$, $C_3$, and $C_4$) in $G$ such that $|C_1|=|C_4|=1$ and $|C_2|=|C_3|=2$; see also Case (1) and Case (2) of Fig.~\ref{formClique5}.
  Let $C_1$, $C_2$, $C_3$ be the proto-clusters contained in $A$ and without loss of generality let $|V(C_1)| = 1$.
  Let $V(C_1)=\{x\},V(C_2)=\{u_1,u_2\}$, and $V(C_3)=\{v_1,v_2\}$.
  Without loss of generality, let $x,u_1,v_1$ belong to a packed $P_3$ and $x,u_2,v_2$ belong to another packed $P_3$. Then $u_1v_2$ and $u_2v_1$ must be covered by packed $P_3$s since otherwise Reduction Rule~\ref{RR3} can be applied to $C_2$ and $C_3$.

  For a contradiction, assume that there are two vertices $y_1,y_2$ such that $y_1,u_1,v_2$ belong to one packed~$P_3$ and $y_2,u_2,v_1$ belong to another packed $P_3$.
  Then $y_1u_2,y_1v_1$ are non-packed non-edges since $u_2$ and $v_1$ are each already incident with two packed $P_3$s.
  It then follows that Reduction Rule~\ref{RR3} can be applied, a contradiction.
  It follows that there is a single vertex~$y$ such that $\{y, u_2, v_1\}$ and $\{y, u_1, v_2\}$ are vertex sets of $P_3$s in $\mathcal{H}$.
  Let $C_4$ be the proto-cluster to which $y$ belongs.

  If $|C_4|>1$, then there must be a non-packed non-edge between $C_4$ and $C_2$ and a non-packed non-edge between $C_4$ and $C_3$.
  Thus Reduction Rule~\ref{RR3} can be applied, a contradiction.
  Thus $|C_4|=1$.
  Since $u_1,u_2,v_1,v_2,x,y$ are all incident with two packed $P_3$s, the subgraph induced by $V(C_1)\cup V(C_2)\cup V(C_3)\cup V(C_4)$ is isolated from the other parts of the graph.
  We can view the graph induced by $V(C_1)\cup V(C_2)\cup V(C_3)\cup V(C_4)$ as a complete graph on $6$ vertices with five missing edges.
  Note that the edge between $x$ and $y$ is missing by the condition of this lemma.
  Suppose that $\{u_1,u_2,v_1,v_2,x,y\}$ does not induce a connected component in $G$.
  This is only possible when every edge incident to $x$ (symmetrically, $y$) is missing
  because a cut of a complete graph on $6$ vertices minus one edge is of size at least $4$.
  However, $x$ (symmetrically, $y$) is incident with two packed $P_3$s and thus at most two of the edges incident to $x$ (symmetrically, $y$) are missing, a contradiction.
  It follows that $V(C_1)\cup V(C_2)\cup V(C_3)\cup V(C_4)$ forms a connected component in $G$.
  This completes the proof for the lemma.
\end{proof}

\begin{figure}[t]
\centering
\begin{tikzpicture}[scale=0.6]
\begin{scope}[shift={(0cm,0cm)}]
		\node [style=nodestyle1,label=above:{ $C_1$}] (0) at (-4, 3) {};
		\node [style=nodestyle1] (1) at (-5.75, 1.75) {};
		\node [style=nodestyle1] (2) at (-5.75, -0.25) {};
		\node [style=nodestyle1] (3) at (-2.25, 1.75) {};
		\node [style=nodestyle1] (4) at (-2.25, -0.25) {};
		\node [style=nodestyle1,label=below:{ $C_4$}] (5) at (-4, -1.5) {};

		\node [style=none,label=left:{ $C_2$}] (6) at (-5.5, 0.75) {};
		\node [style=none,label=right:{ $C_3$}] (7) at (-2.5, 0.75) {};
		\node [style=none,label=below:{(1)}] (8) at (-4, -2.5) {};

		\draw [style=black,very thick] (1) to (2);
		\draw [style=black,very thick] (3) to (4);
		\draw [style=green,thick] (0) to (1);
		\draw [style=green,thick] (1) to (3);
		\draw [style=green,thick,dashed] (0) to (3);
		\draw [style=red,thick] (0) to (2);
		\draw [style=red,thick] (2) to (4);
		\draw [style=red,thick,dashed] (0) to (4);
		\draw [style=blue,thick] (1) to (5);
		\draw [style=blue,thick,dashed] (5) to (4);
		\draw [style=blue,thick] (1) to (4);
		\draw [style=olive,thick] (5) to (2);
		\draw [style=olive,thick] (5) to (3);
		\draw [style=olive,thick,dashed] (2) to (3);
\end{scope}

\begin{scope}[shift={(6cm,0cm)}]
		\node [style=nodestyle1,label=above:{ $C_1$}] (0) at (-4, 3) {};
		\node [style=nodestyle1] (1) at (-5.75, 1.75) {};
		\node [style=nodestyle1] (2) at (-5.75, -0.25) {};
		\node [style=nodestyle1] (3) at (-2.25, 1.75) {};
		\node [style=nodestyle1] (4) at (-2.25, -0.25) {};
		\node [style=nodestyle1,label=below:{ $C_4$}] (5) at (-4, -1.5) {};

		\node [style=none,label=left:{ $C_2$}] (6) at (-5.5, 0.75) {};
		\node [style=none,label=right:{ $C_3$}] (7) at (-2.5, 0.75) {};
		\node [style=none,label=below:{(2)}] (8) at (-4, -2.5) {};

		\draw [style=black,very thick] (1) to (2);
		\draw [style=black,very thick] (3) to (4);
		\draw [style=green,thick] (0) to (1);
		\draw [style=green,thick] (1) to (3);
		\draw [style=green,thick,dashed] (0) to (3);
		\draw [style=red,thick] (0) to (2);
		\draw [style=red,thick,dashed] (2) to (4);
		\draw [style=red,thick] (0) to (4);
		\draw [style=blue,thick] (1) to (5);
		\draw [style=blue,thick,dashed] (5) to (4);
		\draw [style=blue,thick] (1) to (4);
		\draw [style=olive,thick,dashed] (5) to (2);
		\draw [style=olive,thick] (5) to (3);
		\draw [style=olive,thick] (2) to (3);
\end{scope}

\begin{scope}[shift={(12cm,0cm)}]
		\node [style=nodestyle1,label=above:{ $C_1$}] (0) at (-4, 3) {};
		\node [style=nodestyle1] (1) at (-5.75, 1.75) {};
		\node [style=nodestyle1] (2) at (-5.75, -0.25) {};
		\node [style=nodestyle1] (3) at (-2.25, 1.75) {};
		\node [style=nodestyle1] (4) at (-2.25, -0.25) {};
		\node [style=nodestyle1,label=below:{ $C_4$}] (5) at (-4, -1.5) {};

		\node [style=none,label=left:{ $C_2$}] (6) at (-5.5, 0.75) {};
		\node [style=none,label=right:{ $C_3$}] (7) at (-2.5, 0.75) {};
		\node [style=none,label=below:{(3)}] (8) at (-4, -2.5) {};

		\draw [style=black,very thick] (1) to (2);
		\draw [style=black,very thick] (3) to (4);
		\draw [style=green,thick] (0) to (1);
		\draw [style=green,thick] (1) to (3);
		\draw [style=green,thick,dashed] (0) to (3);
		\draw [style=red,thick] (0) to (2);
		\draw [style=red,thick,dashed] (2) to (4);
		\draw [style=red,thick] (0) to (4);
		\draw [style=blue,thick,dashed] (1) to (5);
		\draw [style=blue,thick] (5) to (4);
		\draw [style=blue,thick] (1) to (4);
		\draw [style=olive,thick] (5) to (2);
		\draw [style=olive,thick] (5) to (3);
		\draw [style=olive,thick,dashed] (2) to (3);
\end{scope}

\begin{scope}[shift={(18cm,0cm)}]
		\node [style=nodestyle1,label=above:{$C_1$}] (0) at (-4, 3) {};
		\node [style=nodestyle1] (1) at (-5.75, 1.75) {};
		\node [style=nodestyle1] (2) at (-5.75, -0.25) {};
		\node [style=nodestyle1] (3) at (-2.25, 1.75) {};
		\node [style=nodestyle1] (4) at (-2.25, -0.25) {};
		\node [style=nodestyle1,label=below:{$C_4$}] (5) at (-4, -1.5) {};

		\node [style=none,label=left:{$C_2$}] (6) at (-5.5, 0.75) {};
		\node [style=none,label=right:{ $C_3$}] (7) at (-2.5, 0.75) {};
		\node [style=none,label=below:{(4)}] (8) at (-4, -2.5) {};

		\draw [style=black,very thick] (1) to (2);
		\draw [style=black,very thick] (3) to (4);
		\draw [style=green,thick] (0) to (1);
		\draw [style=green,thick,dashed] (1) to (3);
		\draw [style=green,thick] (0) to (3);
		\draw [style=red,thick] (0) to (2);
		\draw [style=red,thick] (2) to (4);
		\draw [style=red,thick,dashed] (0) to (4);
		\draw [style=blue,thick] (1) to (5);
		\draw [style=blue,thick] (5) to (4);
		\draw [style=blue,thick,dashed] (1) to (4);
		\draw [style=olive,thick] (5) to (2);
		\draw [style=olive,thick,dashed] (5) to (3);
		\draw [style=olive,thick] (2) to (3);
\end{scope}
\end{tikzpicture}

\caption{Some examples of Lemma~\ref{Clique5}. In Case (1), $C_1$ is separated from $C_2$ and $C_3$, and $C_2,C_3,C_4$ are merged into a clique of size $5$ in $G\triangle S$. In Case (2), $C_4$ is separated from $C_2$ and $C_3$, and $C_1,C_2,C_3$ are merged into a clique of size $5$ in $G\triangle S$. In Case (3), $C_1,C_2$ are merged into a clique of size $3$ and $C_3,C_4$ are merged into a clique of size $3$ such that these two cliques of size $3$ are separated from each other. In Case (4), the instance is a NO-instance. Case (3) and Case (4) are not touched by Lemma~\ref{Clique5} but they can be handled by Reduction Rule~\ref{RR4} and~\ref{RRIsClique}.} \label{formClique5}
\end{figure}

As for cliques of size 6, the reduction rule that solved small connected components thus took care of cliques of size~5.

\begin{lemma} \label{noClique5}
  Let $(G,{\mathcal{H}})$ be an instance of \pCEMT\ obtained after applying Reduction Rules \ref{RR1} to \ref{RR6} exhaustively.
  Suppose that $S$ is a solution to $(G,{\mathcal{H}})$.
  Then there is no clique of size exactly $5$ in $G\triangle S$.
\end{lemma}
\begin{proof}
  Suppose for contradiction that $A$ is a clique of size exactly $5$ in $G\triangle S$.
  According to Lemma~\ref{Clique5}, $V(A)$ belongs to a connected component of size $6$ in the input graph.
  Then Reduction Rule~\ref{RR4} or Reduction Rule~\ref{RRIsClique} can be applied, a contradiction.
  This completes the proof for the lemma.
\end{proof}

Summarizing, after applying our reduction rules the cliques in $G \triangle S$ have size at most 4.

\subsubsection{Path-like structures}\label{sec:path-like}

Next, we aim to get rid of cliques of size~4.
This will later enable us to reduce the instance of \pCEMT\ to 2-SAT.
To take care of cliques of size~4, we use a similar strategy as for cliques of size 5 or~6.
We first consider the structure of the proto-clusters taking part in the clique and we then devise reduction rules that remove or simplify these proto-clusters.
The structure here is more involved.
In particular, it is in general not true anymore that cliques of size~4 are contained in small connected components.
However, as we will see, these cliques take part in a path-like structure that can either be solved locally, or that behaves analogously to a~$P_4$, see \cref{fig:RR9} later on.
The following lemma formalizes the underlying structure that may contain cliques of size~4.

\begin{figure}[t]
\centering
\begin{tikzpicture}[scale=0.8]

		\node [style=nodestyle1,label=right:{$z_1$}] (0) at (2, 2) {};
		\node [style=nodestyle1,label=right:{$z_2$}] (1) at (2, -0) {};
		\node [style=nodestyle1,label=left:{$x$}] (2) at (-1, 2) {};
		\node [style=nodestyle1,label=left:{$y$}] (3) at (-1, -0) {};
		\node [style=nodestyle1,label=right:{$u$}] (4) at (1, 3.5) {};
		\node [style=nodestyle1,label=left:{$v$}] (5) at (0.5, -1.5) {};

		\draw [style=black,very thick] (0) to (1);
		\draw [style=red,thick] (2) to (3);
		\draw [style=red,thick] (3) to (1);
		\draw [style=red,thick,dashed] (2) to (1);
		\draw [style=green,thick] (4) to (2);
		\draw [style=green,thick] (2) to (0);
		\draw [style=green,thick,dashed] (4) to (0);
		\draw [style=blue,thick] (3) to (5);
		\draw [style=blue,thick,dashed] (5) to (0);
		\draw [style=blue,thick] (0) to (3);

\end{tikzpicture}
\caption{An example of forming a clique of size $4$ in $G\triangle S$.
  Vertices $z_1,z_2$ form a proto-cluster of size $2$ and each vertex of $u,v,x,y$ belongs to a proto-cluster of size $1$.} \label{fig:formClique4}
\end{figure}

\begin{lemma} \label{Clique4}
  After applying Reduction Rules \ref{RR1} to \ref{RR6} exhaustively, let $(G,{\mathcal{H}})$ be an instance of \pCEMT.
  Let $S$ be a solution to $(G,{\mathcal{H}})$.
  Suppose that $A$ is a clique of size $4$ in $G\triangle S$ and $V(A)=\{x,y,z_1,z_2\}$.
  Then the following statements hold:
\begin{enumerate}[label=(\arabic*)]
\item Three vertices of $V(A)$, say $x,y,z_2$, belong to one packed $P_3$ in $G$, and one vertex of $x,y,z_2$, say $z_2$, together with $z_1$ forms a proto-cluster $C_1$ of size 2 in~$G$.\label{en:Clique4-verts}
\item Vertices~$x$ and $y$ form a proto-cluster $C_2$ of size~1 and a proto-cluster $C_3$ of size 1 in $G$, respectively.\label{en:Clique4-singleton-proto-clusters}
\item There are two vertices $u$ and $v$ such that $x,u,z_1$ belong to a packed $P_3$ in $G$ and $y,v,z_1$ belong to another packed $P_3$ in~$G$.\label{en:Clique4-extraneous}
\item Vertices~$u$ and $v$ form a proto-cluster $C_4$ of size~1 and a proto-cluster $C_5$ of size~1 in $G$, respectively.\label{en:Clique4-extraneous2}
\item $u,v,z_2$ cannot belong to the same packed $P_3$. \label{en:Clique4-extraneous3}
\end{enumerate}
\end{lemma}
\begin{proof}
  We first show the part of \cref{en:Clique4-verts,en:Clique4-singleton-proto-clusters} about the partition of $V(A)$ into proto-clusters.
  For contradiction, suppose that $V(A)$ does not belong to one proto-cluster of size~2 and two proto-clusters of size~1 in~$G$.
  Then there are two cases: (i) Two vertices of $V(A)$, say $x_1,x_2$, belong to a proto-cluster $C_2$ of size two and the other two vertices of $V(A)$, say $y_1,y_2$, belong to a proto-cluster $C_3$ of size~2.
  (ii) All four vertices $x_1,x_2,y_1,y_2$ of $V(A)$ belong to four distinct proto-clusters $C_1$, $C_2$, $C_3$, and $C_4$ of size~1, respectively.

  Case (i): Since all vertex pairs between $C_2$ and $C_3$ need to be covered to form a clique of size~4, without loss of generality, assume that there is a vertex $u \notin V(A)$ such that $u$, $x_1$, and $y_1$ belong to a packed~$P_3$.
  Suppose that there is another vertex $u' \notin V(A) \cup \{u\}$ such that $u'$, $x_2$, and $y_2$ belong to a packed $P_3$.
  Since neither $u,x_2,y_1$ nor $u,x_1,y_2$ could belong to a packed $P_3$ ($uy_1$ and $uy_2$ are already covered by the assumed $P_3$s), one of the vertex pairs $ux_2$ and $uy_2$ must be a non-packed non-edge and thus Reduction Rule~\ref{RR3} can be applied, a contradiction. Thus $u,x_2$ and $y_2$ belong to a packed $P_3$. Similarly, we can show that there is another vertex $v$ such that $v,x_1,y_2$ belong to a packed $P_3$ and $v,x_2,y_1$ belong to a packed $P_3$. It follows that each vertex of $\{x_1,x_2,y_1,y_2,u,v\}$ is incident with two packed $P_3$s.
  First we assume that $u$ and $v$ belong to two different proto-clusters, say $C_1$ and $C_4$, respectively.
  If $|C_1|>1$ or $|C_4|>1$, then there is a non-packed non-edge involving $C_1$ or $C_4$ and thus Reduction Rule~\ref{RR3} can be applied.
  Thus $|C_1|=|C_4|=1$.
  It follows that $V(C_1)\cup V(C_2)\cup V(C_3)\cup V(C_4)$ induces a connected component and Reduction Rule~\ref{RR4} can be applied, a contradiction.
  Assume that $u$ and $v$ belong to one proto-cluster, say $C_1$.
  If $|C_1|>2$, then Reduction Rule~\ref{RR3} can be applied for the same reason as above.
  Thus $|C_1|=2$ and $V(C_1)\cup V(C_2)\cup V(C_3)\cup V(C_4)$ induces a connected component.
  It follows that Reduction Rule~\ref{RR4} can be applied to this connected component, a contradiction.
  Therefore, Case (i) does not happen.

  Case (ii): Since the vertex pair between each pair of $C_1$, $C_2$, $C_3$, and $C_4$ needs to be covered to form a clique of size four and each vertex can be in at most two $P_3$s, without loss of generality, assume that $x_1,x_2,y_1$ belong to a packed~$P_3$.
  Pair $x_1y_2$ also needs to be covered by a packed~$P_3$; observe that by modification-disjointness of the packed $P_3$s, the third vertex in this $P_3$ cannot be contained in~$V(A)$.
  Thus, there is another vertex $y_3 \notin V(A)$ such that $x_1,y_2,y_3$ belong to a packed $P_3$.
  The vertex pairs $x_2y_2$ and $y_1y_2$ cannot be covered by one packed $P_3$ since $x_2y_1$ is already covered by a packed $P_3$.
  Thus $x_2y_2$ and $y_1y_2$ need to be covered by two distinct $P_3$s respectively.
  However, then $y_2$ is incident with three packed $P_3$s, a contradiction.
  Therefore Case (ii) does not happen either.
  It follows that $V(A)$ consists of one proto-cluster of size~2 and two proto-clusters of size~1.

  Next we show that the claims on the $P_3$s in \cref{en:Clique4-verts} as well as \cref{en:Clique4-extraneous,en:Clique4-extraneous2} are true.
  Suppose that $A$ is a clique of size $4$ in $G\triangle S$.
  Let $V(A)=\{x,y,z_1,z_2\}$.
  By the analysis above, we get that two vertices of $A$ belong to a proto-cluster of size~2 and the other two vertices of~$A$ belong to two distinct proto-clusters of size~1 respectively.
  Without loss of generality, assume that $z_1,z_2$ form a proto-cluster $C_1$ of size~2 in $G$ while $x$ and $y$ form a proto-cluster $C_2$ of size~1 and a proto-cluster $C_3$ of size~1 in $G$ respectively.
  See Fig.~\ref{fig:formClique4} for an illustration.

  Since there are three vertex pairs, i.e., $\{xy,xz_1,xz_2\}$, between $x$ and $V(A)\setminus \{x\}$, two of the three vertex pairs are covered by one packed~$P_3$.
  Moreover, this $P_3$ cannot contain two vertices of~$C_1$.
  Without loss of generality, let thus $x,y,z_2$ belong to a packed $P_3$.
  Since $xz_1$ is also covered by a $P_3$ and this $P_3$ is modification-disjoint to the one containing $x,y,z_2$, there is another vertex $u \notin V(A)$ such that $x,u,z_1$ belong to a packed~$P_3$ in~$G$.

  Also $yz_1$ needs to be covered by a packed $P_3$, so there is another vertex $v$ such that $y,v,z_1$ belong to a packed $P_3$ ($u$ and $v$ are different as otherwise the $P_3$s induced by $y,v,z_1$ and $x,u,z_1$ are not modification-disjoint). Suppose that $u$ and $v$ belongs to the same proto-cluster of size at least~2.
  By \cref{cor:noBigClique}, this proto-cluster has size exactly~2.
  Since $x$ and $y$ are incident with two packed $P_3$s respectively, $uy$ and $vx$ are two non-packed non-edges. Thus Reduction Rule~\ref{RR3} can be applied to the proto-clusters adjacent to these non-edges, a contradiction. It follows that $u$ and $v$ must belong to two distinct proto-clusters.
  Assume that there is a vertex $u'$ such that $u'$ and $u$ belong to one proto-cluster of size at least two.
  Since $x,z_1$ are already incident with two packed $P_3$s respectively, $u'x$ and $u'z_1$ must be non-packed non-edges.
  Then Reduction Rule~\ref{RR3} can be applied since $ux$ or $uz_1$ is a packed edge. It follows that $u$ belongs to a proto-cluster of size one, say $C_4$. Similarly, we can show that $v$ belongs to a proto-cluster of size one, say~$C_5$.

  Finally we show that Item~\ref{en:Clique4-extraneous3} is true. Suppose for contradiction that $u,v,z_2$ belong to the same proto-cluster.
  Then every vertex of $\{u,v,x,y,z_1,z_2\}$ is incident with two packed $P_3$s.
  It follows that the subgraph induced by $\{u,v,x,y,z_1,z_2\}$ is a connected component in $G$, which can be handled by Reduction Rule~\ref{RR4}.
  Thus $u,v,z_2$ cannot belong to the same packed $P_3$.
  This completes the proof for the lemma.
\end{proof}

\begin{figure}[t]
\centering
\begin{tikzpicture}[scale=0.47]
\begin{scope}[shift={(-1cm,0cm)}]

		\node [style=nodestyle1,label=above:{$z_1$}] (0) at (0, 1.5) {};
		\node [style=nodestyle1,label=below:{$z_2$}] (1) at (0, -1.5) {};
		\node [style=nodestyle1,label=below:{$y$}] (2) at (-2, -0) {};
		\node [style=nodestyle1,label=below:{$x$}] (3) at (2, -0) {};
		\node [style=nodestyle1,label=below:{$v$}] (4) at (-3.75, -0) {};
		\node [style=nodestyle1,label=below:{$u$}] (5) at (3.75, -0) {};
		\node [style=none,label=below:{Item (1)}] (6) at (0, -2) {};

		\draw [style=black,very thick] (0) to (1);
		\draw [style=blue,thick] (0) to (2);
		\draw [style=blue,thick] (2) to (4);
		\draw [style=blue,thick,dashed] (0) to (4);
		\draw [style=red,thick] (2) to (1);
		\draw [style=red,thick] (1) to (3);
		\draw [style=red,thick,dashed] (2) to (3);
		\draw [style=green,thick] (0) to (3);
		\draw [style=green,thick] (3) to (5);
		\draw [style=green,thick,dashed] (0) to (5);

\end{scope}

\begin{scope}[shift={(11cm,0cm)}]
		\node [style=nodestyle1,label=above:{$z_1$}] (0) at (-2, 1.5) {};
		\node [style=nodestyle1,label=below:{$z_2$}] (1) at (-2, -1.5) {};
		\node [style=nodestyle1,label=above right:{$x$}] (2) at (-2, -0) {};
		\node [style=nodestyle1,label=below:{$y$}] (3) at (-4, -0) {};
		\node [style=nodestyle1,label={[label distance=-0.5mm]below:{$u$}}] (4) at (0, -0) {};
		\node [style=nodestyle1,label=below:{$v$}] (5) at (-6, -0) {};
		\node [style=nodestyle1,label=below:{$w$}] (6) at (2, -0) {};
		\node [style=none,label=below:{Item (2)}] (7) at (-2, -2.2) {};

		\draw [style=blue,thick] (0) to (3);
		\draw [style=blue,thick] (5) to (3);
		\draw [style=blue,thick,dashed] (5) to (0);
		\draw [style=red,thick] (3) to (2);
		\draw [style=red,thick,dashed] (2) to (1);
		\draw [style=red,thick] (3) to (1);
		\draw [style=gray,thick,thick] (4) to (6);
		\draw [style=gray,thick,thick] (4) to (1);
		\draw [style=gray,thick,thick,dashed] (1) to (6);
		\draw [style=green,thick,dashed] (0) to (2);
		\draw [style=green,thick] (2) to (4);
		\draw [style=green,thick] (0) to (4);

		\draw [style=black,very thick,bend right] (0) to (1);
\end{scope}

\begin{scope}[shift={(21cm,0cm)}]
		\node [style=nodestyle1,label=above:{$z_1$}] (0) at (-2, 1.5) {};
		\node [style=nodestyle1,label=below:{$z_2$}] (1) at (-2, -1.5) {};
		\node [style=nodestyle1,label=above right:{$x$}] (2) at (-2, -0) {};
		\node [style=nodestyle1,label=below:{$y$}] (3) at (-4, -0) {};
		\node [style=nodestyle1,label={[label distance=-0.5mm]below:{$u$}}] (4) at (0, -0) {};
		\node [style=nodestyle1,label=below:{$v$}] (5) at (-6, -0) {};
		\node [style=nodestyle1,label=below:{$w$}] (6) at (2, -0) {};
		\node [style=none,label=below:{Item (3)}] (7) at (-2, -2.2) {};

		\draw [style=blue,thick] (0) to (3);
		\draw [style=blue,thick] (5) to (3);
		\draw [style=blue,thick,dashed] (5) to (0);
		\draw [style=red,thick,dashed] (3) to (2);
		\draw [style=red,thick] (2) to (1);
		\draw [style=red,thick] (3) to (1);
		\draw [style=gray,thick,thick] (4) to (6);
		\draw [style=gray,thick,thick] (4) to (1);
		\draw [style=gray,thick,thick,dashed] (1) to (6);
		\draw [style=green,thick,dashed] (0) to (2);
		\draw [style=green,thick] (2) to (4);
		\draw [style=green,thick] (0) to (4);

		\draw [style=black,very thick,bend right] (0) to (1);
\end{scope}

\begin{scope}[shift={(1cm,-7cm)}]
		\node [style=nodestyle1,label=above:{$z_1$}] (0) at (-2, 1.5) {};
		\node [style=nodestyle1,label=below:{$z_2$}] (1) at (-2, -1.5) {};
		\node [style=nodestyle1,label=above right:{$x$}] (2) at (-2, -0) {};
		\node [style=nodestyle1,label=below:{$y$}] (3) at (-4, -0) {};
		\node [style=nodestyle1,label={[label distance=-0.5mm]below:{$u$}}] (4) at (0, -0) {};
		\node [style=nodestyle1,label=below:{$v$}] (5) at (-6, -0) {};
		\node [style=nodestyle1,label=below:{$w$}] (6) at (2, -0) {};
		\node [style=none,label=below:{Item (4)}] (7) at (-2, -2.2) {};

		\draw [style=blue,thick] (0) to (3);
		\draw [style=blue,thick] (5) to (3);
		\draw [style=blue,thick,dashed] (5) to (0);
		\draw [style=red,thick] (3) to (2);
		\draw [style=red,thick,dashed] (2) to (1);
		\draw [style=red,thick] (3) to (1);
		\draw [style=gray,thick,thick] (4) to (6);
		\draw [style=gray,thick,thick] (4) to (1);
		\draw [style=gray,thick,thick,dashed] (1) to (6);
		\draw [style=green,thick] (0) to (2);
		\draw [style=green,thick] (2) to (4);
		\draw [style=green,thick,dashed] (0) to (4);

		\draw [style=black,very thick,bend right] (0) to (1);
\end{scope}

\begin{scope}[shift={(11cm,-7cm)}]
		\node [style=nodestyle1,label=above:{$z_1$}] (0) at (-2, 1.5) {};
		\node [style=nodestyle1,label=below:{$z_2$}] (1) at (-2, -1.5) {};
		\node [style=nodestyle1,label=above right:{$x$}] (2) at (-2, -0) {};
		\node [style=nodestyle1,label=below:{$y$}] (3) at (-4, -0) {};
		\node [style=nodestyle1,label={[label distance=-0.5mm]below:{$u$}}] (4) at (0, -0) {};
		\node [style=nodestyle1,label=below:{$v$}] (5) at (-6, -0) {};
		\node [style=nodestyle1,label=below:{$w$}] (6) at (2, -0) {};
		\node [style=none,label=below:{Item (5)}] (7) at (3, -2.2) {};

		\node [style=nodestyle1,label=below:{$v$}] (8) at (4, -0) {};
		\node [style=nodestyle1,label=below:{$a$}] (9) at (6, -0) {};
		\node [style=nodestyle1,label=below:{$b$}] (10) at (8, -0) {};
		\node [style=nodestyle1,label=below:{$c$}] (11) at (10, -0) {};
		\node [style=nodestyle1,label=below:{$w$}] (12) at (12, -0) {};
		\node [style=none] (13) at (2.5, -0) {};
		\node [style=none] (14) at (3.7, -0) {};

		\draw [style=blue,thick] (0) to (3);
		\draw [style=blue,thick] (5) to (3);
		\draw [style=blue,thick,dashed] (5) to (0);
		\draw [style=red,thick] (3) to (2);
		\draw [style=red,thick] (2) to (1);
		\draw [style=red,thick,dashed] (3) to (1);
		\draw [style=gray,thick,thick] (4) to (6);
		\draw [style=gray,thick,thick] (4) to (1);
		\draw [style=gray,thick,thick,dashed] (1) to (6);
		\draw [style=green,thick] (0) to (2);
		\draw [style=green,thick,dashed] (2) to (4);
		\draw [style=green,thick] (0) to (4);

		\draw [style=black,very thick,bend right] (0) to (1);

		\draw [style=brown,thick] (8) to (9);
		\draw [style=brown,thick] (9) to (10);
		\draw [style=brown,thick,dashed,bend left] (8) to (10);
		\draw [style=cyan,thick] (10) to (11);
		\draw [style=cyan,thick] (11) to (12);
        \draw [style=cyan,thick,dashed,bend left] (10) to (12);
		\draw [->,line width=0.6mm] (13) to (14);
\end{scope}

\end{tikzpicture}
\caption{Examples of Reduction Rule~\ref{RR9}. Vertices $z_1,z_2$ form a proto-cluster of size $2$ and each of the other vertices belongs to a proto-cluster of size $1$.
  Note that in Item~\ref{9I3} the $P_3$ $y, x, z_2$ is not fully specified by the conditions, that is, its packed non-edge could also be between different vertices.} \label{fig:RR9}
\end{figure}

We next leverage the structure observed in \cref{Clique4} in a reduction rule.
Essentially, all the possible ways to realize the structure of \cref{Clique4} result in a situation that can either be solved directly, or can be replaced by a $P_5$ with suitable new packed $P_3$s.

\begin{reduction} \label{RR9}
  After applying Reduction Rules~\ref{RR1} to \ref{RR6} exhaustively, let $C_1, C_2, C_3,C_4$, and $C_5$ be five proto-clusters such that
  \begin{itemize}
  \item $V(C_1)=\{z_1,z_2\}, V(C_2)=\{x\}, V(C_3)=\{y\}, V(C_4)=\{u\}, V(C_5)=\{v\}$,
  \item $x,y,z_2$ belong to a packed $P_3$,
  \item $x,u,z_1$ belong to a packed $P_3$, and
  \item $y,v,z_1$ belong to a packed $P_3$.
  \end{itemize}
  Check which of the following conditions are satisfied and apply the corresponding data reduction.
\\  \\
\noindent If $uz_2$ and $vz_2$ are non-packed non-edges, then

\begin{enumerate}[label=(\arabic*)]
    \item \label{9I1} delete the edges $vy$ and $ux$, insert an edge to the packed non-edge of the $P_3$ which covers $xy$ and remove the corresponding packed $P_3$s from $\mathcal{H}$.
\end{enumerate}
Otherwise, if there is a vertex $w$ such that $u,w,z_2$ belong to a packed $P_3$, then do reductions according to the following cases:
\begin{enumerate}[label=(\arabic*)]
    \setcounter{enumi}{1}
    \item \label{9I2} $xz_1$ and $xz_2$ are packed non-edges: Return NO.
    \item \label{9I3} $xz_1$ is a packed non-edge and $xz_2$ is a packed edge: Delete the edge $uw$, delete the edges between $y$ and $\{x,z_1,z_2\}$, and add an edge to the non-edge $xz_1$.
      Remove the corresponding packed $P_3$s from $\mathcal{H}$.
    \item \label{9I4} $xz_1$ is a packed edge and $xz_2$ is a packed non-edge:
      Delete the edge $vy$, delete the edges between $u$ and $\{x,z_1,z_2\}$, and add an edge to the non-edge $xz_2$.
      Remove the corresponding packed $P_3$s from $\mathcal{H}$.
    \item \label{9I5} $xz_1$ and $xz_2$ are packed edges: Replace the subgraph induced by $\{u,v,w,x,y,z_1,z_2\}$ with two $P_3$s $vab$ and $bcw$.
      Remove the four packed $P_3$s incident with one vertex of $\{u,x,y,z_1,z_2\}$ from $\mathcal{H}$, and add $vab$ and $bcw$ to $\mathcal{H}$.
\end{enumerate}
Otherwise, if there is a vertex $w'$ such that $v,w',z_2$ belong to a packed $P_3$, then do reductions according to the following cases:
\begin{enumerate}[label=(\arabic*)]
    \setcounter{enumi}{5}
    \item \label{9I6} $yz_1$ and $yz_2$ are packed non-edges: Return NO.
    \item \label{9I7} $yz_1$ is a packed non-edge and $yz_2$ is a packed edge:
      Delete the edge $vw'$, delete the edges between $x$ and $\{y,z_1,z_2\}$, and add an edge to the non-edge $yz_1$.
      Remove the corresponding packed $P_3$s from $\mathcal{H}$.
    \item \label{9I8} $yz_1$ is a packed edge and $yz_2$ is a packed non-edge:
      Delete the edge $ux$, delete the edges between $v$ and $\{y,z_1,z_2\}$, and add an edge to the non-edge $yz_2$.
      Remove the corresponding packed $P_3$s from $\mathcal{H}$.
    \item \label{9I9} $yz_1$ and $yz_2$ are packed edges:
      Replace the subgraph induced by $\{u,v,w',x,y,z_1,z_2\}$ with two $P_3$s $w'ab$ and $bcu$.
      Remove the four packed $P_3$s incident with one vertex of $\{v,x,y,z_1,z_2\}$ from $\mathcal{H}$, and add $w'ab$ and $bcu$ to $\mathcal{H}$.
\end{enumerate}
\end{reduction}

\begin{lemma} \label{safeRR9}
Reduction Rule~\ref{RR9} is safe.
\end{lemma}

\begin{proof}
  Note that Items~\ref{9I2} to~\ref{9I5} are symmetric to the Items~\ref{9I6} to~\ref{9I9} respectively.
  To be more precise, we can relabel the vertices in the subgraph induced by $\{u,v,x,y,z_1,z_2,w\}$ of Item~\ref{9I2} as follows: we exchange the labeling of $x$ and $y$, exchange the labeling of $u$ and $v$, and relabel vertex $w$ as $w'$.
  Then we get a subgraph induced by $\{u,v,x,y,z_1,z_2,w'\}$ satisfying the condition of Item~\ref{9I6}.
  Similarly we show that Item~\ref{9I3} is symmetric to Item~\ref{9I7}, Item~\ref{9I4} is symmetric to Item~\ref{9I8}, and Item~\ref{9I5} is symmetric to Item~\ref{9I9}.
  Thus, the safeness of Items~\ref{9I6} to~\ref{9I9} follows from the safeness of Items~\ref{9I2} to~\ref{9I5}.

  Thus, it is enough to prove the correctness of Items~\ref{9I1} to~\ref{9I5}.
  In the following, to ease arguments we will sometimes successively determine edits that are being made by the solution.
  We then tacitly assume that the packed $P_3$ corresponding to the edit is removed from $\mathcal{H}$, obtaining a new instance of \pCEMT.
  We also sometimes leverage the correctness of the previous reduction rules and apply them to the resulting instances.
  This implies additional edits that can be assumed to be in the solution without loss of generality.

  \proofparagraph{Item~\ref{9I1}.}
  Suppose that $(G,{\mathcal{H}})$ is an instance of \pCEMT\ satisfying the condition of Item~\ref{9I1} of Reduction Rule~\ref{RR9}.
  First, we claim that $vz_1$ and $uz_1$ must be packed non-edges. Suppose for contradiction that $uz_1$ or $vz_1$ is a packed edge.
  Since $vz_2$ and $uz_2$ are non-packed non-edges as in the assumption, Reduction Rule~\ref{RR3} can be applied, a contradiction.
  Let $F$ be the set of vertex pairs edited by Item~\ref{9I1}.
  Observe that $F$ contains exactly one vertex pair of each of the packed $P_3$s incident with one of the vertices of $\{x,y,z_1,z_2\}$.
  After applying the operations of Item~\ref{9I1} we get an instance $(G'=G\triangle F,{\mathcal{H}}')$.
  We claim that $(G,{\mathcal{H}})$ is a YES-instance if and only if $(G',{\mathcal{H}}')$ is a YES-instance.

  First assume that $(G',{\mathcal{H}}')$ has a solution $S'$.
  Observe that $\{x, y, z_1, z_2\}$ is a cluster in $G'$ and no packed $P_3$ is incident with any vertex of this cluster.
  Thus, $S' \cap F = \emptyset$ and, moreover, $S'\cup F$ is a solution to $(G,{\mathcal{H}})$.

  Now assume that $(G,{\mathcal{H}})$ has a solution $S$.
  Then $G\triangle S$ is a cluster graph.
  We claim that $z_2$ is not incident with any other packed $P_3$s except the one covering $xy$.
  Suppose that there is another vertex $z_3$ such that $z_3z_2$ is a packed edge.
  Then $z_3z_1$ is a non-packed non-edge since $z_1$ is already incident with two packed $P_3$s.
  Then Reduction Rule~\ref{RR3} can be applied because $z_1$ and $z_2$ are together in a proto-cluster, a contradiction.
  Thus $z_2$ is not incident with any other packed $P_3$s except the one covering $xy$.

  Let $S_1\subseteq S$ be the set of vertex pairs which are packed edges or packed non-edges in the subgraph of $G$ induced by $\{u,v,x,y,z_1,z_2\}$.
  We claim that $\widehat{S}=(S\setminus S_1)\cup F$ is also a solution to $(G,{\mathcal{H}})$.
  Since $S$ is a solution to $(G,{\mathcal{H}})$, $S_1$ must contain exactly one vertex pair of each of the packed $P_3$s incident with one of the vertices of $\{x,y,z_1,z_2\}$.
  Since $F\cap (S\setminus S_1)=\emptyset$, $S_1\subseteq S$ and $|F|=|S_1|$, we get that $|\widehat{S}|=|S|$.
  Since $G\triangle S$ is a cluster graph, the subgraph of $G\triangle S$ induced by $V(G)\setminus \{x,y,z_1,z_2\}$ is also a cluster graph $\widetilde{G}$.
  $x,y,z_1$ are not incident with any other packed $P_3$s except the ones covering $xy,yz_1,xz_1$ since each of them is incident with two packed $P_3$s in $G$.
  We have shown that $z_2$ is not incident with any other packed $P_3$s except the one covering $xy$ in the previous paragraph.
  Since $x,y$ belong to two proto-clusters of size one and $z_1,z_2$ belong to a proto-cluster of size two, there is no edge between $V(G)\setminus \{u,v,x,y,z_1,z_2\}$ and $\{x,y,z_1,z_2\}$.
  It follows that $G\triangle \widehat{S}$ is a cluster graph consisting of two isolated parts, i.e., $\widetilde{G}$ and the clique of size four formed by $\{x,y,z_1,z_2\}$.
  It follows that $S\setminus S_1$ is a solution to $(G',{\mathcal{H}}')$.
  As a result, Item~\ref{9I1} is safe.

  \proofparagraph{Preparation for Items~\ref{9I2} to~\ref{9I5}.}
  For the proof of the correctness of Items~\ref{9I2} -~\ref{9I5}, we claim that $vz_1$ and $wz_2$ are packed non-edges. Suppose for contradiction that $vz_1$ or $wz_2$ is a packed edge. Since $z_1,z_2$ are already incident with two packed $P_3$s, $vz_2$ and $wz_1$ must be non-packed non-edges. Then Reduction Rule~\ref{RR3} can be applied, a contradiction. Thus $vz_1$ and $wz_2$ are packed non-edges.
  Also, we claim that $vz_2$ and $wz_1$ are non-packed non-edges since $z_1,z_2$ are both incident with two packed $P_3$s and there is no proto-cluster of size more than $2$ by Corollary~\ref{cor:noBigClique}.

  \proofparagraph{Item~\ref{9I2}.}
  For a contradiction, suppose that an instance $(G,{\mathcal{H}})$ of \pCEMT\ satisfying the condition of Item~\ref{9I2} of Reduction Rule~\ref{RR9} has a solution $S$.
  Observe that $vx$ is a non-packed non-edge because $x$ is in two packed $P_3$s with other vertices.
  Thus, at least one packed edge of $vy,xy$ belongs to $S$ since otherwise $\{v, x, y\}$ would induce a $P_3$ in $G \triangle S$.

  Suppose that $xy\in S$.
  Then $yz_2$ cannot be removed by $S$ and thus $yz_1 \notin S$ because otherwise $\{y, z_1, z_2\}$ would induce a $P_3$ in $G \triangle S$.
  Furthermore, $vz_1 \notin S$, because otherwise $\{v, z_1, z_2\}$ would induce a $P_3$ in $G \triangle S$.
  Thus, $vy\in S$.
  Consider the instance resulting from making the edits $xy$ an $vy$.
  Then, the subgraph induced by $\{y,z_1,z_2\}$ is a proto-cluster.
  Since $uy$ is a non-packed non-edge and $uz_1$ is a packed edge, Reduction Rule~\ref{RR3} can be applied, which deletes $uz_1$ and makes $ux$ become a non-packed edge.
  Now $u$ and $x$ form a proto-cluster of size two.
  Observe that now $xz_2$ is a non-packed non-edge.
  Thus again by Reduction Rule~\ref{RR3}, $uz_2$ is deleted and $uw$ becomes a non-packed edge.
  Then Reduction Rule~\ref{RR1} can be applied to the proto-cluster formed by $u,w,x$, a contradiction to $S$ being a solution.

  Suppose that $vy\in S$.
  After making this modification, $yz_1$ becomes a non-packed edge, so $yz_2\notin S$ since otherwise Reduction Rule~\ref{RR1} can be applied.
  If $xz_2\in S$, then $xz_1\in S$ since otherwise $\{x, z_1, z_2\}$ would induce a $P_3$.
  After making these two modifications, $uz_1$ and $ux$ become non-packed edges.
  Since $uy$ is a non-packed non-edge, Reduction Rule~\ref{RR1} can be applied to the proto-cluster formed by $u,x,y,z_1,z_2$.
  Thus $xz_2\notin S$.
  It follows that $xy \in S$ because $yz_2 \notin S$ as argued above.
  After making this modification, by Reduction Rule~\ref{RR3}, $uz_1$ is deleted and $ux$ becomes a non-packed edge.
  Now $u,x$ form a proto-cluster of size two.
  Again by Reduction Rule~\ref{RR3}, $uz_2$ is deleted and $uw$ becomes a non-packed edge.
  Then Reduction Rule~\ref{RR1} can be applied to the proto-cluster formed by $u,w,x$, a contradiction.

  As a result, $(G,{\mathcal{H}})$ is a NO-instance and Item~\ref{9I2} is safe.

  \proofparagraph{Item~\ref{9I3}.}
  Given an instance $(G,{\mathcal{H}})$ of \pCEMT\ satisfying the condition of Item~\ref{9I3} of Reduction Rule~\ref{RR9}, let $F$ be the set of vertex pairs edited by Item~\ref{9I3}.
  Observe that $F$ contains exactly one vertex pair of each of the packed $P_3$s incident with one of the vertices of $\{u,x,y,z_1,z_2\}$.
  After applying the operations of Item~\ref{9I3} we get an instance $(G'=G\triangle F,{\mathcal{H}}')$.
  We claim that $(G,{\mathcal{H}})$ is a YES-instance if and only if $(G',{\mathcal{H}}')$ is a YES-instance.

  First assume that $(G',{\mathcal{H}}')$ has a solution $S'$.
  Observe that $y$ is not in a packed $P_3$ in $\mathcal{H}'$ and thus $\{v, y\}$ forms a cluster in $G' \triangle S'$.
  Similarly, $w$ is not adjacent to any vertex of $u, x, z_1, z_2$.
  Thus $S'\cup F$ is a solution to $(G,{\mathcal{H}})$.

  Now assume that $(G,{\mathcal{H}})$ has a solution $S$.
  We claim that $F\subseteq S$.
  Suppose for contradiction that $xz_1\notin S$.
  Then there are two cases: (1) $ux\in S$ and (2) $uz_1\in S$.

  Case (1): If $ux\in S$, then $uz_1$ becomes a non-packed edge and $xz_1$ becomes a non-packed non-edge.
  It follows that $uz_2\notin S$ and $xz_2\in S$ since otherwise Reduction Rule~\ref{RR1} can be applied.
  Since $wz_1$ is a non-packed non-edge, $wz_2\notin S$.
  Thus $uw\in S$.
  We now distinguish whether $xy$ is an edge or a non-edge.
  If $xy$ is a non-edge, then $yz_2$ is an edge.
  Thus, $yz_2$ becomes a non-packed edge after $xz_2$ is deleted.
  Since $uy$ is a non-packed non-edge, Reduction Rule~\ref{RR1} can be applied, a contradiction.
  Otherwise, if $xy$ is an edge, then $yz_2$ is a non-edge.
  Then $y$ and $x$ form a cluster of size two.
  Recall that $vz_1$ is a packed non-edge.
  Thus, $yz_1$ is a packed edge and $uy$ is a non-packed non-edge, meaning that Reduction Rule~\ref{RR3} can be applied and $yz_1\in S$.
  Thus $vy$ becomes a non-packed edge.
  Since $vx$ is a non-packed non-edge, Reduction Rule~\ref{RR1} can be applied, a contradiction.

  Case (2): If $uz_1\in S$, then $uz_2\in S$ since otherwise Reduction Rule~\ref{RR1} can be applied.
  Thus $ux$ and $uw$ become non-packed edges after $uz_1$ and $uz_2$ are deleted. Since $xw$ is a non-packed non-edge, Reduction Rule~\ref{RR1} can be applied, a contradiction.

  Since both cases lead to a contradiction it follows that $xz_1\in S$.
  Thus, $xz_2,uz_2\notin S$.
  Thus $uw\in S$ since otherwise Reduction Rule~\ref{RR1} can be applied.

  It remains to show that all edges between $y$ and $\{x, z_1, z_2\}$ are in $S$.
  Now $u,x,z_1,z_2$ belong to one proto-cluster.
  Then by Reduction Rule~\ref{RR3}, indeed $yz_1\in S$.
  Thus $u,x,z_1,z_2$ are in one proto-cluster and $y$ is in a different proto-cluster.
  Since $uy$ is a non-packed non-edge, if $xy$ is an edge, $xy\in S$ since otherwise Reduction Rule~\ref{RR1} can be applied.
  Then by Reduction Rule~\ref{RR3}, $yz_1\in S$.
  Similarly, if $yz_2$ is an edge, $yz_2\in S$ since otherwise Reduction Rule~\ref{RR1} can be applied.
  Then by Reduction Rule~\ref{RR3}, $yz_1\in S$.
  As a result, $F\subseteq S$.

  We claim that $\widehat{S}=S\setminus F$ is a solution to $(G',{\mathcal{H}}')$.
  Since $u,x,y,z_1,z_2$ are already incident with two packed $P_3$s, $\{u,x,y,z_1,z_2\}$ are isolated from $V(G)\setminus \{u,v,w,x,y,z_1,z_2\}$ in $G$.
  It follows that in $G\triangle S$, $v,y$ belong to a clique of size two, $u,x,z_1,z_2$ belong to a clique of size four and $V(G)\setminus \{u,v,x,y,z_1,z_2\}$ induces a cluster graph such that there are no edges between $\{u,v,x,y,z_1,z_2\}$ and $V(G)\setminus \{u,v,x,y,z_1,z_2\}$.
  Thus $G'\triangle \widehat{S}$ is a cluster graph and $|\widehat{S}|=|{\mathcal{H}}'|$.
  As a result, Item~\ref{9I3} is safe.

  \proofparagraph{Item~\ref{9I4}.}
  Note that Item~\ref{9I4} is symmetric to Item~\ref{9I3}. To be more precise, we can relabel the vertices in the subgraph of Item~\ref{9I4} as follows: we exchange the labeling of $z_1$ and $z_2$, exchange the labeling of $u$ and $y$, and exchange the labeling of $w$ and $v$. Then we get a subgraph satisfying the condition of Item~\ref{9I3}.
  Thus, the safeness of Item~\ref{9I4} follows from the safeness of Item~\ref{9I3}.

  \proofparagraph{Item~\ref{9I5}.}
  Let $(G,{\mathcal{H}})$ be an instance of \pCEMT\ satisfying the condition of Item~\ref{9I5} of Reduction Rule~\ref{RR9}.
  Since $xz_1,xz_2$ are packed edges, there are two packed edges between $y$ and $\{x,z_1,z_2\}$; let the set of the two packed edges be $W_y$.
  Also, there are two packed edges between $u$ and $\{x,z_1,z_2\}$; let the set of the two packed edges be $W_u$.
  After applying the operations of Item~\ref{9I5} we get an instance $(G',{\mathcal{H}}')$.
  We claim that $(G,{\mathcal{H}})$ is a YES-instance if and only if $(G',{\mathcal{H}}')$ is a YES-instance.

  For soundness, suppose that $(G',{\mathcal{H}}')$ has a solution $S'$.
  Thus $G'\triangle S'$ is a cluster graph.
  There are nine possible cases for which vertex pairs in the two $P_3$s are contained in $S'$.
  However, only the following five cases are valid as in the other cases it is easy to see that $S'$ is not a cluster editing set:

  \emph{(1) $\{ab,cw\}\subseteq S'$.}
  Since $G'\triangle S'$ is a cluster graph, the subgraph of $G'\triangle S'$ induced by $V(G)\setminus \{b,c\}$ is also a cluster graph.
  Let $S= (S'\setminus \{ab,cw\}) \cup W_y \cup (\{ux,uz_1,uz_2\} \setminus W_u) \cup \{uw\}$.
  Note that, by construction of $G'$, there are no edges or packed non-edges from $a, b, c$ to other vertices except to $v, w$.
  In particular, this means that $\{v, a\}$ is a cluster in $G' \triangle S'$.
  Thus $G\triangle S$ is also a cluster graph and $|S|=|\mathcal{H}|$.
  Thus $(G,{\mathcal{H}})$ is a YES-instance.

  \emph{(2) $\{va,bc\}\subseteq S'$.}
  We can show that $(G,{\mathcal{H}})$ has a solution in a similar way to that of Case (1).

  \emph{(3) $\{vb,bc\}\subseteq S'$.}
  Since vertices $a,b$ and $c$ are not adjacent to any vertex of $V(G')\setminus \{a,b,c,v,w\}$ in $G'$, $\{v,a,b\}$ induces a triangle which is a connected component and $\{c,w\}$ induces a clique of size two which is also a connected component in $G'\triangle S'$.
  Let $S= (S'\setminus \{vb,bc\}) \cup W_u \cup (\{yx,yz_1,yz_2\}\setminus W_y) \cup \{vy\}$.
  It follows that $G\triangle S$ is a cluster graph and $|S|=|\mathcal{H}|$.
  Thus $(G,{\mathcal{H}})$ is a YES-instance.

  \emph{(4) $\{ab,bw\}\subseteq S'$.}
  We can show that $(G,{\mathcal{H}})$ has a solution in a similar way to that of Case (3).

  \emph{(5) $\{ab,bc\}\subseteq S'$.}
  Since $G'\triangle S'$ is a cluster graph, the subgraph of $G'\triangle S'$ induced by $V(G)\setminus \{b,c\}$ is also a cluster graph.
  Let $S= (S'\setminus \{ab,bc\}) \cup \{yx,yz_1,uz_1,uz_2\}$.
  Note that, by construction of $G'$, there are neither edges nor packed non-edges from $a, b, c$ to other vertices except to $v, w$.
  It follows that $\{v,a\}$ and $\{c,w\}$ are two clusters in $G' \triangle S'$.
  Thus $G\triangle S$ is also a cluster graph and $|S|=|\mathcal{H}|$.
  Thus $(G,{\mathcal{H}})$ is a YES-instance.

  For completeness, suppose that $(G,{\mathcal{H}})$ has a solution $S$.
  We can check that there are only three possible cases ((1) $vy\in S$, $uw\notin S$; (2) $uw\in S$, $vy\notin S$ (3) $vy\notin S$, $uw\notin S$).
  Readers can easily check that the the other case in which $vy\in S$, $uw\in S$ is invalid as there is no such cluster editing set $S$.

  \emph{(1) $F_1=W_u \cup (\{yx,yz_1,yz_2\}\setminus W_y) \cup \{vy\}\subseteq S$.}
  Since vertices $u,x,y,z_1,z_2$ are not adjacent to any vertex of $V(G)\setminus \{u,v,w,x,y,z_1,z_2\}$ in $G$, $\{x,y,z_1,z_2\}$ induces a clique of size four which is a connected component and $\{u,w\}$ induces a clique of size two which is also a connected component in $G\triangle S$.
  Let $S'=S\setminus F_1\cup \{va,bc\}$.
  It follows that $G'\triangle S'$ is a cluster graph and $|S'|=|{\mathcal{H}}'|$.
  Thus $(G',{\mathcal{H}}')$ is a YES-instance.

  \emph{(2) $F_2=W_y \cup (\{ux,uz_1,uz_2\}\setminus W_u) \cup \{uw\}\subseteq S$.}
  Since vertices $u,x,y,z_1,z_2$ are not adjacent to any vertex of $V(G)\setminus \{u,v,w,x,y,z_1,z_2\}$ in $G$, $\{u,x,z_1,z_2\}$ induces a clique of size four which is a connected component and $\{v,y\}$ induces a clique of size two which is also a connected component in $G\triangle S$.
  Let $S'=S\setminus F_2\cup \{ab,cw\}$.
  It follows that $G'\triangle S'$ is a cluster graph and $|S'|=|{\mathcal{H}}'|$.
  Thus $(G',{\mathcal{H}}')$ is a YES-instance.

  \emph{(3) $F_3=W_u\cup W_y\subseteq S$.}
  Since vertices $u,x,y,z_1,z_2$ are not adjacent to any vertex of $V(G)\setminus \{u,v,w,x,y,z_1,z_2\}$ in $G$, $\{x,z_1,z_2\}$, $\{v,y\}$ and $\{u,w\}$ are three clusters in $G\triangle S$.
  Let $S'=(S\setminus F_3)\cup \{ab,bc\}$.
  It follows that $G'\triangle S'$ is a cluster graph and $|S'|=|{\mathcal{H}}'|$.
  Thus $(G',{\mathcal{H}}')$ is a YES-instance.
  
  As a result, Item~\ref{9I5} is safe.
  This completes the proof for the lemma.
\end{proof}

After applying Reduction Rule~\ref{RR9}, Reduction Rule~\ref{RRIsClique} can be applied to remove the isolated cliques.

\begin{lemma}\label{noC4afterReduction}
  After applying Reduction Rules~\ref{RR1} to~\ref{RR9} exhaustively, let $(G,{\mathcal{H}})$ be an instance of \pCEMT\ which has a solution $S$.
  Then there is no clique of size at least $4$ in $G\triangle S$.
\end{lemma}
\begin{proof}
  By Lemma~\ref{noClique7},~\ref{noClique6} and~\ref{noClique5}, there is no clique of size at least $5$ in $G\triangle S$.
  Suppose for contradiction that $A$ is a clique of size $4$ in $G\triangle S$. Let $V(A)=\{x,y,z_1,z_2\}$.
  Then by Lemma~\ref{Clique4}, three vertices of $V(A)$, say $x,y,z_2$ belong to one packed $P_3$ in $G$, and one vertex of $x,y,z_2$, say $z_2$, forms with $z_1$ a proto-cluster $C_1$ of size two in $G$. Meanwhile, $x$ and $y$ form a proto-cluster $C_2$ of size one and a proto-cluster $C_3$ of size one in $G$ respectively.
  Moreover, there are two vertices $u$ and $v$ such that $x,u,z_1$ belong to a packed $P_3$ in $G$, $y,v,z_1$ belong to another packed $P_3$ in $G$, and $u$ and $v$ form a proto-cluster $C_4$ of size one and $C_5$ of size one in $G$ respectively.
  There are five cases:

  \emph{(1) $uz_2$ and $vz_2$ are non-packed non-edges.}
  Then Item~\ref{9I1} of Reduction Rule~\ref{RR9} can be applied.

  \emph{(2) $uz_2$ is a packed edge and $vz_2$ is a non-packed non-edge.}
  Then one of Items~\ref{9I2} -~\ref{9I5} can be applied.

  \emph{(3) $uz_2$ is a packed non-edge and $vz_2$ is a non-packed non-edge.}
  By Item~\ref{en:Clique4-extraneous3} of Lemma~\ref{Clique4}, $u,v,z_2$ cannot belong to one packed $P_3$.
  Thus there is another vertex $w$ such that $u,w,z_2$ belong to a packed $P_3$ and $uz_2$ is a packed non-edge.
  Thus $wz_2$ is a packed edge.
  Since $z_1$ is in a proto-cluster of size one and $z_1$ is already incident with two packed $P_3$s, $wz_1$ must be a non-packed non-edge.
  Since $z_1,z_2$ belong to one proto-cluster, Reduction Rule~\ref{RR3} can be applied.

  \emph{(4) $vz_2$ is a packed edge and $uz_2$ is a non-packed non-edge.}
  Then one of Items~\ref{9I6} -~\ref{9I9} of \cref{RR9} can be applied.

  \emph{(5) $vz_2$ is a packed non-edge and $uz_2$ is a non-packed non-edge.}
  By Item~\ref{en:Clique4-extraneous3} of Lemma~\ref{Clique4}, $u,v,z_2$ cannot belong to one packed $P_3$.
  Thus there is another vertex $w'$ such that $v,w',z_2$ belong to a packed $P_3$ and $vz_2$ is a packed non-edge.
  Thus $w'z_2$ is a packed edge.
  Since $z_1$ is in a proto-cluster of size one and $z_1$ is already incident with two packed $P_3$s, $w'z_1$ must be a non-packed non-edge.
  Since $z_1,z_2$ belong to one proto-cluster, reduction Rule~\ref{RR3} can be applied.

  It follows that there is no clique of size $4$ in $G\triangle S$. This completes the proof for the lemma.
\end{proof}

\subsubsection{Reduction to 2-SAT} \label{sec:2sat}

First, we introduce a new problem called \textsc{Cluster Deletion above modification-disjoint $P_{3}$ packing}. The formal definition is as follows:

\defprobMD{\textsc{Cluster Deletion above modification-disjoint $P_{3}$ packing} (\pCDA)}
{A graph $G=(V,E)$, a modification-disjoint packing \pkg\ of induced $P_{3}$s of $G$, and a nonnegative integer~$\ell$.}
{Is there a \emph{cluster deletion set}, i.e., a set of edges $S\subseteq E$ such that $G'=(V,E\setminus S)$ is a disjoint union of cliques, with $|S|-|\mathcal{H}|\leq\ell$?}

Note that in the definition of \pCDA, the $P_3$s of $\mathcal{H}$ are still modification-disjoint although the solution to the problem contains only edge deletions.
Since in this paper we only need to consider the special case of \pCDA\ where $\ell=0$, we use the tuple $(G,{\mathcal{H}})$ to represent an instance of \pCDA\ in which $\ell=0$.

\begin{lemma} \label{ClusterDeletion}
  Given an instance $(G,{\mathcal{H}})$ of \pCEMT, after applying Reduction Rules~\ref{RR1} to~\ref{RR9} exhaustively, we get an instance $(G',{\mathcal{H}}')$ of \pCEMT.
  Then $(G,{\mathcal{H}})$ is a YES-instance of \pCEMT\ if and only if $(G',{\mathcal{H}}')$ is a YES-instance of \pCDA.
\end{lemma}
\begin{proof}
  \proofparagraphf{Soundness.} Assume that $(G',{\mathcal{H}}')$ is a YES-instance of \pCDA\ and $S'$ is a cluster deletion set of size $|{\mathcal{H}}'|$.
  Obviously $S'$ is also a cluster editing set of $G'$. Thus $(G',{\mathcal{H}}')$ is a YES-instance of \pCEMT. It follows that $(G,{\mathcal{H}})$ is a YES-instance of \pCEMT.

  \proofparagraph{Completeness.} Assume that $(G,{\mathcal{H}})$ is a YES-instance of \pCEMT.
  Then $(G',{\mathcal{H}}')$ is a YES-instance of \pCEMT\ and let $S'$ be its solution. By Lemma~\ref{noC4afterReduction}, there is no clique of size at least four in $G'\triangle S'$. By Observation~\ref{obs2}, every non-edge of $S'$ is a packed non-edge. Let $uw\in S'$ be a non-edge of $G'$ which is covered by a $P_3$ $uvw$ of ${\mathcal{H}}'$. Then in $G'\triangle S'$, $\{u,v,w\}$ induces a triangle which is a connected component. It follows that $S'\setminus \{uw\}\cup \{uv\}$ is also a solution to $(G',{\mathcal{H}}')$. Let $S_1\subseteq S'$ be the set of non-edges of $S'$. Then there is a set $S_2$ of packed edges of $G'$ such that $(S'\setminus S_1)\cup S_2$ is a cluster deletion set for $G'$ of size $|{\mathcal{H}}'|$. Thus $(G',{\mathcal{H}}')$ is a YES-instance of \pCDA. This completes the proof for the lemma.
\end{proof}

Given an instance $(G,{\mathcal{H}})$ of \pCEMT, after applying Reduction Rules~\ref{RR1} to~\ref{RR9} exhaustively, we get an instance $(G',{\mathcal{H}}')$ of \pCDA.
Let $E_c\subseteq E(G')$ be the set of edges covered by some $P_3$ of ${\mathcal{H}}'$ and let $\lambda=2|{\mathcal{H}}'|$. We fix an arbitrary ordering of the edges of $E_c$ and label these edges by $e_0,e_1,...,e_{\lambda-1}$ according to this ordering. We construct an instance of \textsc{2-SAT} with $\lambda$ variables $x_0,x_1,...,x_{\lambda-1}$ as follows. First, initialize the \textsc{2-SAT} formula $\Phi=\true$.  For each induced $P_3$ $xyz\in {\mathcal{H}}'$, let $e_i=xy, e_j=yz$ and update $\Phi\leftarrow\Phi\wedge (x_i\vee x_j)\wedge (\neg x_i\vee \neg x_j)$.
For each induced $P_3$ $uvw$ in $G'$ such that $uv$ and $vw$ belong to two distinct $P_3$s of ${\mathcal{H}}'$, respectively, let $uv=e_p$ and $vw=e_q$ and update $\Phi\leftarrow\Phi\wedge (x_p\vee x_q)$.
This completes the construction of the \textsc{2-SAT} instance.

We now show that formula $\Phi$ and $(G,{\mathcal{H}})$ are equivalent instances.

\begin{lemma} \label{2SAT}
Given an instance $(G,{\mathcal{H}})$ of \pCEMT, after applying Reduction Rules~\ref{RR1} -~\ref{RR9} exhaustively, we get an instance $(G',{\mathcal{H}}')$ of \pCDA. We construct a \textsc{2-SAT} formula $\Phi$ as described above. Then $(G,{\mathcal{H}})$ is a YES-instance if and only if $\Phi$ is satisfiable.
\end{lemma}
\begin{proof}

  \proofparagraph{Completeness.}
  Assume that $(G,{\mathcal{H}})$ is a YES-instance. By Lemma~\ref{ClusterDeletion}, $(G',{\mathcal{H}}')$ is a YES-instance of \pCDA\ and let $S'$ be a cluster deletion set for $G'$ of size $|{\mathcal{H}}'|$. Let $\alpha$ be an assignment to $\Phi$ such that $\alpha(x_i)=\true$ if and only if $e_i\in S'$ for $i=0,...,\lambda-1$. We claim that $\alpha$ is a satisfying assignment to $\Phi$. Since $|S'|=|{\mathcal{H}}'|$ and the $P_3$s of ${\mathcal{H}}'$ are modification-disjoint, $S'$ contains exactly one edge of every packed $P_3$ of ${\mathcal{H}}'$. It follows that for every $P_3$ $xyz\in {\mathcal{H}}'$ ($xy=e_i,yz=e_j$), the two clauses $(x_i\vee x_j)$ and $(\neg x_i\vee \neg x_j)$ are satisfied. Since $S'$ is a solution to $(G',{\mathcal{H}}')$, there is no induced $P_3$ in $G'\triangle S'$. Thus for every induced $P_3$ $uvw$ in $G'$ such that $uv$ and $vw$ belong to two distinct packed $P_3$s ($uv=e_p,vw=e_q$) respectively, at least one edge of $\{uv,vw\}$ belongs to $S'$ and the clause $(x_p\vee x_q)$ is satisfied. As a result, $\alpha$ is a satisfying assignment to $\Phi$.

  \proofparagraphf{Soundness.}
  Assume that $\Phi$ is satisfiable and let $\alpha$ be a satisfying assignment to $\Phi$. Let $S'=\{e_i\mid \alpha(x_i)=\true\}$. We will show that $S'$ is a cluster deletion set for $G'$ of size $|{\mathcal{H}}'|$. First we claim that for every induced $P_3$ $xyz\in {\mathcal{H}}'$, exactly one edge of $xy$ and $yz$ belongs to $S'$. Assume that $e_i=xy$ and $e_j=yz$ for some $i,j\in \{0,\ldots ,\lambda-1\}$. Since $(x_i\vee x_j)$ and $(\neg x_i\vee \neg x_j)$ are two clauses of $\Phi$ and $\alpha$ is a satisfying assignment to $\Phi$, either $x_i=\false,x_j=\true$ or $x_i=\true,x_j=\false$ holds. Thus the claim is true and $|S'|=|{\mathcal{H}}'|$.

  It remains to show that $S'$ is indeed a cluster deletion set, that is, there is no induced $P_3$ in $G' \triangle S'$.
  We show this by going over the possibilities of such an induced $P_3$ for whether its edges are packed or not.
  Before that, for every induced $P_3$ $uvw$ in $G'$ such that $uv$ and $vw$ belong to two distinct $P_3$s of ${\mathcal{H}}'$, let $uv=e_p$ and $vw=e_q$ for some $p,q\in \{0,\ldots,\lambda-1\}$. By the construction, $(x_p\vee x_q)$ is a clause of $\Phi$ so it is satisfied by $\alpha$. Thus at least one edge of $uvw$ belongs to $S'$.

  First, by Corollary~\ref{cor:noBigClique}, there is no proto-cluster of size at least three in $G'$. Thus there is no induced $P_3$ $abc$ in $G'\triangle S'$ such that $ab$ and $bc$ are non-packed edges in $G'$.

  Second, we claim that there is no induced $P_3$ $xyz$ in $G'\triangle S'$ such that both $xy$ and $yz$ are packed edges in $G'$.
  Suppose for a contradiction that there is an induced $P_3$ $xyz$ in $G'\triangle S'$ such that both $xy$ and $yz$ are packed edges in $G'$.
  Then $xy$ and $yz$ must be covered by two distinct packed $P_3$s, since otherwise $xy$ or $yz$ belongs to $S'$ by the definition of $S'$.
  We contend that $xz$ must be a packed edge covered by another packed $P_3$ in $G'$, i.e., $xy,yz$ and $xz$ are covered by three distinct packed $P_3$s in $G'$.
  First of all, $xz$ is an edge of $G'$, because otherwise $xyz$ would be an induced $P_3$ in $G'$.
  Then $xy$ or $yz$ would belong to $S'$ by the definition of $S'$, a contradiction.
  If $xz$ is a non-packed edge in $G'$, then $xz$ is an edge in $G'\triangle S'$ since $S'$ can only contain vertex pairs covered by packed $P_3$s.
  However, this contradicts the assumption that $xyz$ is an induced $P_3$ in $G'\triangle S'$.
  Therefore, $xz$ is indeed a packed edge in $G'$.

  By the construction of $S'$, no two of $xy, yz$, and $xz$ are covered by the same packed $P_3$ as otherwise one of the three edges belongs to $S'$.
  Thus $xy,yz$ and $xz$ are covered by three distinct packed $P_3$s in $G'$.
  Suppose that without loss of generality, $xz$ is covered by $uxz\in {\mathcal{H}}'$.
  Note that $ux,xy,yz\notin S'$ as by our assumption, $xyz$ is an induced $P_3$ in $G'\triangle S'$.
  Since $y$ is already incident with two packed $P_3$s, $uy$ is either a non-packed non-edge in $G'$ or a non-packed edge in $G'$.
  If $uy$ is a non-packed non-edge in $G'$, then $uxy$ is an induced $P_3$ in $G'$. Let $ux=e_i$ and $xy=e_j$, then the clause $(x_i\vee x_j)$ of $\Phi$ is not satisfied, a contradiction.
  Thus $uy$ is a non-packed edge.

  By the analysis above, there is a vertex $w$ such that $x,y,w$ belong to a packed $P_3$ and there is a vertex $w'$ such that $y,z,w'$ belong to a packed $P_3$. We have the following subcases:
  (1) the subgraph induced by $\{x,y,z,u,w,w'\}$ is isolated from $G'\setminus \{x,y,z,u,w,w'\}$. Then Reduction Rule~\ref{RR4} can be applied;
  (2) Either $wy$ is a non-packed non-edge and $wu$ is a packed edge, or $w'y$ is a non-packed non-edge and $w'u$ is a packed edge in $G'$. Then Reduction Rule~\ref{RR3} can be applied as $uy$ is a proto-cluster of size $2$ in $G'$ by our analysis above;
  (3) the subcases (1) and (2) do not hold. Then we can check that one of the items of Reduction Rule~\ref{RR9} can be applied (note that which item can be applied depends on the structure of the subgraph we are considering):
  There could be another vertex $a$ such that $a,w,u$ belong to one packed $P_3$ or another vertex $a'$ such that $a',w',u$ belong to one packed $P_3$.
  If no such vertices $a$ and $a'$ exist, then Item~\ref{9I1} of Reduction Rule~\ref{RR9} applies.
  Otherwise, one of the other items applies.
  To see this more clearly, we relabel the vertices as follows: $y \leftarrow z_1, u \leftarrow z_2, w\leftarrow u, z\leftarrow y, x\leftarrow x, w'\leftarrow v, a\leftarrow w, a'\leftarrow w'$.
  Thus, all three subcases above contradict the assumption that no reduction rules can be applied in $G'$.
  Therefore, the claim holds, that is, there is no induced $P_3$ $xyz$ in $G'\triangle S'$ such that both $xy$ and $yz$ are packed edges in $G'$.

  Third and finally, we claim that there is no induced $P_3$ in $G'\triangle S'$ such that one edge of this $P_3$ is a non-packed edge in $G'$ and the other edge is a packed edge in $G'$.
  Suppose for a contradiction that there is such a $P_3$ $uvw$ in $G'\triangle S'$ such that $uv$ is a non-packed edge and $vw$ is a packed edge in $G'$.
  Then there is another vertex $x$ such that $v,w,x$ belong to a packed $P_3$ in $G'$.
  Since Reduction Rule~\ref{RR3} cannot be applied to $(G',{\mathcal{H}}')$, $uw$ must be covered by a packed $P_3$ in $G'$, i.e., there is a vertex $y$ such that $u,w,y$ belong to a packed $P_3$ in $G'$. 
  We contend that at least one of $vy$ and $ux$ are covered by a packed $P_3$. 
  Suppose for contradiction that both $vy$ and $ux$ are non-packed non-edges. 
  Then, if $uy$ is a packed edge, Reduction Rule~\ref{RR3} could be applied.
  Thus we can assume that $uy$ is a packed non-edge.
  Since $uvw$ is an induced $P_3$ in $G'\triangle S'$, $uw,wx\in S$. 
  Then $vwy$ is an induced $P_3$ in $G'$. 
  Assume that $vw=e_p$ and $wy=e_q$. 
  Then the assignment $\alpha$ cannot satisfy $(x_p\vee x_q)$, which is a clause of $\Phi$, contradicting that $\alpha$ is a satisfying assignment to $\Phi$.
  Thus we can assume that there is a vertex $z$ such that $v,y,z$ belong to a packed $P_3$ in $G'$ (the analysis for the case that there is a vertex $z'$ such that $u,x,z'$ belong to a packed $P_3$ in $G'$ is similar).
  
  We have the following subcases:
  (1) the subgraph induced by $\{x,y,z,u,v,w\}$ is isolated from $G' \setminus \{x,y,z,u,v,w\}$. Then Reduction Rule~\ref{RR4} can be applied;
  (2) $vz$ is a non-packed non-edge and $uz$ is a packed edge. Then Reduction Rule~\ref{RR3} can be applied as $uv$ is a proto-cluster of size $2$ in $G'$;
  (3) the subcase (1) and (2) do not hold.
  Then we can check that one of the items of Reduction Rule~\ref{RR9} can be applied (note that which item can be applied depends on the structure of the subgraph we are considering).
  There could be another vertex $a$ such that $a,x,u$ belong to one packed $P_3$ or another vertex $a'$ such that $a',z,u$ belong to one packed $P_3$.
  If no such vertices $a$ and $a'$ exist, then Item~\ref{9I1} of Reduction Rule~\ref{RR9} can be applied.
  Otherwise, one of the other items applies.
  To see more clearly that Reduction Rule~\ref{RR9} applies, we relabel the vertices as follows: $v \leftarrow z_1, u \leftarrow z_2, z\leftarrow v, w\leftarrow x, x\leftarrow u, y\leftarrow y, a \leftarrow w, a'\leftarrow w'$.
  All three subcases above contradict the assumption that no reduction rules can be applied in $G'$.
  It follows that there is no induced $P_3$ in $G'\triangle S'$ such that one edge of this $P_3$ is a non-packed edge in $G'$ and the other edge of this $P_3$ is a packed edge in $G'$. 

  As a result, $S'$ is a solution to the instance $(G',{\mathcal{H}}')$ of \pCDA. By Lemma~\ref{ClusterDeletion}, $(G,{\mathcal{H}})$ is a YES-instance.
  This concludes the proof for the lemma.
\end{proof}

The above lemma shows that there is a polynomial-time algorithm for the special instances of \pCDA\ with $\ell = 0$ that our reduction rules produces.

We can now prove that, without excess edits, \pCEMT\ can be solved in polynomial time.

\begin{thmhalfintegralpolynoexcess}[Restated]
  \thmhalfintegralpolynoexcessstatement
\end{thmhalfintegralpolynoexcess}

\begin{proof}
  By Lemma~\ref{2SAT}, given an instance $(G,{\mathcal{H}})$ of \pCEMT, after applying Reduction Rules~\ref{RR1} to~\ref{RR9} exhaustively, we reduce it to an equivalent instance of \textsc{2-SAT} in polynomial time.
  Then we can decide the \textsc{2-SAT} instance by invoking the algorithm for \textsc{2-SAT}.
  It is well-known that \textsc{2-SAT} can be solved in polynomial time. This completes the proof for the theorem.
\end{proof}

\section{Conclusions}

Unfortunately the lower bound that we have obtained is a major roadblock in designing fixed-parameter algorithms for \pCE\ parameterized above modification-disjoint~\(P_3\)s.
On the positive side, \pCEATlong\ (\pCEAT) admits an XP-algorithm with respect to the number of excess edits.
We have left open whether \pCEAT\ is fixed-parameter tractable.
Towards this, on the one hand the half-integral \(P_3\) packings provide quite strong structure that can be exploited to design several branching rules.
On the other hand, when attacking this question from several angles we discovered large grid-like structures that seemed difficult to overcome in fixed-parameter time, and a corresponding W[1]-hardness result would also not be surprising.

A different future research direction is to deconstruct our hardness reduction by examining which substructures it contains that are seldom in practical data.
Forbidding such substructures may destroy the already somewhat fragile hardness construction, perhaps paving the way for fixed-parameter algorithms.

Finally, it would be interesting to see how modification-disjoint \(P_3\) packings look in practice.
If it is true that only few vertices are in a large number of packed \(P_3\)s and most of them are in a small constant number, then a strategy that combines settling the clustering around the vertices with large number of \(P_3\)s and applying reduction rules from \cref{sec:xp-algorithm} could be efficient.

\section*{Acknowledgments}
  The authors would like to thank the anonymous referees for their valuable comments and helpful suggestions.
  This research is part of a project that has received funding from the European Research Council (ERC) under the European Union's Horizon 2020 research and innovation programme, grant agreement 714704.
  MS also gratefully acknowledges funding by the Alexander von Humboldt Foundation.
  \begin{center}
    \includegraphics[height=2\baselineskip]{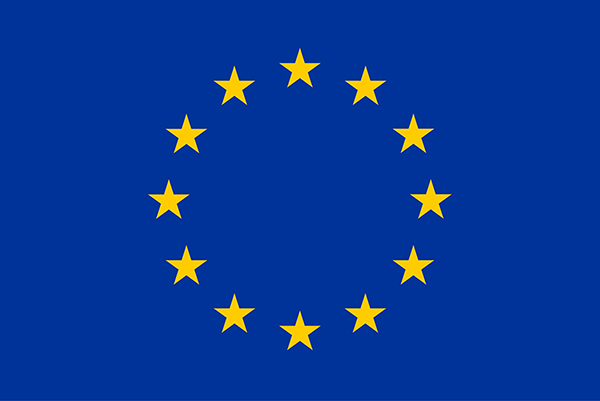}\hspace{1cm}\includegraphics[height=2\baselineskip]{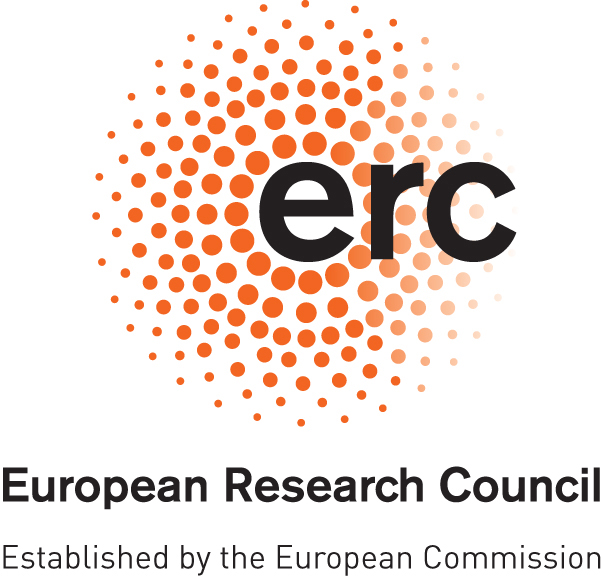}\hspace{1cm}\includegraphics[height=2\baselineskip]{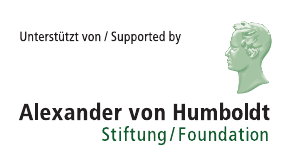}
  \end{center}

\bibliography{ClusterEditing}

\end{document}